\documentclass[12pt]{article} %***
\usepackage[sectionbib]{natbib}
\usepackage{array,epsfig,fancyheadings,rotating}
\usepackage[]{hyperref}  %<-modified by Ivan
\usepackage{sectsty, secdot}
\sectionfont{\fontsize{12}{14pt plus.8pt minus .6pt}\selectfont}
\renewcommand{\theequation}{\thesection\arabic{equation}}
\subsectionfont{\fontsize{12}{14pt plus.8pt minus .6pt}\selectfont}
\usepackage{amsmath}
\usepackage{amssymb}
\usepackage{amsfonts}
\usepackage{multirow}
\usepackage{amsthm}

\newtheorem{lemma}{Lemma}

\newtheorem{proposition}{Proposition}
\theoremstyle{definition}

\allowdisplaybreaks

\usepackage{mathtools}
\usepackage{enumerate}

\usepackage{tikz}
\usepackage{tikz-qtree}
\usetikzlibrary{trees}
\usetikzlibrary{arrows,shapes.arrows,shapes.geometric}
\usetikzlibrary{positioning, quotes, fit}

\DeclareMathOperator{\past}{past}

\DeclareMathOperator{\pa}{pa}

\DeclareMathOperator{\de}{de}

\DeclareMathOperator{\ch}{ch}

\newcommand{\E}{\mathbb{E}} 
\def\ci{\perp\!\!\!\perp}

\newcommand{\gray}{\textcolor{gray}}

\usepackage[normalem]{ulem}

\newcommand{\stkout}[1]{\ifmmode\text{\sout{\ensuremath{#1}}}\else\sout{#1}\fi}

\usepackage{xr}
\begin{document}

%%%%%%%%%%%%%%%%%%%%%%%%%%%%%%%%%%%%%%%%%%%%%%%%%%%%%%%%%%%%%%%%%%%%%%%%%%%%%%%%%%%%%%%%%%%%%%%%%%%%%%%%%%%%%%%%%%%%%%%%%%%%
%%%%%%%%%%%%%%%%%%%%%%%%%%%%%%%%%%%%%%%%%%%%%%%%%%%%%%%%%%%%%%%%%%%%%%%%%%%%%%%%%%%%%%%%%%%%%%%%%%%%%%%%%%%%%%%%%%%%%%%%%%%%

\renewcommand{\baselinestretch}{2}

\markright{ \hbox{\footnotesize\rm Statistica Sinica
%{\footnotesize\bf 24} (201?), 000-000
}\hfill\\[-13pt]
\hbox{\footnotesize\rm
%\href{http://dx.doi.org/10.5705/ss.20??.???}{doi:http://dx.doi.org/10.5705/ss.20??.???}
}\hfill }

\markboth{\hfill{\footnotesize\rm RAZIEH NABI, ROHIT BHATTACHARYA, ILYA SHPITSER, AND JAMES ROBINS} \hfill}
{\hfill {\footnotesize\rm CAUSAL VIEWS OF MISSING DATA MODELS} \hfill}

\renewcommand{\thefootnote}{}
$\ $\par

%%%%%%%%%%%%%%%%%%%%%%%%%%%%%%%%%%%%%%%%%%%%%%%%%%%%%%%%%%%%%%%%%%%%%%%%%%%%%%%%%%%%%%%%%%%%%%%%%%%%%%%%%%%%%%%%%%%%%%%%%%%%

\fontsize{12}{14pt plus.8pt minus .6pt}\selectfont \vspace{0.8pc}
\centerline{\large\bf CAUSAL AND COUNTERFACTUAL VIEWS OF }
\vspace{2pt} 
\centerline{\large\bf MISSING DATA MODELS}
\vspace{.4cm} 
\centerline{Razieh Nabi, Rohit Bhattacharya, Ilya Shpitser, James M.~Robins} 
\vspace{.2cm} 
%\centerline{\it Emory University, Williams College, Johns Hopkins University, Harvard University }
\centerline{\it Emory University, Williams College}
\centerline{\it Johns Hopkins University, Harvard University }
 \vspace{.55cm} \fontsize{9}{11.5pt plus.8pt minus.6pt}\selectfont

%%%%%%%%%%%%%%%%%%%%%%%%%%%%%%%%%%%%%%%%%%%%%%%%%%%%%%%%%%%%%%%%%%%%%%%%%%%%%%%%%%%%%%%%%%%%%%%%%%%%%%%%%%%%%%%%%%%%%%%%%%%%

\begin{quotation}
\noindent {\it Abstract:}
It is often said that the fundamental problem of causal inference is a missing data problem -- the comparison of responses to two hypothetical treatment assignments is made difficult because for every experimental unit only one potential response is observed. In this paper, we consider the implications of the converse view: that missing data problems are a form of causal inference. We make explicit how the missing data problem of recovering the complete data law from the observed law can be viewed as identification of a joint distribution over counterfactual variables corresponding to values had we (possibly contrary to fact) been able to observe them. Drawing analogies with causal inference, we show how identification assumptions in missing data can be encoded in terms of graphical models defined over counterfactual and observed variables. We review recent results in missing data identification from this viewpoint. In doing so, we note interesting similarities and differences between missing data and causal identification theories. 

\vspace{9pt}
\par
\end{quotation}\par

\def\thefigure{\arabic{figure}}
\def\thetable{\arabic{table}}

\renewcommand{\theequation}{\thesection.\arabic{equation}}

\fontsize{12}{14pt plus.8pt minus .6pt}\selectfont

%####################################################
\section{Introduction}
\label{sec:intro} 
%####################################################

Missing data is a common challenge in the analysis of survey, experimental, and observational data, both for the purpose of prediction and for drawing causal conclusions. Complete-case analysis is a popular and simple approach to handling missing data, but it is generally only justified when data entries are missing-completely-at-random (MCAR) \citep{rubin76inference}. When data entries are missing in a way that only depends on observed data values, the data are said to be missing-at-random (MAR) \citep{rubin76inference}. Under MAR assumptions, it is possible to identify target parameters of the underlying data distribution without the need for further parametric assumptions. Moreover, we can estimate parameters identified under MAR via likelihood-based methods such as expectation maximization \citep{dempster77maximum, horton1999maximum, little2002statistical}, multiple imputation \citep{rubin87multiple, schafer1999multiple}, inverse probability weighting \citep{robins1994estimation, li2013weighting}, or semiparametric methods that exploit information about  mechanisms that determine missingness  and are closely related to methods for estimating causal parameters \citep{robins1995analysis, scharfstein99adjusting, robins2001comment,  tsiatis06missing, tchetgen2009commentary}.
% , morikawa2021semiparametric}. 

However, it is often the case that missingness status depends on the underlying values that are themselves censored. This type of missingness is known as missing-not-at-random (MNAR) \citep{rubin76inference}. Without any assumptions, parameters of interest in an MNAR model cannot be %uniquely
identified from the observed data distribution. A common approach  to MNAR problems is to impose sufficient parametric or semiparametric restrictions on the underlying data distribution and missingness selection model, such that they yield identification \citep{wu1988estimation, little2002statistical, ma2003identification, wang2014instrumental, miao2016identifiability, miao2016varieties, sun2018semiparametric}. 
% , morikawa2021semiparametric}. 
Other approaches to handling MNAR mechanisms include conducting sensitivity analysis and obtaining nonparametric bounds \citep{rotnitzky1998semiparametric, robins2000sensitivity, scharfstein2003generalized,  vansteelandt2007estimation, mattei2014identification, moreno2016sensitivity, scharfstein2021semiparametric, duarte2023automated}. 

%A %n alternative
%successful approach to identification proceeds
Identification may be achieved by imposing a set of independence restrictions among variables in the full data distribution that are sufficient to express parameters of interest as functions of the observed data distribution. This approach is also taken in nonparametric identification theory developed in causal inference, where independence restrictions among variables in a full data distribution in a causal model are encoded via directed acyclic graphs (DAGs); this approach has led to sound and complete algorithms for identification of a wide set of causal parameters as functions of the observed data \citep{tian02on,shpitser06id,huang2006do,shpitser18medid,richardson2023nested,bhattacharya2022semiparametric}.  Completeness here means that failure of the algorithm on a particular parameter input implies that the parameter is, in fact, not identified given the set of restrictions encoded by the proposed model. These algorithms generalize many existing results regarding special cases, such as identification by covariate adjustment that relies on the stable unit treatment value assumption and conditional ignorability \citep{rubin76inference}, or the g-computation algorithm that relies on sequential ignorability \citep{robins86new}. 

%DAGs
Directed acyclic graphs have also been adapted to encode independence restrictions in full data distributions in missing data models. Using such representations, many complex scenarios have been described where it is possible to recover target parameters as functions of the observed data distribution %, including settings where data is missing not at random (MNAR) %ones where it was previously believed that nonparametric identification is impossible 
\citep{glymour2006using, daniel2012using, martel2013definition, mohan13missing, thoemmes2014cautious, tian2015missing, shpitser2016consistent, bhattacharya19mid, saadati2019adjustment, nabi20completeness, mohan2021graphical, scharfstein2021markov, nabi2023testability, guo2023sufficient, chen2023causal}. In particular, this line of work has shed light on several classes of  MNAR models %\stkout{(that are thus neither MCAR nor MAR),}
that still permit identification of the target parameter  without relying on any parametric assumptions on the full data distribution. In addition to providing concise representations of statistical models by means of factorizations, graphs also illustrate the causal mechanisms responsible for missingness and provide a natural interpretation of such mechanisms in applied settings. 

It has been noted by many authors that causal inference and missing data are analogous in terminology, theory of identification, and statistical inference. Causal inference has often been phrased as a missing data problem since responses to some treatment interventions are not observed \citep{rubin74potential, Ding2018causal} and missing data is viewed as a form of causal inference where  interventions on missingness indicators can be carried out \citep{robins86new, shpitser15missing, bhattacharya19mid}. At the same time, not much discussion has been devoted to important differences between these frameworks. In this paper, we discuss identifiability of models with MNAR mechanisms and examine new developments in graphical missing data models. We show how identification theory may be understood by viewing missing data models counterfactually, by analogy with causal models, and discuss additional ingredients needed to augment causal identification theories to handle identification in missing data models.  

The paper is organized as follows. 
In Section~\ref{sec:missing_data_models}, we provide a brief overview of classical missing data models and then redefine these models using causal and counterfactual terminology. 
In Section~\ref{sec:causal_inference}, we give an overview of statistical and causal models of DAGs. 
In Section~\ref{sec:missing_data_dag_models}, we formally define missing data DAG models. 
In Section~\ref{sec:miss-ID}, we discuss several unique  techniques for nonparametric identification of complete data distributions in missing data DAGs.
In Section~\ref{sec:diss}, we conclude the paper  by providing a discussion on whether ideas explored in missing data DAG models joined with rank preservation assumptions can be used to obtain novel identification results in the causal inference settings. We also consider missing data DAGs that account for the presence of unmeasured confounders in the supplementary materials. %Appendix~S2. 

%####################################################
\section{Missing Data Models}
\label{sec:missing_data_models}
%####################################################

\subsection{Classical Missing Data Models}
Inferring a parameter of interest in the presence of missing data often involves posing a statistical model that encodes a set of assumptions on the missingness mechanisms. Let $Z=\left( Z_{1}, \dots, Z_{K}\right) ^{T}$ be a vector of $K$ random variables with finite support and probability density $p_Z$. Given a  sample of the random vector $Z$, let $R=\left( R_{1}, \dots, R_{K}\right)^T$ be the vector of binary missingness indicators with $R_{k}=1$ if $Z_{k}$ is observed and $R_{k}=0$ if $Z_k$ is missing. Denote the conditional distribution of $R$ given $Z$ by $p_{R|Z}$ and the joint distribution of $Z$ and $R$ by $p_{(R, Z)}$, and assume $p_Z$, $p_{R|Z}$, and $p_{(R,Z)}$ are contained in statistical models $\mathcal{M}_{Z},$ $\mathcal{M}_{R|Z},$ and $\mathcal{M} = \mathcal{M}_{Z} \otimes \mathcal{M}_{R|Z},$ respectively, where $\otimes$ denotes the direct product of the two statistical models. $\mathcal{M}$ is referred to as a selection model in the missing data literature \citep{little16selection}.
%https://onlinelibrary.wiley.com/doi/abs/10.1002/9781118445112.stat07871
The observed data is often denoted by $O=(R,Z_\text{obs})$, where  $Z_\text{obs}$ is the subvector of $Z$ corresponding to the subvector of $R$ whose entries are $1,$ i.e., $Z_\text{obs} \coloneqq \{Z_k \in Z \text{ s.t. } R_k = 1\}.$ If $\mathcal{M}$ imposes no restrictions on the observed data distribution $p_{O},$ then it is called a \emph{nonparametric saturated} model \citep{robins97non-a}.

\subsection{Counterfactual Views of Classical Missing Data Models}
To motivate causal and counterfactual views of missing data, we first provide a description of causal models. Causal models are often phrased in terms of counterfactual responses to interventions. Random variables of the form $Y^{(a)}$ are used to denote the response of an outcome $Y$ when a treatment $A$ is intervened on and set to value $a$.  The observed (factual) outcomes $Y$ are typically defined as coarsened versions of counterfactual outcomes via the \emph{consistency} property.  For example, for a binary treatment $A$ with values $0$ and $1$, the observed outcome is obtained via the following coarsening mechanism $Y  \coloneqq Y^{(a=1)} \times A + Y^{(a=0)} \times (1 - A)$. That is, the observed outcome $Y$ gives us an imperfect view into the underlying counterfactuals $Y^{(a=1)}$ and $Y^{(a=0)}$: for individuals that received treatment $A=1$, the observed outcome corresponds to the counterfactual response $Y^{(a=1)}$; for those that received  $A=0,$ we gain information regarding the counterfactual response $Y^{(a=0)}$.  In general, we see only one of  the (potentially several for non-binary treatments)  counterfactual responses for each individual. This complicates the task of computing causal parameters, which are phrased in terms of contrasts between different counterfactual responses. This forms the basis of the observation that the fundamental problem of causal inference is a missing data problem.

We can redefine the classical missing data terminology using the terminology of causal models described above.  We may view each missingness indicator $R_k \in R$ as a treatment variable that can be intervened on. Each $Z_k \in Z$ can then be interpreted counterfactually -- a random variable had we, possibly contrary to fact, intervened and set the corresponding missingness indicator $R_k$ to $1$. By analogy with causal models, from here onward, we refer to $Z_k$ as $L_k^{(r_k=1)}$ to highlight the counterfactual nature of the variable. This notation explicitly encodes, in counterfactual language, the assumption implicit in classical missing data models, that the value of $Z_{k}$ remains the same regardless of whether any other $Z_j \in Z$ is observed or missing (or equivalently $R_j$ is $1$ or $0$).
%This assumption is closely related to the ``lack of interference" assumption in causal inference.
We collect all these counterfactual variables into a vector $L^{(r=1)} \coloneqq (L_{1}^{\left(r_{1}=1\right) }, \ldots, L_{K}^{\left( r_{K}=1 \right) }) ^{T}$, and simplify the notation for $L_{k}^{(r_k=1)}, L^{(r=1)}$ via $L_k^{(1)}, L^{(1)}$.

The link between $Z$ and $Z_\text{obs}$ can be viewed as the link between the counterfactual variables $L^{(1)}$ that are of substantive interest, treatment variables $R$, and  factual variables $L$ (a.k.a. proxies) that we observe. Specifically, for any $L_{k}^{(1)}\in {L}^{(1)}$, its corresponding proxy $L_{k}\in L$ is \emph{deterministically} defined as a function of $L_k^{(1)}$ and $R_k$ as follows: $L_{k}=L_{k}^{(1)} $ if $R_{k}=1$ and $L_{k}=``?"$ if $R_{k}=0 $.  This link is closely related to the \emph{consistency} %assumption 
property in causal inference, described above.

The state space of any $L_{k}\in L$ is equal to the state space of the corresponding $L_{k}^{(1)}$ in $L^{(1)}$ joined with the special value ``?". Hence, we denote generic realizations of $L_{k}\in {L}$ and $L_{k}^{(1)}\in {L}^{({1})}$ as $l_{k}$ and $l_{k}^{(1)}$, respectively. We denote realizations of ${L}^{({1})}$ and ${L}$ by ${l}^{({1})}$ and ${l}$, respectively. %For any random variable $V$, we write its density $p_{{V}}({v})$ as $p({v})$, for notational convenience.
By consistency, it is always true that $p(L_{k}^{(1)} \mid R_{k})\vert_{R_{k}=1}=p(L_{k} \mid R_{k})\vert_{R_{k}=1},$ however $p(l_{k}^{(1)} \mid r_{k})|_{r_{k}=1}=p(l_{k} \mid r_{k})|_{r_{k}=1}$ is true if and only if $l_{k}^{(1)}=l_{k}$. 
Here, we use the notation $p(.)|_{R_k=1}$ to mean that $R_k$ in the probability expression $p(.)$ is evaluated at value $1$.  
Our convention is that in any equation or probability expression where $l_{k}^{(1)}$ and $l_{k}$ appear and $r_{k}=1$, $l^{(1)}_{k} = l_{k}$. Thus $p(l_{k}^{(1)} \mid r_{k})|_{r_{k}=1}=p(l_{k} \mid r_{k})|_{r_{k}=1}$ becomes always true. In this way, we can evaluate probability expressions at realizations rather than at the corresponding random variables while still imposing consistency. 

We redefine $Z_\text{obs}$ via $L = L^{(R)} = L^{(r)}\vert_{R = r}.$  With this new notation, the observed data then changes from $O=(R, Z_\text{obs})$ to $O=(R,L)$.  Note, however, that the state space
of $L$ is formed by augmenting the state space of $Z$ with the special value ``?''.

We end this subsection by noting an important difference between causal and missing data models. A causal model posits the existence of two counterfactuals $Y^{(a=0)}$, $Y^{(a=1)}$
for any binary treatment $A$ 
and observed variable $Y$.  In contrast, in a missing data model, there exists only one counterfactual $L_k^{(r_k=1)}$ for any binary missingness indicator $R_k$ and potentially censored variable $L_k$.  As we will show later, this crucial difference allows additional identification theory to be developed specifically for missing data problems that
has no analogue in causal inference identification theory.

\subsection{Identification in Missing Data Models}
\label{subsec:ID_missing_data_models}

A common goal in missing data problems is to determine whether the joint distribution of the complete data ${L}^{\left( {1} \right) } \coloneqq {Z}$, that is $p(l^{\left( {1}\right)}) \coloneqq p(z) \in \mathcal{M}_{{L}^{\left(1\right) }} \coloneqq \mathcal{M}_{{Z}},$ is identified from the observed data $O = (R, L) \coloneqq (R, Z_\text{obs})$, in a model $\mathcal{M=M}_{{R}|{L}^{({1})}} \otimes \mathcal{M}_{{L}^{({1})}}$ defined over the joint distribution of $\left(R, L^{\left(1\right) }\right) \coloneqq \left({R,Z}\right).$ When discussing identification in missing data problems below, we will refer to $p(l^{(1)})$ as the \emph{target law}, $p(r \mid l^{(1)})$ as the \emph{missingness mechanism}, and $p(l^{(1)}, r)$ as the \emph{full law}.  {These distributions may also be extended with a set of auxiliary variables $W$ that are always observed and/or variables $U$ that are never observed. }

Unless explicitly stated otherwise, we assume $\mathcal{M}_{R|L^{({1})}}$ encompasses positive distributions, meaning that for every $p( R=r  \mid {l}^{(1)} ) \in \mathcal{M}_{{R}|{L}^{({1})}}$, $p( R=r  \ \big| \ {l}^{({1})}) >0$ with probability $1$ for all $r \in \left\{ 0,1\right\} ^{K}
$.
{This assumption may be modified to allow pattern restrictions such as monotonicity, where $p( R=r  \ \big| \ {l}^{({1})}) >0$ for any pattern of values of $R$ where, under some ordering of missingness indicators $R_1, \ldots, R_K$, if $R_k = 0$ then $R_{k+1} = 0$ with probability $1$  for every $k \in \{1, \ldots, K-1\}$.}

The model ${\cal M}={\cal M}_{{R}|{L}^{({1})}}\otimes {\cal M}_{{L}^{({1})}}$ is said to be nonparametric just identified  if: (i) ${\cal M}_{{L}^{({1} )}}={\cal M}_{{L}^{({1})}}^{\text{np}}$, where ${\cal M}_{{L}^{({1})}}^{\text{np}} $ is the set of all distributions of $L^{({1})},$ (ii) the distribution of the observed data $L$ is unrestricted, and 
(iii) $p\left( l^{(1)}\right) $ is identified from the observed data distribution $p(l)$ \citep{robins97non-a}. 

Variables in $L$, by definition, contain all information in variables in $R$.  However, in what follows, we will employ observed data distributions $p(l, r)$ which will allow ideas from causal inference to be used for identfiability, with $R$ being viewed as treatment variables. 

A \emph{necessary and sufficient} condition for identification of the target law $p(l^{(1)})$  is that for all $p({r} \mid {l}^{(1)}) \in \mathcal{M}_{{R} | {L}^{({1})}}$, $p({R}={1}\mid {l}^{({1})})>0$  with probability $1$  and $p({R}={1}\mid {l}^{({1})})$ is identified from the observed distribution $p({l, r})$. This follows from the fact that $p\left( R = 1, {l}^{({1})}\right) = p\left( R = 1, {l} \right) $ by consistency and an application of chain rule:
\begin{align}
    p( {l}^{({1})}) ={p({l}, R=1)}/{p(R= 1 \mid {l}^{(1)} )}. 
    \label{eq:target_law}
\end{align}

A \emph{necessary and sufficient} condition for identification of the full law $p(l^{(1)}, r)$ is that for all $p({r} \mid {l}^{(1)}) \in \mathcal{M}_{{R} | {L}^{({1})}},$ $p({R}={1}\mid {l}^{({1})})>0$ with probability $1$  and $p(R = r \mid {l}^{({1})})$ is identified from the observed distribution $p({l, r})$, for any missingness pattern $r \in \{0, 1\}^K$.  This follows from the chain rule of probability, 
\begin{align}
    p( {l}^{(1)}, R = r) =\big\{ p( {l}, R=1 ) /{p(R=1\mid {l}^{(1)})} \big\} \times p(R  = r \mid {l}^{(1)}). 
    \label{eq:full_law}
\end{align}

In this paper, we study a general procedure for analyzing MNAR models proceeds by imposing a set of restrictions on the full data distribution (the target distribution and its missingness mechanism) that are sufficient to yield identification of the parameter of interest. In many models, the restrictions may be represented by a factorization of the full data law with respect to a DAG. Our objective is to identify the target law $p({l}^{(1)})$ from factual data on variables $(R,L)$ where the {missing data} model is represented via a DAG with potentially hidden (never observed) variables. As we will see, restricting attention to such missing data models allows ideas inspired by causal identification theory to be brought to bear. To this end, we first elaborate on the use of DAGs in causal inference.

%####################################################
\section{Directed Acyclic Graphs in Causal Inference}
\label{sec:causal_inference}
%####################################################

Progress in reasoning about counterfactual quantities is achieved by imposing restrictions on a causal model, which consists of sets of joint distributions defined over the factual and counterfactual random variables. Consider the task of identifying the average causal effect, defined as $\E[Y^{(a=1)} - Y^{(a=0)}].$  Given a set of baseline covariates $X$ we may restrict ourselves to  distributions satisfying an independence assumption that $A \ci Y^{(a)} \mid X$ for every value $a$, and positivity of the distribution of $A$ conditioned on $X$, denoted by $p_{A|X}.$ Under the assumptions of this causal model, the average causal effect is identified via the adjustment formula: $\E[\E[Y | A=1,X] - \E[Y | A=0,X]]$.  

Assumptions in a causal model can often be encoded in a more intuitive fashion via DAGs. We formally introduce the statistical and causal models of a DAG below. 

\subsection{Statistical DAG Models}

Let $V = (V_1, \dots, V_K)^T$ be a vector of $K$ random variables with finite support and probability density $p_V.$ We will abbreviate the joint probability $p_V(V=v)$ as simply $p(v).$ Restrictions on the distribution $p_V$ can be encoded via a DAG as follows. Define a DAG ${\cal G}(V)$ consisting of a set of nodes $V$ associated with each random variable $V_i \in V$, and a set of directed edges that form connections between these variables with the restriction that these edges do not form a directed cycle. We will sometimes abbreviate ${\cal G}(V)$ as simply ${\cal G}$ if the set of vertices $V$ is assumed or obvious. For a given DAG $\mathcal{G}$, the statistical model $\mathcal{M}^{\mathcal{G}}$ is the set of distributions that factorize as $	p({v})=\prod_{v_i \in {v}}p\left(v_i \ \big| \ \pa_{\mathcal{G}}(v_i)\right)$, where $\pa_{\mathcal{G}}(v_i)$ is the set of values of variables corresponding to the parents of $V_i$, $\pa_{{\cal G}}(V_i)$, i.e., the set of vertices in $\mathcal{G}$ with directed edges into $V_i$. Distributions in $\mathcal{M}^{\mathcal{G}}$ are said to be Markov relative to $\mathcal{G}$. 

We use the following notation for standard genealogical relations in DAGs. We denote the \emph{children} of a vertex $V_i$ in ${\cal G}$ -- the set of all vertices that have $V_i$ as a parent --  as $\ch_{\cal G}(V_i).$ We denote the \emph{descendants} of $V_i$ -- the set of all vertices $V_j$ such that there exists a directed path from $V_i$ to $V_j$ -- as $\de_{\cal G}(V_i).$ By convention, $\de_{\cal G}(V_i)$ is defined to include $V_i$ itself. 

When ${\cal G}$ is \emph{complete} -- all vertices are pairwise connected via a directed edge --  ${\cal M}^{\cal G}$ imposes no restrictions on $p_V$. When ${\cal G}$ is not complete, it is informative to compare the DAG factorization of $p_V$ to the chain rule factorization to understand how missing edges entail restrictions on the observed distribution. Consider any valid topological ordering $\prec_{\cal G}$ of the variables --  an ordering satisfying the property that whenever $V_{i}\prec _{\cal G}V_{j}$, $V_{i}$ is not a descendant of $V_{j}$ in ${\cal G}$. Then for every variable $V_i \in V$, define $\text{past}_{\mathcal{G}}(V_i)$ to be the set of vertices earlier than $V_i$ under $\prec _{\cal G}$ (we suppress explicit reference to $\prec_{{\cal G}}$ to avoid notational clutter.) Under any variable ordering $\prec_{{\cal G}}$, we have the following equality between the chain factorization and DAG factorization of $p(v):$ $p({v})=\prod_{v_i \in {v}}p(v_i \mid \text{past}_{\mathcal{G}}(v_i)) = \prod_{v_i \in {v}}p(v_i \mid \pa_{\mathcal{G}}(v_i))$. 

Whenever $\pa_{\cal G}(V_i) \subset \past_{\cal G}(V_i)$ (corresponding to missing edges in ${\cal G}$) the above equality implies that $(V_i \ci \past_{\cal G}(V_i) \setminus \pa_{{\cal G}}(V_i) \mid \pa_{{\cal G}}(V_i))$. That is, given a DAG ${\cal G},$ restrictions  in $p_V$ are characterized by the following ordered local Markov property: each variable $V_i$ is independent of its non-parental past given its parents. All restrictions entailed by a DAG ${\cal G}$ are easily read via the d-separation criterion \citep{pearl09causality}.

\subsection{Causal DAG Models}
\label{subsec:causalDAG}

In addition to a statistical DAG model $\mathcal{M}^{\mathcal{G}}$ representing restrictions on the factual variables ${V}$, it is possible to define a causal DAG model associated with ${\cal G}$. Causal DAG models can be generalized from statistical DAG models by equipping them with a special subset ${A}^{\dagger }\subseteq {V}$ referred to as treatment or action variables, {with non-action variables $V \setminus A^{\dag}$ denoted as $Y$.}

	We will index variables $X \subseteq V$ in the model via subscript indices that are consistent with
	a total ordering $\prec$ which must be topological with respect to ${\cal G}$, in the following sense: 	if $X_i$ and $X_j$ are both elements of $X \subseteq V$, and $i<j$, then $X_{i}\prec X_{j}$.
	
	Given any variable $X_i$ in $X \subseteq V$, define $\overline{X}_i$ to be the set consisting of $X_i$ and all elements in $X$ earlier than $X_i$ in the ordering $\prec$.  Similarly, define $\underline{X}_i$ to be the set consisting of $X_i$ and all elements in $X$ later than $X_i$ in the ordering $\prec$.  For any variable $V_i$ and a set $X$, define $X_{j>i}$ to be the $\prec$-smallest element of $X$ larger than $V_i$, and $X_{j<i}$ to be the $\prec$-largest element of $X$ smaller than $V_i$.
	Define $X_{-i}$ to be all elements in $X$ other than $X_i$.
	Finally, for every $V_i \in V$, define
	$\text{past}_{\prec}(V_i)$ to be all variables in $V$ earlier than $V_i$ in the ordering. 
    We will extend this notation in the natural way to values as well, e.g. $\overline{x}_i$ are values of $\overline{X}_i$.  

Although a number of different causal DAG models have been proposed \citep{robins86new,pearl09causality} they all satisfy the properties below, given a particular fixed ordering $\prec$ topological for ${\cal G}$: (Note that any topological ordering yields the same model with different variable indexing.)

\begin{enumerate}[(i)]
	\setlength{\itemsep}{-0.1cm}
	
	\item
	\label{item:existence}
	 {
	{\bf Counterfactual existence.}
	For $V_{i} \in V$, there
	exists a set of counterfactual variables $V_{i}^{\left( {a}^{\dagger }\right)}$
	representing $V_{i}$'s behavior when ${A}^{\dagger}$
	is set to ${a}^{\dagger}$ by external intervention.
	}
	
	\item
	\label{item:no-backwards}
	{
	{\bf No backwards causation.}
	For every $V_i^{(a^{\dag})}$, we have 
    $$V_{i}^{(a^{\dag})}\equiv V_{i}^{(\overline{a}^{\dag}_{j<i}, \ \underline{a}^{\dag}_{j>i})}=V_{i}^{(\overline{a}^{\dag}_{j<i}, \ \underline{a}^{\dag}_{j>i}{'})} \equiv V_{i}^{(\overline{a}^{\dag}_{j<i})},$$ 
    where $\underline{a}^{\dag}_{j>i}{'}$ and $\underline{a}^{\dag}_{j>i}$ are distinct manipulations of variables in $\underline{A}^{\dag}_{j>i}$. 
    For $A_i \in A^{\dag}$, this definition implies that $A_{i}$ does not depend on its own counterfactually set values $a_{i}$.  This follows since the above identity yields $A_{i}^{(a^{\dag})} = A_{i}^{(\overline{a}^{\dag}_{j<i})}$, and $\overline{A}^{\dag}_{j<i}$ does not include $A_i$. 
	}
 
	\item
	\label{item:rec-sub}
	% {
	% {\bf Recursive substitution.}
	% Given any subset $A \subset A^{\dag}$, and values $a$ of $A$, we recursively define $V_i^{(a)}$ as
 %    $V_i^{\left(a, \{ A_j^{(a)} : A_j \in A^{\dag} \setminus A \}\right)}$.
	% In words, we define $V_i^{(a)}$ in terms of the set of $V_i^{(a^{\dag})}$ for all values $a^{\dag}$ as the response $V_i$ had $A$ been manipulated
	% to values $a$ and each element $A_j$ in $A^{\dag} \setminus A$ is manipulated to whatever value it would have had under the intervention that sets $A$ to $a$.  Note that this definition is recursive since $A_j^{(a)}$ must also in general be defined. Property (\ref{item:no-backwards}), and the existence of an ordering $\prec$, due to acyclicity of ${\cal G}$, ensure the induction of the definition is always well-defined in the sense that every counterfactual $V_i^{(a)}$ is definable in terms of counterfactuals assumed to exist by (\ref{item:existence}).
	% }
    {\bf Recursive substitution.} Given any subset $A \subset A^{\dag}$ and values $a$ of $A$, we define $V_i^{(a)}$ as $V_i^{\left(a, \{ A_j^{(a)} : A_j \in A^{\dag} \setminus A \}\right)}$. This means $V_i^{(a)}$ is the response of $V_i$ when $A$ is set to $a$ and $A_j$ in $A^{\dag} \setminus A$ is set to its value under this intervention. This recursive definition ensures each $V_i^{(a)}$ is definable in terms of existing counterfactuals, guaranteed by Property (i) and acyclicity of $\cal G$.

	\item
	\label{item:consistency}
	% {
	% {\bf Consistency.}
	% Every counterfactual variable $V_i^{(a)}$ is linked to the corresponding factual variable $V_i$ by the consistency property, which states that
	% $V_i = \sum_{{a}}\mathbb{I}({A}={a})\times V_{i}^{({a})}=\sum_{\overline{a}_{j}}\mathbb{I}(\overline{A}_{j}=\overline{a}_{j})\times V_{i}^{(\overline{a}_{j})}$;
	% equivalently, $V_{i}=V_{i}^{(a)}$ if ${A}={a}$. Thus, in any equation or probability expression where $v_i$ and $v_i^{({a})}$ appear and ${A}={a}$ holds, we implicitly assume $v_i=v_i^{({a})}$.  This allows us to assert the equality of $p(v_i^{({a})}\mid {a})=p(v_i\mid {a})$ as always true.
	% }
    {\bf Consistency.} Every counterfactual variable $ V_i^{(a)} $ is linked to the factual variable $ V_i $ by the consistency property, which states that $ V_i = \sum_{{a}} \mathbb{I}({A} = {a}) \times V_{i}^{({a})} = \sum_{\overline{a}_{j}}\mathbb{I}(\overline{A}_{j}=\overline{a}_{j})\times V_{i}^{(\overline{a}_{j})}  $, or equivalently, $ V_i = V_i^{(a)} $ if $ {A} = {a} $. Thus, in any context where $ {A} = {a} $, we assume $ V_i = V_i^{(a)} $, allowing us to assert $ p(v_i^{(a)} \mid {a}) = p(v_i \mid {a}) $ as always true.

	\item
	\label{item:positivity}
	{
	{\bf Positivity.}
	For every $A^{\dag}_i \in A^{\dag}$, $p(a^{\dag}_i \mid \text{past}(A^{\dag}_i)) > 0$ for all values of variables in
	$\text{past}_{\prec}(A^{\dag}_i)$.
	}
	
	\item
	\label{item:seq-ignore}
	{
	{\bf Sequential ignorability.}
	For every $A_i \in A \subseteq A^{\dag}$,
%	and the $\prec$-earliest variable $V_{j(i)} \in V$ that occurs after $A_i$ in the order $\prec$,
	we assume $\underline{V}_{j>i}^{({a})} \perp \!\!\!\perp A_{i}\ |\ \text{past}_{\prec}(A_{i})$, for values of $\text{past}_{\prec}(A_{i})$ that are consistent with the treatment assignment $a$.
	}
	
	\item
	\label{item:markov}
	{
	{\bf Markov property.}
	The joint $p(v)$ is Markov relative to ${\cal G}$.
	}

\end{enumerate}
Given any $A \subseteq A^{\dag}$,
assumptions (\ref{item:existence}) %,(\ref{item:ordering}),
and (\ref{item:rec-sub}) allows us to meaningfully discuss counterfactuals $L^{(a)}$, where $L = V \setminus A$.
Furthermore,
under the above causal model, the distribution $p(l^{(a)})$ -- the joint distribution over counterfactuals obtained by setting a subset $A$ of the treatment variables to a set of values $a$ -- can be expressed as a function of the  joint distribution $p(v)$ %over factual variables
as follows.  %Let $L = V \setminus A$.
Assumptions (\ref{item:rec-sub}) and (\ref{item:consistency}), and our notational convention together imply $p({l}^{(a)})\times p( {a} \mid {l}^{(a)}) =p({l}^{({a})},{a})=p({l},{a})$, so
%\stkout{Hence,}

\vspace{-0.5cm}
{\small
\begin{equation}
	%p({y}^{({a})}) = \sum_{l \setminus y}
	p(l^{(a)}) = 
	%\sum_{l \setminus y}
	\frac{ p({l},{a})}{p({a}\mid {l}^{({a})})} =
	%\sum_{l \setminus y}
	\frac{ p(v) }{p({a}\mid {l}^{({a})})}.   \hspace{0.7cm} \textit{(counterfactual g-formula)}
	\label{eqn:ggf}
\end{equation}
}
We refer to (\ref{eqn:ggf}) as the counterfactual g-formula, since counterfactual variables appear in the denominator. Note that this expression is exactly the same as the expression for target law identification from \eqref{eq:target_law} with $a=r$ and $r=1$.

It follows that $p({l}^{({a})})$ is identified from the distribution of the factuals $V$ if and only
if $p({a}\mid {l}^{({a})})$ is identified. The additional assumptions
(\ref{item:consistency}), (\ref{item:positivity}), (\ref{item:seq-ignore}), and (\ref{item:markov})
%(iii),(iv),(v)
imply $p({a}\mid {l}^{({a})})$ is indeed identified as
follows: {\small 
\begin{align*}
p(a\mid l^{(a)})
&=\prod_{a_{k}\in a}p(a_{k}|\overline{a}_{k-1}, \{ l_i^{(a)} : l_i \prec a_{k} \}, \{ l_i^{(a)} : a_{k} \prec l_i \})\\
%&=\prod_{a_{k}\in a}p(a_{k}|\overline{a}_{k-1}, \overline{y}_{k-1}^{(a)},\underline{y}_{k}^{(a)})
&=\prod_{a_{k}\in a}p(a_{k}|\text{past}_{\mathcal{G}}(a_{k}),\{ l_i^{(a)} : a_{k} \prec l_i \})
%&=\prod_{a_{k}\in a}p(a_{k}|\text{past}_{\mathcal{G}}(a_{k}),\underline{{y}}_{k}^{({a})})
=\prod_{a_{k}\in a}p(a_{k}|\text{past}_{\mathcal{G}}(a_{k}))\\
&=\prod_{a_{k}\in a}p(a_{k}|\pa_{\mathcal{G}}(a_{k})). 
\notag
\end{align*}%
}The first equality follows from the chain rule of probability, where for each
conditional $p(a_{k}\mid .)$ we split $l^{{(a)}}$ into the sets
%$\overline{y}_{k-1}^{(a)}$ and $\underline{y}_{k}^{(a)}$.
$\{ l_i^{(a)} : l_i \prec a_{k} \}$ and $\{ l_i^{(a)} : a_{k} \prec l_i \}$.
To see that the second equality follows, note that
no backwards causation (\ref{item:no-backwards}) implies
$\{ l_i^{(a)} : l_i \prec a_{k} \}=\{ l_i^{(\overline{a}_{k-1})} : l_i \prec a_{k} \}$.
%\overline{l}_{k-1}^{(%\overline{a}_{k-1}
%\{ a : a \prec l_{k-1} \})}
%\{ l_i^{(a_j : a_j \prec a_i)} : l_i \prec l_{k-1} \}$,
Conditioning on $\overline{a}_{k-1}$ then implies that
$\{ l_i^{(\overline{a}_{k-1})} : l_i \prec a_{k} \} = \{ l_i : l_i \prec a_{k} \}$.
Finally, if we partition $V$ into $L$ and $A$,
$\past_{\mathcal{G}}(a_{k})=\left( \{ l_i : l_i \prec a_{k} \},\overline{a}_{k-1}\right)$.
The third equality follows from (\ref{item:seq-ignore}). The fourth from (\ref{item:markov}).
Finally, the positivity assumption (\ref{item:positivity}) implies the RHS of the
last equality is a unique function of the distribution of $V$. We conclude that, by term cancellation,
$p(l^{(a)})= {p(l,a)}/{p(a\mid l^{(a)})}$ is identified by the g-formula
\citep{robins86new} as follows: 
%, where the final equality follows from cancellation of terms:
%where 
\begin{align}
%&p_{G}({y}\ \Vert \ {a}) \coloneqq  \\
\label{eqn:g-formula}
&
{
\frac{\prod\limits_{v_{k}\in l\cup a}p(v_{k}\mid \past_{%
\mathcal{G}}(v_{k}))}{\prod\limits_{a_{k}\in a}p\big(a_{k}|\past_{\mathcal{G}}(a_{k})%
\big)}=
}
\frac{\prod\limits_{v_{k}\in l\cup a}p(v_{k}\mid \pa_{%
\mathcal{G}}(v_{k}))}{\prod\limits _{a_{k}\in a}p\big(a_{k}|\pa_{\mathcal{G}}(a_{k})%
\big)}
=\prod\limits_{l_{k}\in {l}}\ p\left( l_{k}\mid \pa_{\mathcal{G}%
}(l_{k})\right) \Bigr|_{A=a},
%\nonumber
\end{align}
a functional we will denote by $p_{\cal G}({l}\ \Vert \ {a})$, following the notation in \citep{Lauritzen96graphical}.
In (\ref{eqn:g-formula}), $p(.)|_{{A}={a}}$ is taken to mean that any free variables within a
probability expression $p(.)$ that intersect the set $A$ are to be evaluated
at corresponding values of ${a}$. 
%Following the notation in \citep{Lauritzen96graphical}, we denote the functional in (\ref{eqn:g-formula}) by
%$p_{\cal G}({l}\ \Vert \ {a})$.
%(Lauritzen 1996)
Note that by our notational convention, it follows
that $p_{\cal G}(l^{\left( a\right) }\ \Vert \ a)=${\large \ }$p_{\cal G}(l\ \Vert \
a)$.
%Since $p(l^{(a)})$ is identified, so is $p(y^{(a)})$ as $\sum_{l \setminus y} p(l^{(a)})$.

The g-formula $p_{\cal G}({l}\ \Vert \ {a})$ in (\ref{eqn:g-formula}) can be
viewed as a truncated factorization in the sense that the terms $p(a_{k}\mid \pa_{\mathcal{G}}(a_{k}))$
that occur in the Markov factorization of the
density $p(l,a)$ are no longer present in the g-formula factorization
representing the intervention distribution. As we will see in Section~\ref{sec:missing_data_dag_models},
the g-formula is also closely connected to identification in missing data problems.

The truncated factorization $p_{\cal G}({l}\ \Vert \ {a})$ is still a
distribution as it provides a mapping from values $a$ of $A$ to normalized
densities over variables in $Y$. We call objects of this type \emph{kernels}%
. A kernel acts in most respects like a conditional distribution. In
particular, given a kernel $p(v\ \Vert \ {w})$ and a subset ${Z}\subseteq {V}
$, conditioning and marginalization are defined in the usual way as
{\small 
\begin{equation}
p({z}\ \Vert \ {w})\coloneqq\sum_{{v}\setminus {z}}\ p(v\ \Vert \ w)\hspace{%
0.5cm}\text{and}\hspace{0.5cm}p({v}\setminus {z}\mid {z}\ \Vert \ {w})%
\coloneqq\frac{p({v}\ \Vert \ {w})}{p({z}\ \Vert \ {w})}.  \label{eqn:kernel}
\end{equation}%

Property (\ref{item:seq-ignore}) above is the critical identifying assumption for obtaining the g-formula: in words, for each
%$k$, there exists a factual past $(\overline{A}_{k-1},\overline{Y}_{k-1})\subset V$ such that treatment
$A_i$ there exists a factual past: $\past_{\prec}(A_i) %$\{ V_i : V_i \prec A_k \}
\subseteq V$ such that treatment
$A_{i}$ is as if randomly assigned conditional on this past -- and hence independent of future counterfactuals.
In the language of epidemiologists, %Proposition 4 states that
conditional on %$(\overline{A}_{k-1},\overline{Y}_{k-1}),$
$\past_{\prec}(A_i)$
the causal effect of $A_{i}$ on %$\underline{{Y}}_{k}$
$\underline{V}_{j>i} = \{ V_{j} \in V : A_i \prec V_j \}$
is unconfounded.

In most observational studies, some of the variables that need to be
included in %$\overline{Y}_{k-1}$
$\past_{\prec}(A_i)$
 (and thus in $V)$ to make $A_{i}$
unconfounded may be unknown to the investigators and/or known but not
measured for financial or logistical reasons. As a consequence $p(l^{({a})})$
%, and thus also $p(y^{(a)})$,
will obviously not be identified from the the factual distribution $p\left( o,a\right)$ of the observed variables $\left( O,A\right),$
where $O\subset L$. Moreover, even the counterfactual distribution over observed outcomes,
\begin{eqnarray*}
p({o}^{({a})}) &=&\sum\limits_{l^{(a)}\backslash o^{(a)}}p({l}^{({a}%
)})=\sum\limits_{l\backslash o}p_{\cal G}({l}\ \Vert \ {a})\equiv p_{\cal G}({o}\
\Vert \ {a}) \\
&=&
{
\sum\limits_{l\backslash o}\frac{\prod_{v_{k}}p(v_{k}\mid \past_{\mathcal{G%
}}(v_{k}))}{\prod_{a_{k}\in a}p\big(a_{k}|\past_{\mathcal{G}}(a_{k})\big)}
=
}
\sum\limits_{l\backslash o}\frac{\prod_{v_{k}}p(v_{k}\mid \pa_{\mathcal{G%
}}(v_{k}))}{\prod_{a_{k}\in a}p\big(a_{k}|\pa_{\mathcal{G}}(a_{k})\big)}, 
\end{eqnarray*}
may not be identified from $p\left( o,a\right)$.  This is because, in
general, property (\ref{item:seq-ignore}) (sequential ignorability) does \emph{not} imply 
%\begin{equation*}
$\underline{{O}}_{j>i}^{({a})}\perp \!\!\!\perp A_{i}\ |\ \text{past-obv}_{%
\mathcal{G}}(A_{i})
\text{ for all }i$,
%\end{equation*}
where $\underline{O}_{j>i} = \{ O_j \in O : A_i \prec O_j \}$, and $\text{past-obv}_{\mathcal{G}}(A_{k})$ %\equiv (\overline{A}_{k-1},\overline{O}_{k-1})$
are the elements in $\text{past}_{\cal G}(A_k)$ that are observed.

As an example illustrating Assumptions~\eqref{item:existence}-\eqref{item:markov}, consider the DAG in Fig.~\ref{fig:iterate}(a). The statistical model $\mathcal{M}^{\mathcal{G}}$ of this DAG is the set of
distributions $p(v)$ that factorize as: $p(u_1) p(u_2 | u_1) p(r_1 | u_2) p(r_2 | l_1, r_1) p(l_1 | r_1, u_1) p(l_2 | r_2, u_2).$ Suppose we choose the treatment set $A^{\dagger }=\{R_{1},R_{2}\}$ and the outcome set ${Y}\coloneqq\{U_{1},U_{2},L_{1},L_{2}\}$. 
Assumption~\eqref{item:existence} entails, for every pair of values $r_1,r_2$, the existence of counterfactuals:
$L_1^{(r_1,r_2)}$, $L_2^{(r_1,r_2)}$, $R_2^{(r_1, r_2)}$, $R_1^{(r_1,r_2)}$, $U_1^{(r_1,r_2)}$, $U_2^{(r_1,r_2)}$.
Assumption~\eqref{item:no-backwards} entails the following identifies:
$R_2^{(r_1,r_2)} = R_2^{(r_1)}$, $R_1^{(r_1,r_2)} = R_1$, $U_1^{(r_1,r_2)} = U_1$, $U_2^{(r_1,r_2)} = U_2$.  Note, in particular, that counterfactuals corresponding to $R_1$ and $R_2$ are not influenced by their own counterfactually set values. 
Assumption~\eqref{item:rec-sub} allows us to define all possible counterfactuals allowed by the model in terms of counterfactuals assumed to exist via (i).  For example, we may define $L_2^{(r_1)}$ as $L_2^{(R_2(r_1), r_1)}$, and $L_2^{(r_2)}$ as $L_2^{(R_1,r_2)}$.
Assumption~\eqref{item:consistency} implies the following identities:
$L_1 = \sum_{r_1} \mathbb{I}(R_1=r_1) L_1^{(r_1)}$, $L_2 = \sum_{r_1,r_2} \mathbb{I}(R_1=r_1,R_2=r_2) L_2^{(r_1,r_2)}$.  For binary $R_1,R_2$, these may be rewritten as $L_1 = L_1^{(1)} R_1 + L_1^{(0)} (1 - R_1)$, and $L_2 = L_2^{(1,1)} R_1 R_2 + L_2^{(1,0)} R_1 (1 - R_2) + L_2^{(0,1)} (1 - R_1) R_2 + L_2^{(0,0)} (1 - R_1) (1 - R_2)$.
Assumption~\eqref{item:positivity} implies that $p(r_1 | u_2)$ and $p(r_2 | l_1, r_1)$ are positive for all values of $u_2$ and $l_1,r_1$, respectively.
Assumption~\eqref{item:seq-ignore} implies the following independence assumptions:
$L_1^{(r_1)}, R_2^{(r_1)}, L_2^{(r_1)} \ci R_1 \mid U_1,U_2$ and
$L_2^{(r_1,r_2)} \ci R_2 \mid R_1=r_1, L_1,U_2,U_1$.
Assumption~\eqref{item:markov} implies the factorization that was described above.  Under the above described causal DAG, the
counterfactual distribution $p(u_{1}^{({r}_{1},{r}_{2})},u_{2}^{({r}_{1},{r}%
_{2})},l_{1}^{({r}_{1},{r}_{2})},l_{2}^{({r}_{1},{r}_{2})})$ is identified
via the g-formula as:
\begin{align}
p_{\cal G}(y\ \Vert \ {r}_{1},{r}_{2})={p(u_{1},u_{2},{r}%
_{1},l_{1},{r}_{2},l_{2})}/\{p({r}_{2}\mid {r}_{1},l_{1})\times p({r}%
_{1}\mid u_{2})\}.
\label{eqn:g-formula-2}
\end{align}
Similar to how the DAG factorization is the factorized representation of a factual density Markov relative to a given causal DAG, the g-formula %, or equivalently the kernel
$p_{\cal G}(y \ \Vert \ a)$, is the factorized representation of a counterfactual density Markov relative to a truncated version of the causal DAG, a.k.a. conditional causal DAG, in which the edges pointing into the treatment variables $A$ that are intervened upon are removed. %Clearly,
%\stkout{\jamie{See (\ref{eqn:g-formula-2}) and the associated example below.} }
%We will use this fact to illustrate how independence restrictions play a role in inductive identifiability derivations via examples.
The conditional causal DAG corresponding to (\ref{eqn:g-formula-2}) is shown in Fig.~\ref{fig:iterate}(b).

%Consider the DAG in Fig.~\ref{fig:iterate}(a). The statistical model $\mathcal{M}^{\cal G}$ of this  DAG  is the set of distributions $p(v)=p(u_1, u_2, r_1, r_2, l_1, l_2)$ that factorize as: 
%{\small
%\begin{align*}
%	p(u_{1}) \times p(u_{2} \mid u_{1}) \times p(r_{1} \mid u_{2}) \times p(l_{1} \mid u_{1},r_{1}) \times  p(r_{2} \mid r_{1},l_{1}) \times  p(l_{2} \mid u_{2},r_{2}).
%\end{align*}%
%}
%Suppose we choose the treatment set $A^\dagger=\{R_1, R_2\}$ and the outcome set ${Y} \coloneqq \{U_{1},U_{2}, L_{1},L_{2}\}$. We are then interested in identifying the counterfactual distribution $p(u_{1}^{({r}_{1},{r}_{2})},u_{2}^{({r}_{1},{r}_{2})}, l_{1}^{({r}_{1},{r}_{2})},l_{2}^{({r}_{1},{r}_{2})})$.  Under the causal DAG model corresponding to the DAG in Fig.~\ref{fig:iterate}(a), this distribution is identified via the g-formula as: $p(u_{1},u_{2}, l_{1},l_{2} \ \Vert \  {r}_{1},{r}_{2}) = {
%	p(u_{1},u_{2}, {r}_{1}, l_{1}, {r}_{2}, l_{2})}/\{p({r}_{2}\mid {r}_{1},l_{1}) \times p({r}
%	_{1}\mid u_{2})\} $
%%{\small
%%	\begin{align}
%%		p(u_{1},u_{2}, l_{1},l_{2} \ \Vert \  {r}_{1},{r}_{2}) &= \frac{
%%			p(u_{1},u_{2}, {r}_{1}, l_{1}, {r}_{2}, l_{2})}{p({r}_{2}\mid {r}_{1},l_{1}) \times p({r}
%%			_{1}\mid u_{2})} 
%%		\label{eqn:g-ex-u} 
%%	\end{align}%
%%}
%The resulting conditional causal DAG is shown in Fig.~\ref{fig:iterate}(b).

\begin{figure}[t] 
	\begin{center}
		\scalebox{0.85}{
			\begin{tikzpicture}[>=stealth, node distance=1.5cm]
				\tikzstyle{format} = [thick, circle, minimum size=1.0mm, inner sep=0pt]
				\tikzstyle{square} = [draw, thick, minimum size=4.5mm, inner sep=3pt]
				
				\begin{scope}[xshift=0cm]
					\path[->, thick]
					node[format] (x11) {${U_1}$}
					node[format, right of=x11, xshift=0.45cm] (x21) {${U_2}$}
					node[format, below of=x11] (r1) {$R_1$}
					node[format, below of=x21] (r2) {$R_2$}
					node[format, below of=r1] (x1) {$L_1$}
					node[format, below of=r2] (x2) {$L_2$}
					
					(x11) edge[blue] (x21) 
					(x21) edge[blue] (r1)
					(r1) edge[blue] (r2)
					(x1) edge[blue] (r2)
					
					(x11) edge[blue, bend right=25] (x1)
					(x21) edge[blue, bend left=25] (x2)
					(r1) edge[blue] (x1)
					(r2) edge[blue] (x2)
					
					node[format, below of=x1, xshift=0.95cm, yshift=0.7cm] (a) {(a)} ;
				\end{scope}

				\begin{scope}[xshift=5.5cm]
					\path[->, thick]
					node[format] (x11) {$U_1$}
					node[format, right of=x11, xshift=0.45cm] (x21) {$U_2$}
					node[square, below of=x11] (r1) {$r_1$}
					node[square, below of=x21] (r2) {$r_2$}
					node[format, below of=r1] (x1) {$L_1$}
					node[format, below of=r2] (x2) {$L_2$}
					
					(x11) edge[blue] (x21) 
					(x11) edge[blue, bend right=25] (x1)
					(x21) edge[blue, bend left=25] (x2)
					(r1) edge[blue] (x1)
					(r2) edge[blue] (x2)
					
					node[format, below of=x1, xshift=0.95cm, yshift=0.7cm] (c) {(b) } ;
				\end{scope}
				
			\end{tikzpicture}
		}
		\vspace{-0.75cm}
		\caption{(a) A DAG where $U_1$ and $U_2$ may be unmeasured;  (b) A conditional DAG illustrating interventions on $R_1$ and $R_2$. }
		\label{fig:iterate}
	\end{center}
\end{figure}
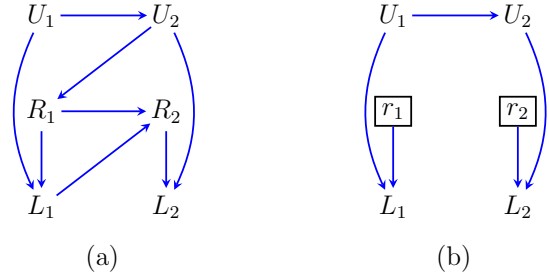

{
Suppose $U_{1},U_{2}$ are unobserved, $O=\left( L_{1},L_{2}\right)$.
The distribution of interest $$p\left( o^{({r}_{1},{r}_{2})}\right) \equiv
p(l_{1}^{({r}_{1},{r}_{2})},l_{2}^{({r}_{1},{r}_{2})})\coloneqq %
p_{\cal G}(l_{1},l_{2}\ \Vert \ r_{1},r_{2})=\sum_{u_{1},u_{2}}p_{\cal G}(l_{1},l_{2},u_{1},u_{2}
\Vert \ r_{1},r_{2})$$ is not a function of the observed data distribution $%
\sum_{u_{1},u_{2}}p(u_{1},u_{2},r_{1},r_{2},l_{1},l_{2})$ %
\citep{shpitser06id}. Intuitively, this is because the g-formula $p(l_{1},l_{2}\ \Vert \ r_{1},r_{2})$ (moving forward we drop the subscript ${\cal G}$, when it is clear from the context what graph we are referring to)
%in (\ref{eqn:g-ex-u}) 
contains the term $p({r}_{1}\mid u_{2})$ in the denominator, {and there is no information on $U_2$
in the observed data.
}
}

% As we will see, in missing data problems every  variable $L_k$ is either always observed, or (if it is missing) corresponds to a single counterfactual version $L_k^{(1)}$ -- the value of $L_k$ had we in fact been able to observe it.  This allows additional identification theory to be developed specifically for missing data problems that have no analogue in causal inference identification theory. We now formalize missing data as a causal inference problem before discussing identification in graphical models of missing data.

We now describe how DAGs are used to encode independence restrictions in a missing data model.  

%####################################################
\section{Missing Data DAG Models}
\label{sec:missing_data_dag_models}
%####################################################

% Unlike a causal DAG, a missing data DAG (or m-DAG for short), includes counterfactual variables on the graph. The only well-defined counterfactual variables in missing data problems are ones representing variables $L$ had they been observed, possibly contrary to fact. The superscript in $L^{(1)}$ thus corresponds to an intervention setting indicators $R$ to $1$. 

Unlike a causal DAG, a missing data DAG (or m-DAG for short) includes counterfactual variables directly on the graph. Specifically, an
m-DAG ${\cal G}_m(V)$ consists of a set of vertices $V$ associated with variables $L^{(1)}, R, L$. There may exist additional vertices corresponding to variables that are always observed, and variables that are never observed.  See Supplemental Section~S2
% \ref{sec:missing_data_w_hidden_vars} 
for examples.
In addition to acyclicity, an m-DAG restricts the presence of certain edges in the following ways:
\begin{enumerate}[(a)]
	
	\item $\pa_{\mathcal{G}_m}(L_{k})=\{L_{k}^{(1)},R_{k}\}$. That is, each proxy variable $L_k \in L$ has only two parents in ${\cal G}_m$ to represent the deterministic function that defines $L_k$ in terms of $L_k^{(1)}$ and $R_k$. These edges are drawn in gray in all  m-DAGs in this manuscript in order to distinguish deterministic relations from probabilistic ones. 
	
	\item $L^{(1)} \cap \left\{ \de_{\mathcal{G}_m}(R) \cup \de_{\mathcal{G}_m}(L) \right\} = \emptyset$. 
	That is, variables in $R$ and $L$ cannot have  directed paths to variables in $L^{(1)}$.
	{In the presence of always observed %auxiliary
 variables $W$ or always missing variables $U$, we also assume $(W \cup U) \cap \left\{ \de_{\mathcal{G}_m}(R) \cup \de_{\mathcal{G}_m}(L) \right\} = \emptyset$.}
%\footnote{In the presence of fully observed auxiliary variables $W,$ we also assume $W \cap \left\{ \de_{\mathcal{G}_m}(R) \cup \de_{\mathcal{G}_m}(L) \right\} = \emptyset$.} 
	
\end{enumerate}
{
A special case of missing data DAGs where $\ch_{{\cal G}_m}(L_i) = \emptyset$ for every $L_i$ was considered by \citep{mohan13missing}.
}

Restriction (a) is imposed by definition.  Since every $L_k$ is a deterministic function of $L_k^{(1)}$ and $R_k$, only those two variables can serve as causes of $L_k$, and thus as parents of $L_k$ in the graph.
Restriction (b) is imposed to ensure missingness indicators $R$, which are the missing data analogues of  treatment variables in causal inference, cannot influence counterfactual variables. This restriction is a consequence of the missing data version of consistency, which implies observed variables are caused by their corresponding counterfactuals, and not vice versa. For a detailed discussion on the implication of relaxing this assumption see
\citep{srinivasan2023graphical}.

The above restrictions imply that $\text{ch}_{{\cal G}_m}(L_{k}) \subseteq \{ R_{i} \in R \mid i \neq k\}$ and $\text{ch}_{{\cal G}_m}( R_{k}) \subseteq \{L_k\} \cup \{ R_{i} \in R \mid i \neq k\}$. That is, the values of the observed $R_k,L_k$ can only influence the decisions $R_i$ as to which of the other variables $L_i^{(1)}$ will be observed.

The statistical model of an m-DAG ${\cal G}_{m}$, denoted by ${\cal M}^{{\cal G}_{m}}$,
 consists of a set of joint distributions that 
 factorize with respect to ${\cal G}_m$.
 In the simple case where the variables in $U,W$ are not present, the joint distributions
 $p({l},{r},{l}^{({1})})$ in ${\cal M}^{{\cal G}_{m}}$
 %that   factorize with respect to an m-DAG $\mathcal{G}_m$ with vertex set $\{{L}, {R}, {L}^{({1})}\}$
 factorize as follows
{\small
\begin{equation}
%	p({l}, {r}, {l}^{({1})} = 
	\prod_{v_i \in {l} \cup {r}\cup {l}^{({1})}} 
	p(v_i \mid \pa_{\mathcal{G}_{m}}(v_i))
	= \prod_{l_{k}\in {l}}  p(l_{k}\mid
	l_{k}^{(1)},r_{k}) \times \prod_{v_i \in {r}\cup {l}^{({1})}} p(v_i \mid \pa_{\mathcal{G}_m}(v_i)).
	\label{eq:miss-dag-factorization}
\end{equation}%
}
%while also satisfying determinism restrictions stemming from the definition of each $L_k$ in terms of $R_k$ and $L^{(1)}_k$.

The terms $p(l_{k}\mid l_{k}^{(1)},r_{k})$ are deterministically defined: $p(L_{k}=l_{k}\mid L_{k}^{(1)}=l_{k}^{(1)},R_{k}=1)=1$ and $p(L_{k} =\text{\textquotedblleft ?\textquotedblright } \mid L_{k}^{(1)}=l_{k}^{(1)}, R_{k}=0)= p(L_{k}=\text{\textquotedblleft ?\textquotedblright  }\mid R_{k}=0)=1 $ for any $L_{k} \in L$. The joint distribution  $p({l},{r},{l}^{({1})})$ further satisfies the positivity assumption stated as follows: 
$
p( R_k = r_k \mid \pa_{{\cal G}_m}({r}_k) ) > 0$ with probability $1$ for all $R_k \in R$.
%As pointed out before,
This assumption excludes monotone missingness.

%\color{red}
%This result follows, since
%Note that
We can view an m-DAG as a special case of a causal DAG (with hidden variables) described in the previous section
%satisfying properties
%(i), (ii), and (iii)
where $Y$ is taken to be $L$, and $A^{\dag}$ is taken to be $R$, with a set of additional restrictions.
Specifically, every variable $L_k$ that is potentially missing corresponds to \emph{one} non-trivial counterfactual $L_k^{(1)}$, with the other counterfactual $L_k^{(0)}$ trivially defined as ``?''.  In addition, every treatment variable $R_k$ in $R$ can only affect exactly one outcome variable, namely $L_k$.

Given a missing data DAG model, our objective is to determine whether the target law $p({l}^{(1)})$, or a fixed function of the target law, can be identified as a function of the observed data law $p\left({r,l}\right)$, and if so, find the identifying functional. To aid subsequent developments we will reformulate the identification problem in missing data models using the language common in causal inference.  In this view, given an underlying variable $L_i^{(1)}$, the corresponding proxy variable $L_i$ is viewed as an observed version of ``the outcome,'' and the corresponding missing indicator $R_i$ as the observed version of ``the treatment.''

Since m-DAGs are a special case of causal DAGs, the following results immediately follows. 
\begin{proposition}
	\label{prop:miss-id}
	Under the missingness model %associated with
	of an m-DAG $\mathcal{G}_m$
	{\small
	\begin{align}
		p({l}^{({1})}) &=   \prod_{v_k \in {l}^{({1})}} p(v_k \mid \pa_{\mathcal{G}_m}(v_k)) 
		= 
		\left. \frac{p({l},{r})}{ \prod_{r_k \in {r}} p(r_k \mid \pa_{\mathcal{G}_m}(r_k)) } \right|_{{r}={1}}. 
		\label{eq:mDAG-bayes}
	\end{align}
	}
\end{proposition}
The second equality in (\ref{eq:mDAG-bayes}) is the missing data DAG equivalent of the counterfactual g-formula in (\ref{eqn:ggf}) for causal inference problems. In missing data DAG models, any counterfactual variable $L_{k}^{\left( 1\right) }$ is allowed to have elements of $R$ as children. This means that the g-formula in (\ref{eq:mDAG-bayes})  does not necessarily lead to identification of $p({l}^{({1})})$ in terms of the observed data distribution $p({l},{r})$. This is because $\pa_{\mathcal{G}_{{m}}}(r_k)$ in $p(r_k \mid \pa_{\mathcal{G}_{{m}}}(r_k))$ may involve values in ${l}^{({1})}$ which are not always observed and cannot be immediately dropped via independence assumptions. {This is analogous to why unmeasured variables lead to non-identification in causal inference problems, as discussed earlier.} In the next section, we illustrate via a number of examples how identification may nevertheless be accomplished in some missing data DAG models representing MNAR mechanisms. 

\subsection{Hierarchy of missing data DAG models}

Similar to Rubin's hierarchy of missingness mechanisms, it is possible to set up a hierarchy for missing data DAG models that define the complexity of identification techniques required. The missing data DAG model for a graph ${\cal G}_m$ with vertices $\{L^{(1)}, R, L\}$ is said to be
\begin{itemize}
	\setlength{\itemsep}{0.cm} 
	
	\item Missing Completely At Random (MCAR) if $\prod_{r_k \in r} p(r_k \mid \pa_{{\cal G}_m}(r_k))$ is not a function of variables in $L^{(1)}$ and $L$. Graphically speaking there are no edges that point to variables in $R$.
	\item Missing At Random (MAR) if $\prod_{r_k \in r} p(r_k \mid \pa_{{\cal G}_m}(r_k))$ is not a function of variables in $L^{(1)}$. Graphically speaking there are no edges from variables in $L^{(1)}$ to variables in $R$. {Note that this model is a submodel of the MAR model, described below, with the additional restriction that the full law factorizes with respect to an m-DAG.}
	\item Missing Not At Random (MNAR) otherwise.
\end{itemize}
A MAR missingness mechanism, according to Rubin's definition, requires that for any given missingness pattern, the missingness is independent of the missing values given the observed values. These types of restrictions are with respect to missingness patterns $R=r$ (a total of $2^K$ distinct patterns for models with $K$ missing variables). This %\stkout{pattern-specific MAR}
missingness mechanism cannot be represented via graphs, and several authors have noted the difficulty in interpretation %and application 
of MAR models in practice \citep{gill97coarsening, robins1997non, schafer2002missing, mcknight2007missing, graham2012missing, tian2015missing}. The above graphical hierarchy  provides a more intuitive description of the above types of missingness mechanisms, in the sense that any missing data model associated with an m-DAG comes equipped with a ready description of a data generating process, where variables are generated sequentially according to a total order consistent with the m-DAG.

Further, while there exist MNAR models whose restrictions cannot be represented graphically, for instance the complete-case missing value model \citep{little1993pattern}, {or the discrete choice model} \citep{tchetgen16discrete}, the restrictions posed in several popular MNAR models, such as the permutation model \citep{robins97non-a}, the block-sequential MAR model \citep{zhou10block},
%the itemwise conditionally independent nonresponse (ICIN) model \citep{shpitser2016consistent, sadinle16itemwise}, 
and those in \cite{thoemmes2013selection, martel2013definition, mohan13missing, shpitser2016consistent,  saadati2019adjustment, bhattacharya19mid} and \cite{nabi20completeness}  correspond to DAGs.  {Models described in \citep{shpitser2016consistent, sadinle16itemwise,malinsky2021semiparametric} correspond to graphical models that generalize DAG models, and are instead associated with chain graphs \citep{Lauritzen96graphical}.}
%, and \cite{malinsky2021semiparametric}, are either explicitly graphical, or can be interpreted as such.

%Recently, \cite{chen2022pattern} also proposed a graphical representation of the complete-case missing value MNAR model described in \cite{little1993pattern} and \cite{tchetgen16discrete}. These so-called pattern graph missing data models, however, are quite different from the Markov factorization of m-DAG models considered in this work. 

\subsection{Examples of missing data DAG models}

We now present examples of missing data models from prior literature  that impose restrictions, which can be encoded via m-DAGs.

\textbf{Permutation model.}  
Given an ordering (permutation) indexed by $k\in \{1,\ldots ,K\}$ on variables in $L^{({1})}$, \cite{robins97non-a} defined a \emph{permutation missingness} model by the restrictions that $R_{k}$ is independent of the current and past counterfactuals given the observed past and future counterfactual variables. We denote this model by ${\cal M}^\text{per}.$ Formally, ${\cal M}^\text{per}$ is defined by the following conditional independence restrictions: 
%{\small
%\begin{equation}
$R_{k} \ \ci \
	\overline{L}_{k}^{(1)} \ \big| \ 
	\overline{R}_{k-1}, \overline{L}_{k-1}, \underline{L}_{k+1}^{(1)}$ for all  $k \in \{1, \ldots, K\}. $
%	\label{eqn:per}
%\end{equation}%
%}
The graphical representation of $\mathcal{M}^\text{per}$ for two time points defined by these set of independencies  is shown in Fig.~\ref{fig:m-DAG-examples}(a). The local Markov property for this m-DAG model yields the following set of independence restrictions: $R_{1} \ \ci \ L_{1}^{(1)} \mid L_{2}^{(1)}$ and $R_{2} \ \ci \ L_{2}^{(1)},L_{1}^{(1)} \mid L_{1},R_{1}.$
%{\small
%\begin{align*}
%	R_{1}& \ \ci \ L_{1}^{(1)} \mid L_{2}^{(1)} \quad \text{and} \quad R_{2} \ \ci \ L_{2}^{(1)},L_{1}^{(1)} \mid L_{1},R_{1}. 
%\end{align*}%
%}
%These two restrictions correspond to the ones in display (\ref{eqn:per}) for two time points. 

%\rohit{
%Note the similarity between the permutation missing data model and the hidden variable causal DAG in Fig.~\ref{fig:iterate}(a). While it was impossible to identify $p(l_1^{(r_1, r_2)}, l_2^{(r_1, r_2)})$ from Fig.~\ref{fig:iterate}(a), in Section~\ref{sec:miss-ID} we discuss how $p(l_1^{(1)}, l_2^{(1)})$ is identified under the permutation model.
%}

\begin{figure}[!t] 
	\begin{center}
		\scalebox{0.8}{
			\begin{tikzpicture}[>=stealth, node distance=1.5cm]
				\tikzstyle{format} = [thick, circle, minimum size=1.0mm, inner sep=0pt]
				\tikzstyle{square} = [draw, thick, minimum size=4.5mm, inner sep=3pt]
				
				\begin{scope}[xshift=0cm]
					\path[->, thick]
					node[format] (x11) {$L^{(1)}_1$}
					node[format, right of=x11, xshift=0.45cm] (x21) {$L^{(1)}_2$}
					node[format, below of=x11] (r1) {$R_1$}
					node[format, below of=x21] (r2) {$R_2$}
					node[format, below of=r1] (x1) {$L_1$}
					node[format, below of=r2] (x2) {$L_2$}
					
					(x11) edge[blue] (x21) 
					(x21) edge[blue] (r1)
					(r1) edge[blue] (r2)
					(x1) edge[blue] (r2)
					(x11) edge[gray, bend right=25] (x1)
					(x21) edge[gray, bend left=25] (x2)
					(r1) edge[gray] (x1)
					(r2) edge[gray] (x2)
					
					node[format, below of=x1, xshift=1cm, yshift=0.5cm] (a) {(a) {\small Permutation}} ;
				\end{scope}
				
				\begin{scope}[xshift=4.5cm]
					\path[->, thick]
					node[format] (x11) {$L^{(1)}_1$}
					node[format, right of=x11, xshift=0.45cm] (x21) {$L^{(1)}_2$}
					node[format, below of=x11] (r1) {$R_1$}
					node[format, below of=x21] (r2) {$R_2$}
					node[format, below of=r1] (x1) {$L_1$}
					node[format, below of=r2] (x2) {$L_2$}
					
					(x11) edge[blue] (x21) 
					(x11) edge[blue] (r2)
					(x21) edge[blue] (r1)
					(x11) edge[gray, bend right=25] (x1)
					(x21) edge[gray, bend left=25] (x2)
					(r1) edge[gray] (x1)
					(r2) edge[gray] (x2)
					
					node[format, below of=x1, xshift=1cm, yshift=0.5cm] (b) {(b) {\small Block-parallel}} ;
				\end{scope}
				
				\begin{scope}[xshift=9cm]
					\path[->, thick]
					node[format] (x11) {$L^{(1)}_1$}
					node[format, right of=x11, xshift=0.45cm] (x21) {$L^{(1)}_2$}
					node[format, below of=x11] (r1) {$R_1$}
					node[format, below of=x21] (r2) {$R_2$}
					node[format, below of=r1] (x1) {$L_1$}
					node[format, below of=r2] (x2) {$L_2$}
					
					(x11) edge[blue] (x21) 
					(x11) edge[blue] (r2)
					(r1) edge[blue] (r2)
					(x11) edge[gray, bend right=25] (x1)
					(x21) edge[gray, bend left=25] (x2)
					(r1) edge[gray] (x1)
					(r2) edge[gray] (x2)
					
					node[format, below of=x1, xshift=1cm, yshift=0.5cm] (c) {(c) {\small Block-sequential}} ;
				\end{scope}
				
				\begin{scope}[xshift=13.5cm]
					\path[->, thick]
					node[format] (x11) {$L^{(1)}_1$}
					node[format, right of=x11, xshift=0.45cm] (x21) {$L^{(1)}_2$}
					node[format, below of=x11] (r1) {$R_1$}
					node[format, below of=x21] (r2) {$R_2$}
					node[format, below of=r1] (x1) {$L_1$}
					node[format, below of=r2] (x2) {$L_2$}
					
					(x11) edge[blue] (x21) 
					(r1) edge[blue] (r2)
					(x1) edge[blue] (r2)
					(x11) edge[gray, bend right=25] (x1)
					(x21) edge[gray, bend left=25] (x2)
					(r1) edge[gray] (x1)
					(r2) edge[gray] (x2)
					
					node[format, below of=x1, xshift=1cm, yshift=0.5cm] (d) {(d) {\small MAR example}} ;
				\end{scope}
			\end{tikzpicture}
		}
		\vspace{-2.5cm}
		\caption{Examples of missing data DAG models. }
		\label{fig:m-DAG-examples}
	\end{center}
\end{figure}
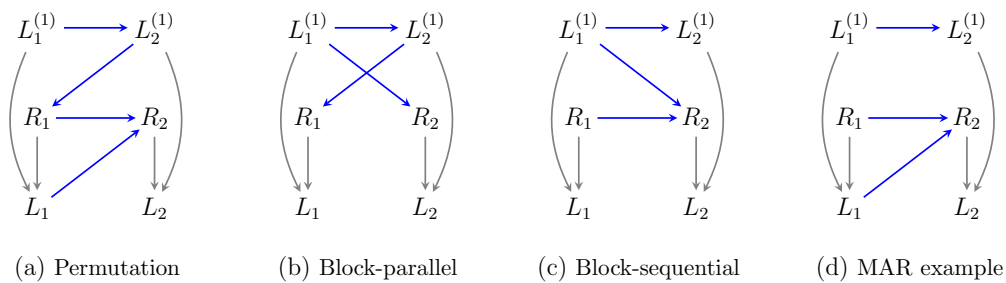

\textbf{Block-parallel MNAR model.} 
The block-parallel MNAR model, denoted $\mathcal{M}^\text{b-par}$, was introduced in \cite{mohan13missing}.
%\footnote{Though \cite{mohan13missing} did not use the name block-parallel MNAR to refer to models defined by the restrictions in (\ref{eq:b-par}), the reason for our use of this term should become apparent in Section~\ref{subsec:parallel-fix} when we discuss identification of missing data models that satisfy (\ref{eq:b-par}).} 
It is defined by the following conditional independence restrictions: $R_{k} \ \ci \ \{L^{(1)}_{k}, R_{-k}\}  \ \mid {L}^{(1)}_{-k},$ for all $k \in \{1, \ldots, K\}$, 
%{\small
%\begin{equation}
%	R_{k} \ \ci \ \{L^{(1)}_{k}, R_{-k}\}  \ \mid {L}^{(1)}_{-k},   \quad \forall k \in \{1, \ldots, K\}
%	\label{eq:b-par}
%\end{equation}%
%}
where $V_{-k} \coloneqq V \setminus V_k.$  
{The reason for our use of this name should become apparent in Section~\ref{subsec:parallel-fix} when we discuss the  identification of this missing data model.}
The graphical representation of ${\cal M}^\text{b-par}$ for two time points, defined by this set of independencies, 
%in display (\ref{eq:b-par}), 
is shown in Fig.~\ref{fig:m-DAG-examples}(b). The local Markov property for this m-DAG model implies: $R_{1} \ \ci \ L_{1}^{(1)}, R_2 \mid L_{2}^{(1)}$ and $R_{2} \ \ci \ L_{2}^{(1)}, R_1 \mid L_{1}^{(1)}.$
%{\small
%\begin{align*}
%	R_{1}& \ \ci \ L_{1}^{(1)}, R_2 \mid L_{2}^{(1)} \quad \text{and} \quad
%	R_{2} \ \ci \ L_{2}^{(1)}, R_1 \mid L_{1}^{(1)}. 
%\end{align*}% 
%}

\textbf{Block-sequential MNAR model.}  
The block-sequential MNAR model, denoted by $\mathcal{M}^\text{b-seq}$, was introduced in \cite{zhou10block}{, under the name ``block-conditional MAR.'' We use our alternative name to make it clear that this model is not MAR but MNAR.}
%\footnote{The reference denotes this model as ``block-conditional MAR.'' We use our alternative name to make it clear that this model is not MAR but MNAR.}
It is defined by the following conditional independence restrictions: $R_{k+1} \ \ci \ \underline{{L}}^{(1)}_{k+1}  \ \mid \ \overline{R}_{k}, \overline{L}^{(1)}_{k},  \quad \forall k \in \{1, \ldots, K\}.$
%{\small
%\begin{equation}
%	R_{k+1} \ \ci \ \underline{{L}}^{(1)}_{k+1}  \ \mid \ \overline{R}_{k}, \overline{L}^{(1)}_{k},  \quad \forall k \in \{1, \ldots, K\}. 
%	\label{eq:bs-MAR}
%\end{equation}%
%}
%Under the monotonicity assumption, the block-sequential MNAR model reduces down to the monotone MAR model. 
The graphical representation of ${\cal M}^\text{b-seq}$ for two time points, defined by this set of independencies, %in display (\ref{eq:bs-MAR}), 
is shown in Fig.~\ref{fig:m-DAG-examples}(c). The local Markov property for this m-DAG model implies: $R_{1} \ \ci \ L_{1}^{(1)}, L_{2}^{(1)}$ and $R_{2} \ \ci \ L_{2}^{(1)} \mid R_1, L_{1}^{(1)}. $
%{\small
%\begin{align*}
%	R_{1}& \ \ci \ L_{1}^{(1)}, L_{2}^{(1)} \quad  \text{and} \quad
%	R_{2} \ \ci \ L_{2}^{(1)} \mid R_1, L_{1}^{(1)}. 
%\end{align*}% 
%}

\textbf{An example of a MAR model.} 
The model introduced in \cite{rubin76inference} is defined via the following conditional independence restrictions: $\mathbb{I}({R}={r}) \ \ci  \ \left\{ L_{i}^{(1)}; \ r_{i}=0\right\} 
\ \Big| \ \left\{ L_{j}^{(1)}; \ r_{j}=1\right\} \ \text{ for all } r \in \{0, 1\}^K$, 
%{\small
%\begin{equation}
%	\mathbb{I}({R}={r}) \ \ci  \ \left\{ L_{i}^{(1)}; \ r_{i}=0\right\} 
%	\ \Big| \ \left\{ L_{j}^{(1)}; \ r_{j}=1\right\} \ \text{ for all } r \in \{0, 1\}^K,
%	\label{eq:non-monotone-MAR}
%\end{equation}%
%}
where $\mathbb{I}(.)$ is the indicator function. 
This MAR model cannot be represented by a DAG $\mathcal{G}_m$ with vertices $V=\{ {L}^{({1})},R, L\}$. This follows because the $2^{K}$ variables $\mathbb{I}({R}={r})$ are not vertices on the graph. However, we can have more intuitive {submodels of the} MAR model that can be represented graphically. The missing data DAG for two time points in Fig.~\ref{fig:m-DAG-examples}(d) is one example. The local Markov property for this m-DAG model implies: $R_{1} \ \ci \ L_{1}^{(1)}, L_{2}^{(1)}$ and $R_{2} \ \ci \ L_{2}^{(1)} \mid R_1, L_{1}. $

An interesting observation is that under the monotonicity assumption, the MAR model in Fig.~\ref{fig:m-DAG-examples}(d) and the block-sequential MNAR model in Fig.~\ref{fig:m-DAG-examples}(c) are identical. The monotonicity assumption in the block-sequential MNAR model imposes restrictions on the univariate conditionals of each $R_k$ given their parents on the graph: evaluating any of the parental missingness indicators at zero deterministically defines the conditional density of $p(r_k \mid \pa_{{\cal G}_m}(r_k))$. This is because if $i \prec_{{\cal G}_m} k$, and $R_i = 0$ then it must be the case that $R_k = 0$, otherwise the monotonicity assumption is violated. The only non-deterministic evaluation of the univariate conditionals occurs when $R_i = 1, \forall R_i\in \pa_{{\cal G}_m}(R_k)$. Thus, by consistency, we can replace the pair $(R_i = 1, L^{(1)}_i) \in \pa_{{\cal G}_m}(R_k)$ with  $(R_i = 1, L_i)$. The  equivalent graphical operation is replacing the set of edges $\{L^{(1)}_i \ \rightarrow \ R_j \ \leftarrow \ R_i \}$ with the set $\{L_i \ \rightarrow \ R_j \ \leftarrow  \ R_i \}$. This renders the DAG in Fig.~\ref{fig:m-DAG-examples}(c) and the one in Fig.~\ref{fig:m-DAG-examples}(d) equivalent.  The monotonicity assumption is a restriction on patterns, and cannot be directly represented in an m-DAG{, unless it is augmented with additional edge markings denoting deterministic relationships among $R$ variables that define monotonicity.}

% \red{How would an analyst choose the appropriate missing data DAG representation for their analysis? Some but not all restrictions implied by a missing data DAG are empirically testable using observed data. We now provide a brief discussion on how one might choose the appropriate representation based on intuitive explanations of the missing data generation process, and briefly note when these explanations also correspond to empirically testable restrictions.}

We now compare the above missingness models on the basis of how they differ in telling a story about the underlying missingness mechanisms.

Fig.~\ref{fig:m-DAG-examples}(a) is a two-variable version of the permutation model described in \citep{robins97non-a}. 
Suppose a reform school offered HIV testing, with 35\% of those tested being HIV positive among the 30\% who accepted the offer, i.e., $p(L_1=1 \mid R=1) = 0.35$. Due to concerns about non-random non-response, data on HIV risks and fears were abstracted as variable $L_2$ from a %20\% simple random 
stratified sample of counseling files, with sampling fractions of $20\%$, in the stratum  $L_1= ``?$'', $50\%$ in the stratum $L_1=1$, and $60\%$ in the stratum $L_1=0$.
By design, $p(R_2=1 \mid L^{(1)}_1, L^{(1)}_2, R_1)$ is only a function of $L_1$ (and thus also of $R_1$). Under the additional assumption that being tested ($R_1$) is independent of HIV status ($L^{(1)}_1$) given $L^{(1)}_2$, the data follow a non-monotone non-ignorable permutation model process where $(L^{(1)}_1, L^{(1)}_2)$ % = (L_1, L_2)$.
are true HIV status, and abstract risks and fears of HIV infection, respectively.

Fig.~\ref{fig:m-DAG-examples}(b) is a two variable version of the model first described in \cite{mohan14missing}.  In this model, the censoring status of each of the two variables is not determined by the underlying value of the variable itself, but by the underlying value of the other variable.  An example of such censoring processes occur in prisoner's dilemma situations, where $L_1$ represent recorded parts of a criminal case against defendant $1$ obtained from testimony, while $L_1^{(1)}$ represent the true facts about defendant $1$.  Many jurisdictions (such as the United States) forbid prosecutors from coercing self-incrimination, however co-conspirators are often induced to incriminate each other under plea deals.  In such cases, observability of $L_1$, represented by $R_1$ may be causally influenced by the severity of the crime $L_2^{(1)}$ of a co-conspirator $2$ of defendant $1$, and similarly the observability status $R_2$ of $L_2$ would be influenced by $L_1^{(1)}$.  Laws against self-incrimination prevent $L_1^{(1)}$ from influencing $R_1$, and $L_2^{(1)}$ from influencing $R_2$.

Fig.~\ref{fig:m-DAG-examples}(d) is the classic monotone MAR model, which has been used to represent dropout in longitudinal observational studies. For instance, consider a study of the effect of highly active antiretroviral therapy (HAART) on the CD4 outcome in HIV patients \citep{cain2016using}.  In such a study, a set of covariates $L_i^{(1)}$ are measured at every followup period $i$ for every patient that remains in the study.  However, once a patient withdraws from the study starting at a particular time period $j$, for instance due to withholding informed consent, no measurements of $L_k^{(1)}$ for $k \geq j$ take place.  Whether a patient withdraws from the study depends on covariates for the patient which have been measured so far.
Fig.~\ref{fig:m-DAG-examples}(d) represents a very simple such scenario for two time points, where the dropout indicator $R_2$ depends on observed covariates $L_1$ measured at the previous time point. 
In addition, $R_2$ depends on the dropout indicator $R_1$ to enforce the monotonicity restriction, which states that $p(R_2 = 0 \mid R_1 = 0, L_1 = ``?") = 1$.

Fig.~\ref{fig:m-DAG-examples}(c) represents a similar type of longitudinal study, but where missingness is instead \emph{intermittent}.  In such a study, a set of covariates $L_i^{(1)}$, including determinants of health and socioeconomic factors, are surveyed at every followup period $i$ for every study participant, provided that they show up for their appointment.  However, not every participant shows up to every appointment, often for reasons caused by the covariates that were not measured, such as economic precarity.

In the simple two timepoint version of such a study shown in Fig.~\ref{fig:m-DAG-examples}(c), the censoring indicator $R_2$ depends on the indicator $R_1$ at the previous time point, as well as covariates $L_1^{(1)}$ at that time point, which may be observed for some participants, and not observed for others.

% \blue{We now compare the above missingness models on the basis of how they differ in telling a story about the underlying missingness mechanisms. Say $L_1^{(1)}$ corresponds to the true smoking status of an individual and $L_2^{(1)}$ corresponds to their diagnosis of bronchitis. The missingness indicators $R_1$ and $R_2$ then encode whether these variables have been measured or not in the observed sample.}  

% The ${\cal M}^{\text{per}}$ model in Fig.~\ref{fig:m-DAG-examples}(a) implies that measurement of an individual's smoking status depends on the counterfactual value of their bronchitis status; this may occur for example when a patient's smoking status is inquired on a \emph{suspected} diagnosis of bronchitis before administering the test. The model further implies that whether the true bronchitis status is measured via a diagnostic test depends on the doctor's awareness of the individual's smoking status ($R_1$) and their observed value of smoking ($L_1$). The MAR model in Fig~\ref{fig:m-DAG-examples}(d) is simply a subgraph of this ${\cal M}^{\text{per}}$ model, where the absence of the $L_2^{(1)}\rightarrow R_1$ edge implies that inquiry into smoking status is now made randomly, without reference to a suspected diagnosis of bronchitis. Adding the monotonicity restriction further implies that a diagnostic test for bronchitis is administered only when the smoking status of an individual has been confirmed; this is similar  to settings where tests for certain diseases can be administered only after confirming the symptomatic/asymptomatic status of the individual.) 

Note that the permutation model in Fig.~\ref{fig:m-DAG-examples}(a) is untestable (i.e., though it implies restrictions on the full data law $p(l^{(1)}, r, l)$, it does not impose any restrictions on the observed data law $p(r, l)$) \citep{robins97non-a}; however, the MAR model in Fig.~\ref{fig:m-DAG-examples}(d) has a testable implication on the observed data distribution which can be used for falsification; see \cite{nabi2023testability}. 
% It is important to note that \cite{gill97coarsening} showed the pattern-specific MAR model  defined in \cite{rubin76inference} is indeed untestable.
%
% The ${\cal M}^{\text{b-par}}$ model in Fig.~\ref{fig:m-DAG-examples}(b) tells a different story. This m-DAG indicates that measurements of smoking and bronchitis (via inquiry on smoking status and administration of a diagnostic test respectively) depend on suspected values of these variables. Specifically, administration of the diagnostic test depends on the suspected smoking status of an individual, and inquiry into smoking status depends on a suspected diagnosis of bronchitis. 
%
The block-parallel model in Fig.~\ref{fig:m-DAG-examples}(b) also has a testable implication on the observed data distribution and hence can be falsified; see  \cite{malinsky2021semiparametric} and \cite{nabi2023testability}.

In the next section, we discuss how identification strategies developed in causal inference may be applied to  missing data problems, and how additional structures, unique to missing data, yields new identification strategies which were never encountered  in classical causal inference problems.

%####################################################
\section{Identification in Missing Data DAG Models} 
\label{sec:miss-ID}
%####################################################

In the view of missing data models associated with m-DAGs presented so far, the distribution $p({l}^{(1)})$ may be viewed instead as $p({l} \ \Vert \ r=1)$ -- the distribution where all missingness indicators in $R$ are intervened on and set to $1$.
%By consistency in missing data, we have the following equalities
Note that we have the following identities: $p(l^{(1)}) = p({l} \ \Vert \ r = 1) = p({l}^{(1)} \ \Vert \ r = 1)$. %$ = p({l}^{(1)}, {l} \ \Vert \ r = 1)$.
The first follows since $p(l^{(1)})$ is identified by the g-formula, and the second by consistency.

Recall from Proposition~(\ref{prop:miss-id}) that identification of the target law $p({l}^{({1})}) \coloneqq p({l} \ \Vert \ r = 1)$ is equivalent to identification of the missingness mechanism evaluated at $R=1$, i.e., the probability of observing the complete cases missing data pattern:
{$\prod_{R_i \in R} p(R_i = 1 \mid \pa_{{\cal G}_m}(r_i)) \vert_{R=1}$.}
If each conditional factor $p({r}_k \mid \pa_{{\cal G}_{{m}}}({r}_k))$ (evaluated at $R=1$) is identified in this product, then the complete cases missing data pattern, and consequently the target law would be identified. We refer to the conditional factor $p(R_k  =1 \mid \pa_{{\cal G}_{{m}}}({r}_k))$ as the \emph{propensity score} for $R_k$. Via a series of examples, we explore different identification strategies for identifying the distribution of the missingness mechanism  $p({r} \mid \pa_{{\cal G}_{{m}}}({r}))$ evaluated at the complete cases pattern values: $R=1$.
%\stkout{We will use ideas from causal inference in discussing these identification strategies, and thus reformulate missing data identification using the language of interventions.}

\noindent {\bf A Single Variable Interventional Reformulation of the Counterfactual G-formula.} 
Given a joint distribution $p(l^{(1)}, r, l)$ that factorizes according to a DAG ${\cal G}_m$, the joint distribution after an intervention on $R_k \in R$ is %defined via 
equal to the truncated factorization where the joint distribution is divided by the identified propensity score of $R_k$. %That is
Letting $r_k$ be $a$  in (\ref{eqn:g-formula}), and $l^{(1)}_{-k} \cup r_{-k} \cup l$ be $l$ in (\ref{eqn:g-formula}), we can write this distribution as
%In keeping with notation in (\ref{eqn:g-formula}), we can write this distribution as
{\small
\begin{align}
%	p({l}^{(1)} \setminus {l}^{(1)}_k, {r} \setminus {r}_k, {l} \ \Vert \ r_k = 1) \left.
%	= \frac{p( l^{(1)} \setminus l^{(1)}_k, r, l)}{p(r_k \mid \pa_{{\cal G}_{{m}}}({r}_k))} \right|_{{R_k}={1}}.
	p({l}^{(1)}_{-k}, {r}_{-k}, {l} \ \Vert \ r_k = 1) \left.
	= \frac{p( l^{(1)}_{-k}, r, l)}{p(r_k \mid \pa_{{\cal G}_{{m}}}({r}_k))} \right|_{{R_k}={1}}.
	\label{eq:intervention}
\end{align}
}
When we intervene on $R_k$ and set it to $1$, it becomes redundant to include both $L^{(1)}_k$ and $L_k$ in the joint distribution as they represent the same random variable when $R_k=1$; hence we drop $L^{(1)}_k$ from both sides of the equation.

The propensity score of $R_k$ is identified if we can replace each $L^{(1)}_j \in \pa_{{\cal G}_{{m}}}(R_k)$ with $\{L_j, R_j=1\}$. Such replacements are sometimes justified due to conditional independence restrictions in the full data distribution. For instance in Fig.~\ref{fig:m-DAG-examples}(b), the propensity of $R_1$ is identified because we have $p({r}_1 \mid {l}^{(1)}_2) = p({r}_1 \mid {l}^{(1)}_2, R_2 = 1) = p({r}_1 \mid {l}_2, R_2=1)$, since $R_1 \ci R_2 \mid L^{(1)}_2$.  When missingness indicators are connected, we may lose some of these convenient independence constraints. For instance in Fig.~\ref{fig:m-DAG-examples}(a), $R_1$ and $R_2$ are no longer conditionally independent; thus identification of $p({r}_1 \mid {l}^{(1)}_2)$ is not as straightforward as it is in  Fig.~\ref{fig:m-DAG-examples}(b). However, due to an \emph{invariance} property of the propensity scores, we can sometimes succeed in identification by exploring interventional distributions where a subset of observed variables are intervened on and consequently certain edges are removed. We formalize this property in the following lemma, which is analogous to the invariance property in causal inference. 

\begin{lemma}[Invariance property]
	\label{lem:invariance}
	Given the propensity score for $R_k \in R$, the conditioning set $\pa_{{\cal G}_{{m}}}(R_k)$ captures the  direct causes of $R_k$ and hence remains invariant under any set of interventions that disrupts other parts of the joint factorization. Given $R^* \subseteq R \setminus R_k$, we have
	{\small
	\begin{align}
		p\left({r}_k \mid \pa_{{\cal G}_{{m}}}({r}_k)\right) = p\left({r}_k \mid \pa_{{\cal G}_{{m}}}({r}_k) \ \Vert \ r^* = 1\right). \hspace{1.25cm} %\textit{(Invariance property)} 
		\label{eq:invariance}
	\end{align}%
	}
\end{lemma} %

%\vspace{-1.0cm}
Using this property, we now explore various strategies for target law identification in MNAR models. In Section~\ref{subsec:seq-fix}, we use the  permutation model as an example to illustrate how missingness mechanisms can sometimes be identified via a \emph{sequence} of interventions on missingness indicators. This is a generalization of causal inference techniques used in longitudinal studies to sequentially identify the effect of multiple treatments, where missingness indicator interventions are identified by treating counterfactuals as confounders. An intervention in missing data, unlike interventions in causal inference, can sometimes induce \emph{selection bias} due to conditioning on a subset of missingness indicators that have not yet been intervened on. Introduction of selection during intervention operations may make identification by means of sequential applications of the truncated factorization impossible.  In Section~\ref{subsec:parallel-fix}, we illustrate the selection issue of sequential interventions on the block-parallel model, and show how \emph{parallel} (simultaneous) interventions can obtain identification even if selection is present.  In some missing data models, such as the block-sequential model, we can identify the missingness probability of the complete cases pattern via either sequential or parallel application of interventions. 
%via truncated factorization. 
In Sections~\ref{subsec:seq-par-fix} through \ref{subsec:outside-R-fix}, we go over examples where a combination of sequential and parallel interventions are needed to identify each propensity score.  In Section~\ref{sec:unified}, we unify ideas explored in this section to yield a general identification procedure for arbitrary missing data DAGs by exploring all possible \emph{partial orders} of interventions defined on the observed variables on the graph. We also discuss some graphical structures in missing data DAG models that impede nonparametric identification of the model. 

%++++++++++++++++++++++++++++++++++++

\subsection{Sequential Interventions}
\label{subsec:seq-fix}

Consider the permutation model $\mathcal{M}^\text{per}$ with two time points, redrawn in Fig.~\ref{fig:permutation}(a). Our objective is to identify $p(l_1^{(1)}, l_2^{(1)}) \coloneqq p(l_1, l_2 \ \Vert \ r_1 = 1, r_2=1)$ as a function of observed data. This can be done in two steps: we first intervene on $R_2$, then we intervene on $R_1$.   

\noindent {\bf Step 1.} Intervene on $R_2$ to get $p(. \ \Vert \ r_2=1)$.  The propensity score of $R_2$, $p(r_2 \mid r_1, l_1)$, is a function of observed data and so the corresponding post intervention distribution is immediately identified. Intervening on $R_2$ and setting it to $1$ yields the following kernel:
{\small
\begin{align}
	p(l_1^{(1)}, l_2, r_1, l_1 \ \Vert \ r_2=1) &= \frac{ p(l^{(1)}_1, l_1, l_2,
		r_1, r_2 = 1) }{ p(r_2 = 1 \mid r_1, l_1) }.
	\label{eq:perm-r2}
\end{align}%
}
By consistency in missing data, $p(l_2 \ \Vert \ r_2=1)$, $p(l_2^{(1)} \ \Vert \ r_2=1)$, and $p(l_2^{(1)}, l_2 \ \Vert \ r_2=1)$ all represent the same object. This kernel factorizes with respect to the (conditional) m-DAG shown in Fig.~\ref{fig:permutation}(b), obtained from Fig.~\ref{fig:permutation}(a) by removing edges $R_1 \to R_2$ and $L_1 \to R_2$, and denoting the vertex $R_2$ as a square showing $r_2=1$, indicating the intervention setting $R_2$ to $1$.

In general, the graphical analogue of the intervention on $R_k$ entails removing all edges with arrowheads into $R_k$ in the corresponding missing data graph. We denote vertices corresponding to variables that have been intervened on with rectangles.  Further, the pair $(L^{(1)}_k, L_k)$ is treated as a single variable on the graph after intervening and setting $R_k = 1$ due to the deterministic relation of $L_k$ with $L_k^{(1)}$ and $R_k$. We keep the proxy variable on the graph but in gray with dashed  edges.

\noindent {\bf Step 2.} Intervene on $R_1$ after intervening on $R_2$ to get $p(. \ \Vert \ r_1 = 1, r_2=1)$.  In the second step, we want to intervene on $R_1$ in the kernel $p(. \ \Vert \ r_2 = 1)$ which is Markov relative to the DAG in Fig.~\ref{fig:permutation}(b). In order to perform this intervention, we need to show that $p(r_1 \mid l^{(1)}_2 \ \Vert \ r_2 = 1)$ is a function of $p(r, l).$ Using consistency, we have $p(r_1 \mid l^{(1)}_2 \ \Vert \ r_2 = 1) = p(r_1 \mid l_2 \ \Vert \ r_2 = 1)$, and using kernel probability rules (provided in \ref{eqn:kernel}), we have $p(r_1 \mid l_2 \ \Vert \ r_2=1) = { p(l_2, r_1 \ \Vert \ r_2=1) }/{	\sum_{r_1}  p(l_2, r_1 \ \Vert \ r_2=1) }$. The numerator here is simply a marginal of the kernel $p(. \ \Vert \ r_2=1)$ in (\ref{eq:perm-r2}) and is identified as follows:
{\small
\begin{align*}
	p(l_2, r_1 \ \Vert \ r_2=1) 
	&= \!  \sum_{l_1, l^{(1)}_1} p(l^{(1)}_1, l_2, r_1, l_1 \ \Vert \ r_2=1) 
	\! = \!  \sum_{l_1, l^{(1)}_1} \frac{ p(l^{(1)}_1, l_1, l_2, r_1, r_2 = 1) }{ p(r_2 = 1 \mid r_1, l_1) }  \\
	&\hspace{-0.5cm}= \sum_{l_1} \ \frac{ p(l_1, l_2, r_1, r_2 = 1) }{ p(r_2 = 1 \mid r_1, l_1) } 
	= \sum_{l_1} p(r_1, l_1) \times p(l_2 \mid r_1, l_1, r_2=1). 
\end{align*}
}
Consequently, we are able to identify the propensity score of $R_1$ in the second step and proceed with the intervention on $R_1$, using the kernel in (\ref{eq:perm-r2}). This yields a new kernel that factorizes with respect to the DAG in Fig.~\ref{fig:permutation}(c) and corresponds to the target law $p(l_1^{(1)}, l_2^{(1)})$. Putting all the pieces together, the target law is identified as:
{\small
\begin{align}
	&p(l_1, l_2 \ \Vert \ r=1) = 
	\frac{ p( l^{(1)}_1, l_2, r_1 = 1, l_1 \ \Vert \ r_2=1) }{ p(r_1=1 \mid l_2 \ \Vert \ r_2=1) } 	\label{eq:ex-sequential-per} \\
	&\hspace{0.5cm}= \frac{p(l_1,l_2,r_1=1,r_2=1)}{ p(r_2=1 \mid r_1 = 1,l_1) \times p(r_1=1 \mid l_2 \ \Vert \
		r_2=1)}  \nonumber \\
	&\hspace{0.5cm}= \frac{p(l_1,l_2,r_1=1,r_2=1)}{ p(r_2=1 \mid r_1 = 1,l_1) \times \displaystyle \frac{ \sum_{l_1}
			p(l_2 \mid r_2=1,r_1=1,l_1) \times p(r_1=1,l_1) }{ \sum_{l_1,r_1} p(l_2 \mid
			r_2=1,r_1,l_1) \times p(r_1,l_1) }}. \nonumber 
\end{align}
}
The identified forms of the two propensity scores in this example are quite different. The propensity score of $R_2$, $p(r_2 = 1 \mid r_1,l_1)$, corresponds to the conditional factor in the full factorization of the joint. The propensity score of $R_1$, $p(r_1=1 \mid l_2 \ \Vert \ r_2=1)$, is a complex function of $p(r, l)$ and  corresponds to a conditional distribution in a hypothetical world where $L_2^{(1)}$ is rendered observed and equal to $L_2$ via an intervention that sets $R_2$ to $1.$  

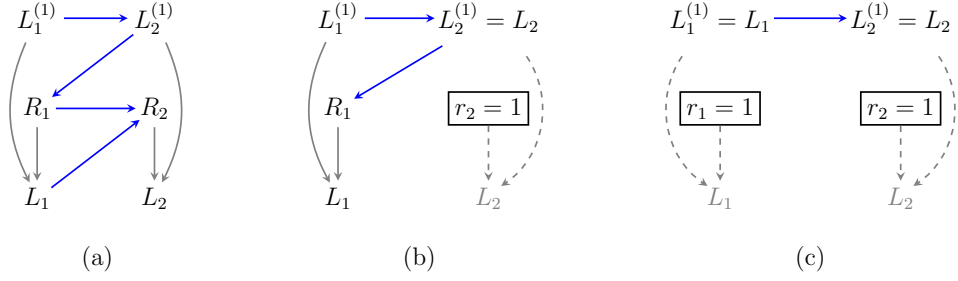
\begin{figure}[!t] 
	\begin{center}
		\scalebox{0.8}{
			\begin{tikzpicture}[>=stealth, node distance=1.5cm]
				\tikzstyle{format} = [thick, circle, minimum size=1.0mm, inner sep=0pt]
				\tikzstyle{square} = [draw, thick, minimum size=4.5mm, inner sep=3pt]
				
				\begin{scope}[xshift=0cm]
					\path[->, thick]
					node[format] (x11) {$L^{(1)}_1$}
					node[format, right of=x11, xshift=0.45cm] (x21) {$L^{(1)}_2$}
					node[format, below of=x11] (r1) {$R_1$}
					node[format, below of=x21] (r2) {$R_2$}
					node[format, below of=r1] (x1) {$L_1$}
					node[format, below of=r2] (x2) {$L_2$}
					
					(x11) edge[blue] (x21) 
					(x21) edge[blue] (r1)
					(r1) edge[blue] (r2)
					(x1) edge[blue] (r2)
					(x11) edge[gray, bend right=25] (x1)
					(x21) edge[gray, bend left=25] (x2)
					(r1) edge[gray] (x1)
					(r2) edge[gray] (x2)
%					;
					node[format, below of=x1, xshift=1cm, yshift=0.5cm] (a) {(a)} ;
				\end{scope}
				
				\begin{scope}[xshift=5.0cm]
					\path[->, thick]
					node[format] (x11) {$L^{(1)}_1$}
					node[format, right of=x11, xshift=1cm] (x21) {$L^{(1)}_2 = L_2$}
					node[format, below of=x11] (r1) {$R_1$}
					node[square, below of=x21] (r2) {$r_2 = 1$}
					node[format, below of=r1] (x1) {$L_1$}
					node[format, below of=r2] (x2) {\gray{$L_2$}}
					
					(x11) edge[blue] (x21) 
					(x21) edge[blue] (r1)
					(x11) edge[gray, bend right=27] (x1)
					(x21) edge[gray, dashed, bend left=45] (x2)
					(r1) edge[gray] (x1)
					(r2) edge[gray, dashed, ] (x2)
%					;
					node[format, below of=x1, xshift=1.35cm, yshift=0.5cm] (b) {(b)} ;
				\end{scope}
				
				\begin{scope}[xshift=11.35cm]
					\path[->, thick]
					node[format] (x11) {$L^{(1)}_1 = L_1$}
					node[format, right of=x11, xshift=1.5cm] (x21) {$L^{(1)}_2 = L_2$}
					node[square, below of=x11] (r1) {$r_1 = 1$}
					node[square, below of=x21] (r2) {$r_2 = 1$}
					node[format, below of=r1] (x1) {\gray{$L_1$}}
					node[format, below of=r2] (x2) {\gray{$L_2$}}
					
					(x11) edge[blue] (x21) 
					(x11) edge[gray, dashed, bend right=45] (x1)
					(x21) edge[gray, dashed, bend left=45] (x2)
					(r1) edge[gray, dashed] (x1)
					(r2) edge[gray, dashed] (x2)
%					;
					node[format, below of=x1, xshift=1.5cm, yshift=0.5cm] (c) {(c)};
				\end{scope}
				
			\end{tikzpicture}
		}
%		\vspace{-3cm}
		\caption{An illustration of the operation of a \textbf{sequential} identification algorithm. (a) Permutation model; (b)  Intervention on $R_2$; (c) Intervention on $R_1$ after $R_2$. } 
		\label{fig:permutation}
	\end{center}
\end{figure}

The fact that $R_1$ had a counterfactual parent $L_2^{(1)}$ in Fig.~\ref{fig:permutation}(a), but $R_1 \not\ci R_2 \mid \pa_{{\cal G}_m}(R_1)$  ensured that the propensity score $p(r_1 \mid \pa_{{\cal G}_m}(r_1))$ required to intervene on $R_1$ is not immediately identifiable. This induced a strict \emph{total ordering} by which $R_1$ and $R_2$  had to be intervened on. A total order is a special case of a \emph{strict partial order} which is defined as an irreflexible, anti-symmetric, and transitive binary relation.
%\footnote{For formal definitions of total and partial orders and their properties see chapter 7 in \cite{kwong2015spiral}.}
In our simple example above, to get identification for the target law we needed to follow the order $\{I_{r_2} \ < \ I_{r_1}\}$, where $I_{r_k}$ denotes an intervention on $R_k$. This sequential procedure generalizes to an arbitrary number of time points in the permutation model, where investigators take repeated measurements of some outcome over $K$ time points. The total order for identifying the target law $p(l_1^{(1)}, \dots, l_K^{(1)})$ is then given by a reverse topological ordering on the missingness indicators $\{I_{r_K} \ < \cdots \ < \ I_{r_1}\}$, i.e., we begin by intervening on the missingness indicator  $R_K$ corresponding to the final time point and proceed backwards \citep{robins97non-a}. %For a discussion on why regular causal identification theory fails to identify the target law in this example, see Appendix~\ref{app:causal_id}. 

{
Recall that in the causal model discussed earlier corresponding to Fig.~\ref{fig:iterate},
identification was not possible due to the presence of $U_1$ and $U_2$ instead of
the counterfactuals $L_1^{(1)}$ and $L_2^{(2)}$.
%with u1 and u2 in the role of the 2 counterfactuals we did not have identification.
Clearly the key for obtaining identification was the
additional information provided by counterfactuals $L^{(1)}$, which are sometimes observed, rather than $U$ variables,
which are never observed.
}
{In analogy with missing data, one may consequently suggest to consider placing counterfactual versions of observed variables on the graph in causal inference (instead of the hidden variables), with the understanding that a pair of counterfactuals $L^{(1)},L^{(0)}$ corresponding to $L$ (for some treatment $A$) may potentially be easier to take into account for identification compared to hidden variables. This is because the counterfactual pair $L^{(1)},L^{(0)}$ is partiallly observed, due to the consistency property, unlike a purely hidden variable $U$.  However, this is also insufficient for obtaining identification. In Section~\ref{sec:diss}, we discuss why identification strategies in missing data does not readily translate to causal inference. }

%ie had we not had the counterfatuals as the parents of  R1 we coukld not get the instantion we needed.

%So this is the key why do causal graphs not analogously allow to put
%counterfactuals on the graph replacing U's if its substantive sense. The
%reason is it does not help in identification because their factual's per
%treatment .. \ \ So instantiation of 1 does block confounding by the other
%and so we get no new identification by replacing Us with counterfactuals. \
%\ However in Secvtio \ we wikll show under an additonal sttrong abouyt the
%joing of ther pair factuals for each of rank preservation we do now get
%additonal identification in some causal models by replacing U by
%counterfactuals . and in fact these formally exactly analogous to the extra
%identifications we gt in missing data models. 

%++++++++++++++++++++++++++++++++++++

\subsection{Parallel Interventions}
\label{subsec:parallel-fix}

The sequential identification strategy in $\mathcal{M}^\text{per}$ coupled with consistency (i.e., $L_i^{(1)} = L_i$ if $R_i=1$)  resembles identification strategies employed in causal inference problems where the intervention distribution associated with a multivariate treatment $R=(R_1,R_2)^T$ is obtained sequentially by intervening on treatment variables one at a time. We now discuss an example that shows sequential strategies as in causal inference are insufficient and that there exists a class of missing data models where missingness indicators must be intervened on ``in parallel'' to identify the target law. 

The simplest example of a model in this class is the block-parallel MNAR model $\mathcal{M}^\text{b-par}$ with two time points, redrawn in Fig.~\ref{fig:parallel}(a). Identification of $p(l_1^{(1)}, l_2^{(1)}) \coloneqq p(l_1, l_2 \ \Vert \ r_1 = 1, r_2=1)$ can be obtained in a single step as follows:
{\small
\begin{align}
	\label{eqn:block-parallel-id}
 p(l_1^{(1)}, l_2^{(1)}) =
	\left. \frac{p(l_1, l_2, r_1, r_2)}{ p(r_1 \mid
		l_2^{(1)}) \times p(r_2 \mid l_1^{(1)})} \right|_{r=1} \!\!\!\! =  
  % \frac{p(l_1,l_2, r_1 = 1, r_2 = 1)}{ p(r_1 = 1 \! \mid \! l_2, r_2 = 1) \times p(r_2 = 1 \! \mid \! l_1,r_1 = 1)}. 
  \left. \frac{p(l_1,l_2, r_1, r_2)}{ p(r_1 \! \mid \! l_2, r_2) \times p(r_2 \! \mid \! l_1,r_1)} \right|_{r=1}. 
\end{align}%
}
The first equality holds by (\ref{eq:mDAG-bayes}) and the second equality holds due to the independence restrictions implied by $\mathcal{M}^\text{b-par}$ which are: $R_1 \ci R_2 \mid L_2^{(1)}$ and $R_1 \ci R_2 \mid L_1^{(1)}$ and consistency. Given these independencies, both propensity scores in the denominator are easily identified as  functions of the observed data.  

\begin{figure}[!t] 
	\begin{center}
		\scalebox{0.8}{
			\begin{tikzpicture}[>=stealth, node distance=1.5cm]
				\tikzstyle{format} = [thick, circle, minimum size=1.0mm, inner sep=0pt]
				\tikzstyle{square} = [draw, thick, minimum size=4.5mm, inner sep=3pt]
				\tikzstyle{format2} = [draw, circle, fill=gray!10, minimum size=2.5mm, inner sep=1.5pt]
				
				\begin{scope}[xshift=0cm]
					\path[->, thick]
					node[format] (x11) {$L^{(1)}_1$}
					node[format, right of=x11, xshift=0.45cm] (x21) {$L^{(1)}_2$}
					node[format, below of=x11] (r1) {$R_1$}
					node[format, below of=x21] (r2) {$R_2$}
					node[format, below of=r1] (x1) {$L_1$}
					node[format, below of=r2] (x2) {$L_2$}
					
					(x11) edge[blue] (x21) 
					(x21) edge[blue] (r1)
					(x11) edge[blue] (r2)
					(x11) edge[gray, bend right=25] (x1)
					(x21) edge[gray, bend left=25] (x2)
					(r1) edge[gray] (x1)
					(r2) edge[gray] (x2)
					
					node[format, below of=x1, xshift=0.85cm, yshift=0.5cm] (a) {(a)} ;
				\end{scope}
				
				\begin{scope}[xshift=5.5cm]
					\path[->, thick]
					node[format] (x11) {$L^{(1)}_1$}
					node[format, right of=x11, xshift=1cm] (x21) {$L^{(1)}_2 = L_2$}
					node[format2, below of=x11] (r1) {$R_1$}
					node[square, below of=x21] (r2) {$r_2 = 1$}
					node[format, below of=r1] (x1) {$L_1$}
					node[format, below of=r2] (x2) {\gray{$L_2$}}
					
					(x11) edge[blue] (x21) 
					(x21) edge[blue] (r1)
					(x11) edge[gray, bend right=30] (x1)
					(x21) edge[gray, dashed, bend left=45] (x2)
					(r1) edge[gray] (x1)
					(r2) edge[gray, dashed, ] (x2)
					
					node[format, below of=x1, xshift=1.35cm, yshift=0.5cm] (b) {(b)} ;
				\end{scope}
				
				\begin{scope}[xshift=12cm]
					\path[->, thick]
					node[format] (x11) {$L^{(1)}_1 = L_1$}
					node[format, right of=x11, xshift=1cm] (x21) {$L^{(1)}_2$}
					node[square, below of=x11] (r1) {$r_1 = 1$}
					node[format2, below of=x21] (r2) {$R_2$}
					node[format, below of=r1] (x1) {\gray{$L_1$}}
					node[format, below of=r2] (x2) {$L_2$}
					
					(x11) edge[blue] (x21) 
					(x11) edge[blue] (r2) 
					(x11) edge[gray, bend right=45, dashed] (x1)
					(x21) edge[gray, bend left=30] (x2)
					(r1) edge[gray, dashed] (x1)
					(r2) edge[gray] (x2)
					
					node[format, below of=x1, xshift=1.25cm, yshift=0.5cm] (c) {(c)} ;
				\end{scope}
				
			\end{tikzpicture}
		}
		\vspace{-0.75cm}
		\caption{An illustration of the operation of a \textbf{parallel} identification algorithm. (a) Block-parallel; (b) Intervention on $R_2$ and selection on $R_1$; (c) Intervention on $R_1$ and selection on $R_2$. } 
		\label{fig:parallel}
	\end{center}
\end{figure}
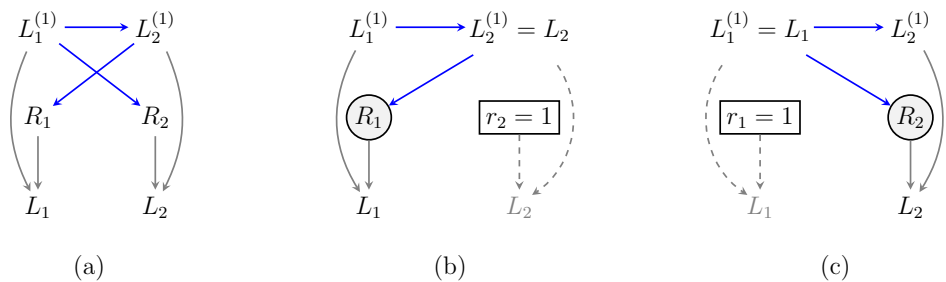

We now illustrate why the sequential approach to identification fails in this example and that the simultaneous evaluation of the two propensity scores is necessary to yield target law identification. Assume we proceed with our two-step sequential procedure by first intervening on $R_2 =1$. This results in the following kernel Markov relative to  the DAG in Fig.~\ref{fig:parallel}(b)
{\small
\begin{align}
	p(l_1^{(1)}, l_2, r_1, l_1 \! \ \Vert \ \! r_2 = 1) \!=  \! \left. \frac{p(l^{(1)}_1, l_1, l_2,
		r_1, r_2)}{p(r_2 \mid l_1^{(1)})} \right|_{r_2=1} \!\!\!\!\! = \! \frac{p(l^{(1)}_1,l_1,  l_2,
		r_1, r_2 = 1)}{p(r_2 = 1 \mid l_1^{(1)}, r_1)}. 
	\label{eq:ex-bp}
\end{align}
}
The intervention on $R_1$ in the second step of the sequential procedure requires dividing the above kernel $p(. \ \Vert \ r_2 = 1)$ by the propensity score of $R_1$, $p(r_1 \mid l^{(1)}_2) = p(r_1 \mid l_2 \ \Vert \ r_2 = 1)$. Using Bayes rule, this propensity can be obtained from $p(r_1, l_2 \ \Vert \ r_2=1)$. However, we can only identify $p(r_1 = 1, l_2 \ \Vert \ r_2=1)$ from $p(. \ \Vert \ r_2 = 1)$. This is because $l^{(1)}_1$ appears in both the numerator and denominator, and when $R_1=0$, we cannot observe $l^{(1)}_1. $ Hence, the kernel in (\ref{eq:ex-bp}) is identified only when evaluated at $R_1=1$.  The fact that this kernel is not available at all levels of $R_1$ prevents us from sequentially obtaining $p({r}_1  \mid \pa_{{\cal G}_{{m}}}({r}_1)) = p(r_1 \mid l_2 \ \Vert\ r_2=1) = p(l_2, r_1 \ \Vert\  r_2=1)/p(l_2 \ \Vert\  r_2=1)$ (where the first equality follows from invariance and the second from rules of kernel probability) due to our inability to sum out $R_1$ from $p(l_2, r_1 \ \Vert \ r_2 = 1)$ to obtain $p(l_2 \  \Vert \ r_2 = 1)$. 

The above is an example of what we term  \emph{selection} bias. In this case, division by $p(r_2 \mid l_1^{(1)})$ in (\ref{eq:ex-bp}) induces \emph{selection} on $R_1$, a variable that is not yet intervened on, in the kernel $p(. \ \Vert \ r_2 = 1)$. We draw a gray circle around the vertices on the graph that are selected on (or conditioned on) to distinguish it from interventions. Attempting a sequential identification procedure starting with $R_1$ would similarly induce selection bias on $R_2$ in the kernel $p(. \ \Vert \ r_1 = 1)$, as shown in Fig.~\ref{fig:parallel}(c). 
%Instead, the target law $p(l^{(1)}_1, l^{(1)}_2)$ is identified using  the two distinct kernels:  $p(l^{(1)}_1, l_2, r_1 = 1 \ \Vert \ r_2= 1)$ and $p(l_1, l^{(1)}_2, r_2 = 1 \ \Vert \ r_1 = 1)$. 
Since it is not possible to obtain identification by performing the intervention operation in a sequence, no total  ordering on missingness indicators can be imposed here.  Instead, interventions on $R_1$ and $R_2$ are \emph{incomparable}, and thus form a partial rather than a total ordering,  which is simply denoted via $\{ I_{r_1}, I_{r_2} \}$.  {Specifically, this notation denotes a partial order relationship on a set of two elements corresponding to interventions on indicators $R_1$ and $R_2$, where these elements are incompatible according to the order.}

In some missingness models, we might be able to identify the target law in multiple ways.  For instance, in the block-sequential MNAR model, the target law can be identified via either a sequence of interventions on missingness indicators in $R$, or via a one-step parallel intervention on all variables in $R$. These two strategies result in two equivalent representations of the same identifying functional.  These representations may, however, suggest different estimation strategies.

%++++++++++++++++++++++++++++++++++++

\subsection{Sequential and Parallel Interventions}
\label{subsec:seq-par-fix}

The previous two examples considered how identification of the target law in missing data problems entails evaluating the g-formula either sequentially or in parallel. These examples are special cases of a general identification algorithm for graphical missing data models, introduced in \cite{shpitser15missing}.  However, this algorithm is not complete in that it fails to account for the following example, and examples like it, where sequential and parallel applications of the g-formula must be combined in order to identify all of the propensity scores, and hence the target law. This strategy was introduced in \cite{bhattacharya19mid}.

Consider the missing data model shown in Fig.~\ref{fig:sequential-parallel}(a) on three variables. The target law $p(l_1^{(1)}, l_2^{(1)}, l_3^{(1)})$ $\coloneqq p(l_1, l_2, l_3 \ \Vert \ r_1=1, r_2=1, r_3=1)$ is equivalent to the following:
{\small
\begin{align}
	%p(l_1, l_2, l_3 \ \Vert \ r=1) = 
	\left. \frac{ p(l_1, l_2, l_3, r_1, r_2, r_3) }{ p(r_1 \mid l_2^{(1)}, l_3^{(1)}) \times  p(r_2 \mid l_3^{(1)}, r_1) \times p(r_3 \mid l_2^{(1)}, r_1)  } \right|_{{r} = {1}}. 
	\label{eq:ex-parallel-sec-gformula}
\end{align}%
}
Neither of the approaches discussed in the previous two subsections would yield an identified missingness mechanism for the example of Fig.~\ref{fig:sequential-parallel}(a). First, the sequential application discussed in (\ref{subsec:seq-fix}) fails since the selection induced on $R_2$ after intervening on $R_3$ impedes our next move in doing an intervention on $R_2$. This is because obtaining $p(r_2 \mid \pa_{{\cal G}_{{m}}}(r_2))$ from the kernel $p(. \ \Vert \ r_3 = 1)$ requires summing over $R_2$ (by Bayes rule), and that is not possible as $R_2$ is conditioned/forced to be one in the kernel. A similar problem persists if we intervene on  $R_2$ first ($R_3$ is conditioned to be one in the kernel $p(. \ \Vert \ r_2=1)$). A sequential application starting with $R_1$ is not possible since we cannot immediately obtain $p(r_1 \mid \pa_{{\cal G}_{{m}}}(r_1))$. This is because $R_1$ has counterfactual parents $L_2^{(1)}$ and $L_3^{(1)}$, but is not conditionally independent of the corresponding missingness indicators $R_2$ and $R_3$ given its parents. The parallel application discussed in (\ref{subsec:parallel-fix}) fails for the same reason, i.e., the propensity score for $R_1$ is not immediately identified.

We now show identification of this target distribution is still possible, but only by evaluating the g-formula in two sequential steps, with the first step consisting of a parallel evaluation of propensity scores for $R_2$ and $R_3$, and the second step consisting of the evaluation of propensity score for $R_1$ using the kernel distribution obtained from the first step. 

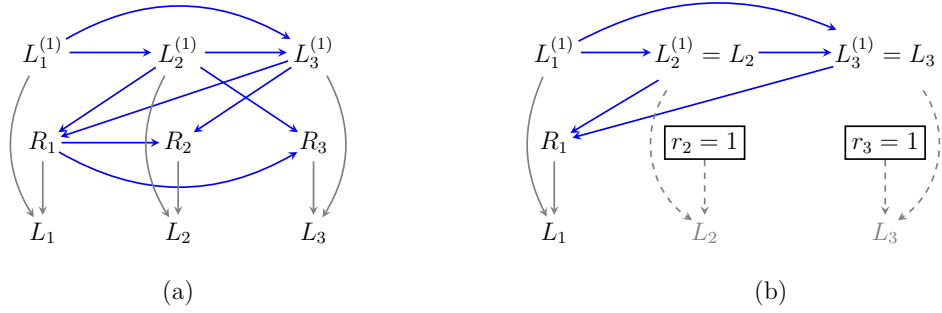
\begin{figure}[!t] 
	\begin{center}
		\scalebox{0.8}{
			\begin{tikzpicture}[>=stealth, node distance=1.5cm]
				\tikzstyle{format} = [thick, circle, minimum size=1.0mm, inner sep=0pt]
				\tikzstyle{square} = [draw, thick, minimum size=4.5mm, inner sep=3pt]
				
				\begin{scope}[xshift=0cm]
					\path[->, thick]
					node[format] (x11) {$L^{(1)}_1$}
					node[format, right of=x11, xshift=0.75cm] (x21) {$L^{(1)}_2$}
					node[format, right of=x21, xshift=0.75cm] (x31) {$L^{(1)}_3$}
					node[format, below of=x11] (r1) {$R_1$}
					node[format, below of=x21] (r2) {$R_2$}
					node[format, below of=x31] (r3) {$R_3$}
					node[format, below of=r1] (x1) {$L_1$}
					node[format, below of=r2] (x2) {$L_2$}
					node[format, below of=r3] (x3) {$L_3$}
					
					(x11) edge[blue] (x21) 
					(x21) edge[blue] (x31) 
					(x11) edge[blue, bend left] (x31) 
					
					(x21) edge[blue] (r1)
					(x21) edge[blue] (r3)
					(x31) edge[blue] (r1)
					(x31) edge[blue] (r2)
					
					(r1) edge[blue] (r2) 
					(r1) edge[blue, bend right] (r3)
					
					(x11) edge[gray, bend right=30] (x1)
					(x21) edge[gray, bend right=30] (x2)
					(x31) edge[gray, bend left=30] (x3)
					(r1) edge[gray] (x1)
					(r2) edge[gray] (x2)
					(r3) edge[gray] (x3)
					
					node[format, below of=x2, xshift=0.cm, yshift=0.5cm] (a) {(a)} ;
				\end{scope}
				\begin{scope}[xshift=8.5cm]
					\path[->, thick]
					node[format] (x11) {$L^{(1)}_1$}
					node[format, right of=x11, xshift=1.0cm] (x21) {$L^{(1)}_2 = L_2$}
					node[format, right of=x21, xshift=1.5cm] (x31) {$L^{(1)}_3 = L_3$}
					node[format, below of=x11] (r1) {$R_1$}
					node[square, below of=x21] (r2) {$r_2 = 1$}
					node[square, below of=x31] (r3) {$r_3 = 1$}
					node[format, below of=r1] (x1) {$L_1$}
					node[format, below of=r2] (x2) {\gray{$L_2$}}
					node[format, below of=r3] (x3) {\gray{$L_3$}}
					
					(x11) edge[blue] (x21) 
					(x21) edge[blue] (x31) 
					(x11) edge[blue, bend left=25] (x31) 
					(x21) edge[blue] (r1)
					(x31) edge[blue] (r1)
					
					(x11) edge[gray, bend right=25] (x1)
					(x21) edge[gray, dashed, bend right=45] (x2)
					(x31) edge[gray, dashed, bend left=45] (x3)
					(r1) edge[gray] (x1)
					(r2) edge[gray, dashed] (x2)
					(r3) edge[gray, dashed] (x3)
					
					node[format, below of=x2, xshift=1.1cm, yshift=0.5cm] (b) {(b)} ;
				\end{scope}
			\end{tikzpicture}
		}
		 \vspace{-0.75cm}
		\caption{(a) An example m-DAG corresponding to a model where interventions must be applied both \emph{sequentially and in parallel} to yield identification; (b) Graph derived from (a) representing the intermediate step  of the identification algorithm where $R_2$ and $R_3$ are simultaneously set to $1$.} 
		\label{fig:sequential-parallel}
	\end{center}
\end{figure}

{\bf Step 1.} \textit{Intervene on $R_2$ and $R_3$ in parallel to get $p(. \ \Vert \ r_2 = 1, r_3=1)$. }
The factorization of join distribution in Fig.~\ref{fig:sequential-parallel}(a) implies $R_2\ci R_3 \mid L_3^{(1)}, R_1$ and $R_3 \ci R_2 \mid L_2^{(1)}, R_1$. These independencies directly yield identified propensity scores for $R_2$ and $R_3$ as follows:
{\small
\begin{align}
	p(r_2 \mid \pa_{\mathcal{G}_{{m}}}(r_2))\vert_{r=1} &= p(r_2 \mid l_3^{(1)}, r_1)\vert_{r=1} = 
	%p(r_2 \mid l_3^{(1)}, r_1, r_3) \vert_{r=1} = 
	p(r_2  = 1 \mid l_3, r_1 =  r_3 = 1), \label{ex:parallel-seq-r2} \\
	p(r_3 \mid \pa_{\mathcal{G}_{{m}}}(r_3))\vert_{r=1} &= p(r_3 \mid l_2^{(1)}, r_1)\vert_{r=1} = 
	%p(r_3 \mid	l_2^{(1)}, r_1, r_2) \vert_{r=1} = 
	p(r_3 = 1 \mid l_2, r_1 = r_2 = 1). \label{ex:parallel-seq-r3}
\end{align}
}
This immediately implies that $R_2$ and $R_3$ can be intervened on in parallel, similar to the block-parallel example in the previous subsection. This results in the following kernel that factorizes with respect to the DAG in Fig.~\ref{fig:sequential-parallel}(b), 
{\small
\begin{align*}
	p(l_1^{(1)}, l_2, l_3, r_1, l_1 \ \Vert \ r_2=r_3=1) 
	= \frac{ p(l_1^{(1)}, l_1, l_2, l_3, r_1, r_2 = 1, r_3 = 1) }{  p(r_2 = 1 | l_3, r_1, r_3 = 1) \! \times \! p(r_3 = 1 | l_2,
		r_1, r_2 = 1)}.
\end{align*}%
}
The same strategy cannot be used for expressing $p(r_1\mid \pa_{\mathcal{G}}(r_1)) = p(r_1 \mid l_2^{(1)}, l_3^{(1)})$ as a function of the factual distribution since $R_1 \not\ci \{ R_2, R_3 \} \mid L_2^{(1)}, L_3^{(1)}$ under this model. However, a second, more involved, step leads to identifying the propensity score of $R_1$ and thus the target law $p(l_1^{(1)}, l_2^{(1)}, l_3^{(1)})$. 

{\bf Step 2.} \textit{Intervene on $R_1$ after intervening on $R_2, R_3$ to get $p(. \ \Vert \ r_1 = 1, r_2=1, r_3=1)$. }
In the second step, we want to intervene on $R_1$ using the above kernel $p(. \ \Vert \ r_2 = 1, r_3=1)$ Markov relative to the DAG in Fig.~\ref{fig:sequential-parallel}(b). In order to perform this intervention, we need to show that the propensity of $R_1$, $p(r_1 \mid l^{(1)}_2, l^{(1)}_2)= p(r_r \mid l_2, l_3 \ \Vert \ r_2 = 1, r_3=1)$ is a function of $p(r, l).$ We have $p(r_1 \mid l_2, l_3 \ \Vert \ r_2=r_3=1) = { p(l_2, l_3, r_1  \Vert  r_2=r_3=1) }/\sum_{r_1}  p(l_2, l_3, $$ r_1 \Vert r_2=r_3=1) $. The numerator here is identified by simply marginalizing out $l^{(1)}_1$ from the kernel $p(. \Vert r_2=r_3=1)$. Consequently, the propensity score of $R_3$ is rendered identified and we can identify the target law in (\ref{eq:ex-parallel-sec-gformula})  via: 
{\small
\begin{align*}
	p(l_1, l_2, l_3 \ \Vert \ r = 1) &= \left. \frac{ p(l_1^{(1)}, l_2, l_3, r_1, l_1 \ \Vert \ r_2=1, r_3=1) }{ p(r_1 \mid l_2, l_3 \ \Vert \ r_2=1,r_3=1) } \right|_{r_1=1}.
\end{align*}%
}
In conclusion, the parallel interventions on $R_2$ and $R_3$ do not induce selection bias on $R_1$. Therefore, we are able to proceed with the sequential application of the g-formula and obtain $p(r_1 \mid \pa_{{\cal G}_{{m}}}(r_1))$ from the kernel $p(.  \ \Vert \ r_2 = 1, r_3 = 1)$. In other words, identification of the target law is only possible if $R_1$ was intervened on only after $R_2$ and $R_3$ were simultaneously intervened on. This identification procedure also induces a partial ordering for the interventions on the missingness indicators. The partial order of interventions in the above example can be written as $\{ \{I_{r_2}, I_{r_3}\} < I_{r_1} \}$; that is interventions on $R_2$ and $R_3$ are incompatible, but both must occur prior to an intervention on $R_1$.

%++++++++++++++++++++++++++++++++++++

\subsection{ Identifying Propensity Scores with Different Partial Orders} 
\label{subsec:partial-fix}

In the examples discussed so far, all missingness indicators were intervened on according to a \emph{partial order} defined on the set $\{R_i \in R\}$. A procedure for target law identification may then aim at exploring the space of all possible partial orders. This effectively entails trying out all possible combinations of parallel and sequential application of the g-formula. \cite{bhattacharya19mid} showed that such a procedure remains incomplete, meaning that it fails to recognize a class of missing data models where the target law is indeed identified. They took an alternative view of the target law identification problem, where each propensity score $p(r_i \mid \pa_{\cal G}(r_i))$ may be identified separately, using a potentially distinct partial order of intervention operations. This entails considering subproblems where selection bias, hidden variables, or both, are introduced, even if these complications are absent in the original problem. 

Consider the missing data DAG in Fig.~\ref{fig:partial-order}(a). According to Proposition~\ref{prop:miss-id}, the target law $p(l^{(1)}_1, l^{(1)}_2, l^{(1)}_3)$ $\coloneqq p(l_1, l_2, l_3 \ \Vert \ r_1=1, r_2=1, r_3=1)$ is equivalent to the following:
{\small
\begin{align}
	p(l_1, l_2, l_3 \ \Vert \ r=1) \! = \! \left. \frac{ p(l_1, l_2, l_3, r_1,
		r_2, r_3) }{ p(r_1 | l_2^{(1)}) \times  p(r_2
		| l^{(1)}_1, l_3^{(1)}, r_1) \times p(r_3 | l_2^{(1)}, r_1)  } \right|_{{r} = {1}}. 
	\label{eq:ex-partial-order-gformula}
\end{align}%
}
We can easily use the encoded independence restrictions, along with the consistency facts in missing data, to identify the propensity scores of $R_2$ and $R_3$ evaluated at $R=1$. These restrictions are: $R_2 \ci R_3 \mid R_1, L^{(1)}_1, L^{(1)}_3$ and $R_3 \ci R_2 \mid R_1, L^{(1)}_2$. We have,  
{\small
\begin{align}
	p(r_2 \mid \pa_{\mathcal{G}_{{m}}}(r_2))\vert_{r=1}  &= 
	%\! p(r_2 | l^{(1)}_1, l_3^{(1)}, r_1)\vert_{r=1} = 
	 p(r_2 = 1 \mid
	l_1, l_3, r_1 = 1,  r_3 = 1), \label{ex:partial-seq-r2} \\
	p(r_3 \mid \pa_{\mathcal{G}_{{m}}}(r_3))\vert_{r=1} &= 
	%p(r_3 \mid l_2^{(1)}, r_1)\vert_{r=1} = 
	p(r_3 = 1 \mid
	l_2, r_1 = 1, r_2 = 1). \label{ex:partial-seq-r3}
\end{align}%
}
The propensity score of $R_1$ however, is not immediately identified since $R_1 \not\ci R_2 \mid L^{(1)}_2.$ A natural question to ask is whether this propensity score is identified from a kernel distribution where $R_2$ and $R_3$ are simultaneously intervened on. If so, then we can follow a similar argument as in the previous example with the two-step sequential procedure. Performing parallel interventions on $R_2$ and $R_3$ yield the following kernel distribution,
{\small
\begin{align*}
	p(l^{(1)}_1, l_2, l_3, r_1 \ \Vert \ r_2=r_3=1)
	\! = \! \frac{ p(l_1^{(1)}, l_1, l_2, l_3, r_1, r_2 = r_3 = 1) }{  p(r_2 = 1 | l^{(1)}_1, l_3, r_1, r_3 = 1) \! \times \! p(r_3 = 1 | l_2,r_1, r_2 = 1)}.
\end{align*}%
}
Following the second step of the sequential procedure, we need to first identify the propensity score of $R_1$ via the above kernel, i.e., identifying $p(r_1 \mid l_2 \ \Vert \ r_2 = 1, r_3=1)$ which by Bayes rule can be obtained from $p(r_1, l_2 \ \Vert \ r_2 = 1, r_3=1)$. Unfortunately, the above kernel $p(. \ \Vert \ r_2=1, r_3=1)$ exhibits selection on $R_1$. Consequently, the two-stage sequential strategy in the previous example fails to identify the target law in this current example. Nevertheless, the target law is still identifiable using a more involved argument outlined below. Since the propensity scores of $R_2$ and $R_3$ are identified via (\ref{ex:partial-seq-r2}) and (\ref{ex:partial-seq-r3}), all we need to do to identify the target law is to find a way to identify the propensity score of $R_1$, $p(r_1 \mid l^{(1)}_2)$.

\begin{figure}[!t] 
	\begin{center}
		\scalebox{0.8}{
			\begin{tikzpicture}[>=stealth, node distance=1.4cm]
				\tikzstyle{format} = [thick, circle, minimum size=1.0mm, inner sep=0pt]
				\tikzstyle{square} = [draw, thick, minimum size=4.5mm, inner sep=3pt]
				\tikzstyle{format2} = [draw, circle, fill=gray!10, minimum size=2.5mm, inner sep=1.5pt]
				
				\begin{scope}[xshift=0cm]
					\path[->, thick]
					node[format] (x11) {$L^{(1)}_1$}
					node[format, right of=x11, xshift=0.25cm] (x21) {$L^{(1)}_2$}
					node[format, right of=x21, xshift=0.25cm] (x31) {$L^{(1)}_3$}
					node[format, below of=x11] (r1) {$R_1$}
					node[format, below of=x21] (r2) {$R_2$}
					node[format, below of=x31] (r3) {$R_3$}
					node[format, below of=r1] (x1) {$L_1$}
					node[format, below of=r2] (x2) {$L_2$}
					node[format, below of=r3] (x3) {$L_3$} 
					
					(x21) edge[blue] (x31) 
					
					(x11) edge[blue] (r2)
					(x21) edge[blue] (r1)
					(x21) edge[blue] (r3)
					(x31) edge[blue] (r2)
					
					(r1) edge[blue] (r2) 
					(r1) edge[blue, bend right] (r3)
					
					(x11) edge[gray, bend right=30] (x1)
					(x21) edge[gray, bend right=30] (x2)
					(x31) edge[gray, bend left=30] (x3)
					(r1) edge[gray] (x1)
					(r2) edge[gray] (x2)
					(r3) edge[gray] (x3)
					
					node[format, below of=x2, yshift=0.5cm] (a) {(a)} ;
				\end{scope}
				\begin{scope}[xshift=5.25cm, yshift=0cm]
					\path[->, thick]
					node[format] (x11) {}
					node[format, right of=x11, xshift=0.25cm] (x21) {$L^{(1)}_2$}
					node[format, right of=x21, xshift=0.25cm] (x31) {$L^{(1)}_3$}
					node[format, below of=x11] (r1) {$R_1$}
					node[format, below of=x21] (r2) {$R_2$}
					node[format, below of=x31] (r3) {$R_3$}
					node[format, below of=r2] (x2) {$L_2$}
					node[format, below of=r3] (x3) {$L_3$} 
					
					(x21) edge[blue] (x31) 
					(x21) edge[blue] (r1)
					(x21) edge[blue] (r3)
					(x31) edge[blue] (r2)
					(r1) edge[blue] (r2) 
					(r1) edge[blue, bend right] (r3)
					
					(x21) edge[gray, bend right=30] (x2)
					(x31) edge[gray, bend left=30] (x3)
					(r2) edge[gray] (x2)
					(r3) edge[gray] (x3)
					
					node[format, below of=x2, yshift=0.5cm] (b) {(b)} ;
				\end{scope}
				\begin{scope}[xshift=10.5cm, yshift=0cm]
					\path[->, thick]
					node[format] (x11) {}
					node[format, right of=x11, xshift=0.5cm] (x21) {$L^{(1)}_2 = L_2$}
					node[format, right of=x21, xshift=1.cm] (x31) {$L^{(1)}_3 = L_3$}
					node[format, below of=x11] (r1) {$R_1$}
					node[square, below of=x21] (r2) {$r_2 = 1$}
					node[square, below of=x31] (r3) {$r_3 = 1$}
					node[format, below of=r1] (x1) {}
					node[format, below of=r2] (x2) {\gray{$L_2$}}
					node[format, below of=r3] (x3) {\gray{$L_3$}}
					
					(x21) edge[blue] (x31) 
					(x21) edge[blue] (r1)
					(x21) edge[gray, dashed, bend right=45] (x2)
					(x31) edge[gray, dashed, bend left=45] (x3)
					(r2) edge[gray, dashed] (x2)
					(r3) edge[gray, dashed] (x3)
					
					node[format, below of=x2, xshift=1.6cm, yshift=0.5cm] (c) {(c)} ;
				\end{scope}
			\end{tikzpicture}
		}
		 \vspace{-1.5cm}
		\caption{(a) An example m-DAG corresponding to a model where intervention on a single missingness indicator entails following a partial order of interventions; (b), (c) Selection bias on $R_1$ is avoidable by following a partial order intervention on a graph induced by projecting out $L^{(1)}_1, L_1$.}
		\label{fig:partial-order}
	\end{center}
\end{figure}
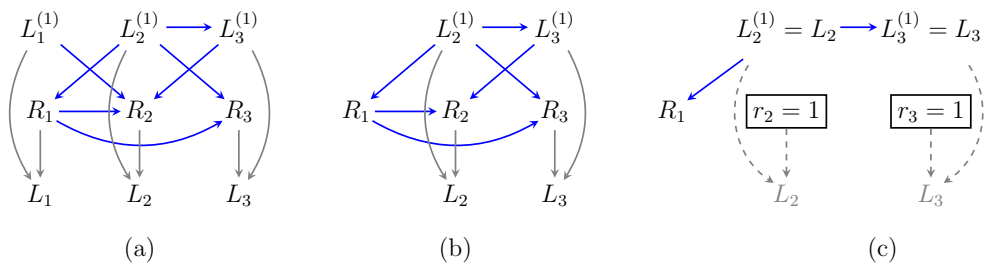

Consider the marginal distribution $\sum_{l_1, l^{(1)}_1}  p(l^{(1)}, r, l)$, where the full law factorizes according to the m-DAG in Fig.~\ref{fig:partial-order}(a). The non-deterministic portion of this marginal distribution factorizes as: $p(l^{(1)}_2, l^{(1)}_3, r_1, r_2, r_3) =  p(l_2^{(1)}) \times p(l_3^{(1)} | l^{(1)}_2)  \times p(r_1 | l_2^{(1)}) \times  p(r_2 | r_1, l_3^{(1)})  \times p(r_3 | l_2^{(1)}, r_1). $
%{\small
%\begin{align}
%	p(l^{(1)}_2, l^{(1)}_3, r_1, r_2, r_3) =  p(l_2^{(1)}) \times p(l_3^{(1)} \mid l^{(1)}_2)  \times p(r_1 \mid l_2^{(1)}) \times  p(r_2 \mid r_1, l_3^{(1)})  \times p(r_3 \mid l_2^{(1)}, r_1).  
%	\label{eq:projecting-out-l1}
%\end{align}%
%}
The above factorization is Markov relative to the m-DAG in Fig.~\ref{fig:partial-order}(b), where $L_1, L^{(1)}_1$ are projected out (or in other words, treated as hidden variables) from Fig.~\ref{fig:partial-order}(a). Aside from absence of $p(l^{(1)}_1)$ and $p(l_1 \mid r_1, l^{(1)}_1)$ factors, the difference between the above factorization and the one from original m-DAG in Fig.~\ref{fig:partial-order}(a) is the difference in propensity score of $R_2$, i.e., $p(r_2 \mid r_1, l^{(1)}_1, l^{(1)}_3)$ vs.~$p(r_2 \mid r_1, l^{(1)}_3)$. We now illustrate that this change in the propensity score of $R_2$, which we call a \emph{pseudo-propensity score} for $R_2$, can overcome the selection bias issue on $R_1$ that resulted from an intervention on $R_2$. 

The pseudo-propensity score of $R_2$, denoted by $\widetilde{p}(r_2 \mid r_1, l^{(1)}_3)$, is easily identified via $\widetilde{p}(r_2 \mid r_1, l_3, r_3=1)$ since $R_2 \ci R_3 \mid R_1, L^{(1)}_3$ in Fig.~\ref{fig:partial-order}(b). Now to identify the propensity score of $R_1$, we proceed as follows. Using the pseudo-propensity score of $R_2$ and the propensity score of $R_3$, we obtain the following kernel distribution: 
{\small
\begin{align*}
	\widetilde{p}(l_2, l_3, r_1 \ \Vert \ r_2 = 1, r_3 = 1) 
	=  \frac{{p}(l_2, l_3, r_1, r_2 = 1, r_3 = 1)}{\widetilde{p}(r_2 = 1 \mid r_1, l_3, r_3 = 1) \times p(r_3 = 1 \mid l_2, r_1, r_2 = 1)}.  
\end{align*}%
}
The above kernel $\widetilde{p}(. \ \Vert \ r_2 = 1, r_3 = 1)$, which factorizes according to the m-DAG in Fig.~\ref{fig:parallel}(c), is a direct function of observed data law without inducing any selection bias on $R_1$. Consequently, we can identify the propensity score of $R_1$, $p(r_1 \mid l_2  \ \Vert \ r_2 = 1, r_3 = 1)$ using the kernel $\widetilde{p}(. \ \Vert \ r_2 = 1, r_3 = 1)$. Using Bayes rule, we have $p(r_1 \mid l_2  \ \Vert \ r_2 = 1, r_3 = 1) = \sum_{l_3} p(r_1, l_2, l_3  \ \Vert \ r_2 = 1, r_3 = 1) / \sum_{l_3, r_1} p(r_1, l_2, l_3  \ \Vert \ r_2 = 1, r_3 = 1)$. This concludes that the target law in (\ref{eq:ex-partial-order-gformula}) is identified, since we showed each term in the denominator is identified. 

The idea of  focusing on identifying  each propensity score separately  leads to new identification strategies. Specifically, in the above example, the propensity score of $R_1$ was identified only after treating $L^{(1)}_1$ as a hidden variable and marginalizing it out from the original distribution. \cite{bhattacharya19mid} developed a general procedure based on these observations which has significantly  narrowed the identifiability gap in graphical models of missing data. Their procedure generalizes the notion of finding a partial order of interventions on the set $\{R_i \in R\}$ to finding partial orders of interventions for each individual $R_k \in R$ by exploring subproblems where a set of variables are treated as hidden or unmeasured. 

In the above example, the partial orders for $R_2$ and $R_3$ are trivial -- the corresponding propensity scores are immediately identified from the observed distribution. We can summarize the identification procedure for obtaining the propensity score of $R_1$ via the following partial order executed in a graph where $\{L_1^{(1)}, L_1\}$ are treated as hidden variables, which we will denote as ${\cal G}_m(V \setminus \{L_1^{(1)}, L_1\})$: the partial order of interventions for $R_1$ can be summarized via $\big\{ \{ I_{r_2}, I_{r_3}  \} < I_{r_1}  \big\} \text{ in } {{\cal G}_m (V \setminus \{L_1^{(1)}, L_1\} )}.$ That is, intervention on $R_1$ must occur after interventions on $R_2$ and $R_3$ in a graph where $L^{(1)}_1$ and $L_1$ are marginalized out, and as mentioned earlier interventions on $R_2$ and $R_3$ are incompatible. 

%++++++++++++++++++++++++++++++++++++

\subsection{Intervention on Variables Outside of R} 
\label{subsec:outside-R-fix}

A feature common to all previous examples is that all propensity scores were obtained via partial orders of interventions defined only on the missingness indicators. This  however, is also not always sufficient. In general the propensity score of $R_k$ might be identified only by intervening on variables outside of $R$, including variables in $L^{(1)}$ that become observed after intervening or conditioning on relevant elements in $R$. As an example, consider the m-DAG in Fig~\ref{fig:UAI}(a) where $L_3$ is fully observed. This graph was considered in \cite{bhattacharya19mid} as a generalization of a model described in \cite{shpitser15missing}. According to Proposition~\ref{prop:miss-id}, we can write down the target law $p(l_1^{(1)}, l_2^{(1)}, l_4^{(1)}, l_3)$ as follows: 
{\small
\begin{align}
	p(l_1, l_2, l_3, l_4 \ \Vert \ r=1) = \left. \frac{ p(l_1, l_2, l_3, l_4, r_1,
		r_2, r_4) }{ p(r_1 | l_2^{(1)}, l_4^{(1)}, r_2) \times  p(r_2
		| l^{(1)}_1, l_4^{(1)}) \times p(r_4 | l_3)  } \right|_{{r} = {1}}. 
	\label{eq:ex-uai-gformula}
\end{align}% 
}
The propensity score of $R_4$, $p(r_4 \mid l_3)$, is a direct function of observed data. The propensity score of $R_1$, $p(r_1 \mid r_2, l^{(1)}_2, l^{(1)}_4)$ evaluated at $R=1$ is also a function of observed data via $p(r_1 = 1  \mid r_2=1, l_2, l_4, r_4=1)$ since $R_1 \ci R_4 \mid R_2, L^{(1)}_2, L^{(1)}_4$. Thus, identification of the target law relies on whether the propensity score of $R_2$, $p(r_2 \mid l^{(1)}_1, l^{(1)}_4)$ is identified or not, which is not immediately clear since $R_2 \not\ci R_1 \mid L^{(1)}_1, L^{(1)}_4$. We now outline a procedure to identify this propensity score. 

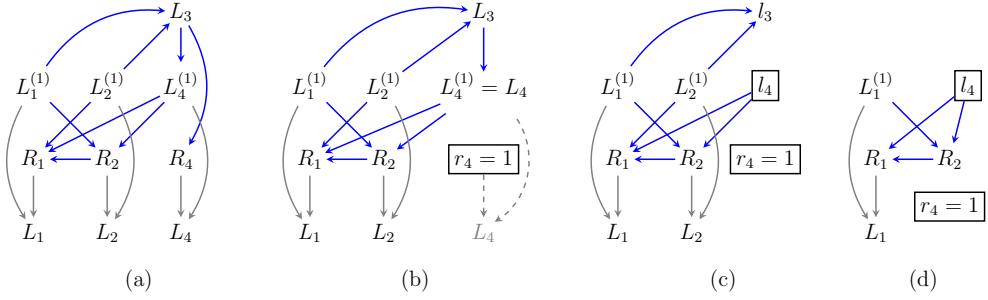
\begin{figure}[!t]
	\begin{center}
		\scalebox{0.7}{
			\begin{tikzpicture}[>=stealth, node distance=1.4cm]
				\tikzstyle{format} = [thick, circle, minimum size=1.0mm, inner sep=0pt]
				\tikzstyle{square} = [draw, thick, minimum size=1mm, inner sep=3pt]
				\begin{scope}[xshift=0cm]
					\path[->,  thick]
					node[format] (x11) {$L^{(1)}_1$}
					node[format, right of=x11] (x21) {$L^{(1)}_2$} 
					node[format, right of=x21] (x41) {$L^{(1)}_4$} 
					node[format, above of=x41] (x31) {$L_3$}
					node[format, below of=x11] (r1) {$R_1$} 
					node[format, below of=x21] (r2) {$R_2$} 
					node[format, below of=x41] (r4) {$R_4$} 
					node[format, below of=r4] (x4) {$L_4$}
					node[format, below of=r1] (x1) {$L_1$} 
					node[format, below of=r2] (x2) {$L_2$} 
					
					(x11) edge[blue, bend left=30] (x31) 			
					(x21) edge[blue] (x31)	
					(x31) edge[blue] (x41)		
					(x41) edge[blue] (r2)			
					(x41) edge[blue] (r1)
					(r2)   edge[blue] (r1)		
					(x31) edge[blue, bend left] (r4)	
					(r1) edge[gray] (x1)
					(r2) edge[gray] (x2)	
					(r4) edge[gray] (x4)	
					(x41) edge[gray, bend left=30] (x4)	
					(x11) edge[blue, bend right=0] (r2)
					(x21) edge[blue, bend left=0] (r1)	
					(x11) edge[gray, bend right=30] (x1)	
					(x21) edge[gray, bend left=30] (x2)		
					node [below of=x2, xshift=0.6cm, yshift=0.5cm] {(a)} ;
				\end{scope}
				\begin{scope}[xshift=5.25cm]
					\path[->,  thick]
					node[format] (x11) {$L^{(1)}_1$}
					node[format, right of=x11] (x21) {$L^{(1)}_2$} 
					node[format, right of=x21, xshift=0.5cm] (x41) {$L^{(1)}_4 = L_4$} 
					node[format, above of=x41] (x31) {$L_3$}
					node[format, below of=x31, yshift=0.2cm] (x0) {} 
					node[format, below of=x11] (r1) {$R_1$} 
					node[format, below of=x21] (r2) {$R_2$} 
					node[square, below of=x41] (r4) {$r_4=1$} 
					node[format, below of=r4] (x4) {\gray{$L_4$}}
					node[format, below of=r1] (x1) {$L_1$} 
					node[format, below of=r2] (x2) {$L_2$} 
					
					(x11) edge[blue, bend left=30] (x31) 			
					(x21) edge[blue] (x31)	
					(x31) edge[blue] (x0)		
					(x41) edge[blue] (r2)			
					(x41) edge[blue] (r1)
					(r2)   edge[blue] (r1)		
					(r1) edge[gray] (x1)
					(r2) edge[gray] (x2)	
					(r4) edge[gray, dashed] (x4)	
					(x41) edge[gray, bend left=45, dashed] (x4)	
					(x11) edge[blue, bend right=0] (r2)
					(x21) edge[blue, bend left=0] (r1)	
					(x11) edge[gray, bend right=30] (x1)	
					(x21) edge[gray, bend left=30] (x2)		
					node [below of=x2, xshift=0.6cm, yshift=0.5cm] {(b)} ;
				\end{scope}
				\begin{scope}[xshift=11.1cm, yshift=0cm]
					\path[->,  thick]
					node[format] (x11) {$L^{(1)}_1$}
					node[format, right of=x11] (x21) {$L^{(1)}_2$} 
					node[square, right of=x21, xshift=0cm] (x41) {$l_4$} 
					node[format, above of=x41] (x31) {$l_3$}
					node[format, below of=x11] (r1) {$R_1$} 
					node[format, below of=x21] (r2) {$R_2$} 
					node[square, below of=x41, xshift=0cm] (r4) {$r_4=1$} 
					node[format, below of=r4] (x4) {}
					node[format, below of=r1] (x1) {$L_1$} 
					node[format, below of=r2] (x2) {$L_2$} 
					
					(x11) edge[blue, bend left=30] (x31) 			
					(x21) edge[blue] (x31)	
					(x41) edge[blue] (r2)			
					(x41) edge[blue] (r1)
					(r2)   edge[blue] (r1)		
					(r1) edge[gray] (x1)
					(r2) edge[gray] (x2)	
					(x11) edge[blue, bend right=0] (r2)
					(x21) edge[blue, bend left=0] (r1)	
					(x11) edge[gray, bend right=30] (x1)	
					(x21) edge[gray, bend left=30] (x2)		
					node [below of=x2, xshift=0.6cm, yshift=0.5cm] {(c)} ;
				\end{scope}
				\begin{scope}[xshift=16.cm, yshift=0cm]
					\path[->,  thick]
					node[format] (x11) {$L^{(1)}_1$}
					node[format, right of=x11] (x21) {} 
					node[format, right of=x21] (x31) {}
					node[square, right of=x11, xshift=0.35cm] (x41) {$l_4$} 
					node[format, below of=x11] (r1) {$R_1$} 
					node[format, below of=x21] (r2) {$R_2$} 
					node[square, below of=r2, yshift=0.5cm] (r4) {$r_4=1$} 
					node[format, below of=r4] (x4) {}
					node[format, below of=r1] (x1) {$L_1$} 
					node[format, below of=r2] (x2) {} 
					
					(x41) edge[blue] (r2)			
					(x41) edge[blue] (r1)
					(r2)   edge[blue] (r1)		
					(r1) edge[gray] (x1)
					(x11) edge[blue, bend right=0] (r2)
					(x11) edge[gray, bend right=30] (x1)		
					node [below of=x1, xshift=0.9cm, yshift=0.5cm] {(d)} ;
				\end{scope}
			\end{tikzpicture}
		}
	\end{center}
	\vspace{-1cm}
	\caption{(a) An example m-DAG corresponding to a model where variables besides $R$s are required to be intervened on in order to identify the propensity scores; (b) A conditional m-DAG where  $R_4$ is intervened on; (c) A conditional m-DAG where $L_4$ is intervened on after an intervention on $R_4$; (d) A conditional m-DAG where $L_2^{(1)}$ and $L_3$ are marginalized out from the kernel $p(. \ \Vert\  r_4=1, l_4)$.}
	\label{fig:UAI}
\end{figure}

We first intervene on $R_4$, i.e., dropping the term $p(r_4 \mid l_3)$ from the original factorization, yielding the kernel below which is Markov relative to the graph in Fig.~\ref{fig:UAI}(b), 
{\small
\begin{align*}
	p(l_1^{(1)}, l_2^{(1)}, l_3, l_4, r_1, r_2, l_1, l_2 \ \Vert \ r_4=1)&=
	\frac{ p(l_1^{(1)},  l_2^{(1)}, l_3, l_4, l_1, l_2, r_1, r_2, r_4 = 1) }{ p(r_4 = 1 \mid l_3)}. 
\end{align*}%
}
After intervening on $R_4$, $L^{(1)}_4$ is fully observed. We then intervene on $L_4,$ i.e., dropping the term $p(l_4 \mid l_3 \ \Vert \ r_4 = 1)$ from the above kernel and obtain the following kernel which is Markov relative to the graph in Fig.~\ref{fig:UAI}(c),
{\small
\begin{align*}
	p(l_1^{(1)}, l_2^{(1)}, l_3, r_1, r_2, l_1, l_2 \ \Vert \ r_4=1, l_4)
	%&= \frac{p(l_1^{(1)}, l_2^{(1)}, l_3, l_4, r_1, r_2, l_1, l_2 \ \Vert \ r_4=1)}{p(l_4 \mid l_3 \ \Vert \ r_4 = 1)} \\
	=  \frac{ p(l_1^{(1)}, l_2^{(1)}, l_3, l_4, l_1, l_2, r_1, r_2, r_4 = 1) }{ p(r_4 = 1 \mid l_3) \times p(l_4 \mid l_3, r_4 = 1)}.
\end{align*}%
}
{Note that in this example, marginalizing out $L_3$ is equivalent to intervening on $L_3$.  We can thus}
%We can
safely marginalize out $L^{(1)}_2, L_2,$ and $L_3$ from the above expression
%\footnote{In this example, marginalizing out $l_3$ is equivalent to intervening on $l_3$.} 
without changing the propensity score of $R_2$, yielding a new kernel that is Markov relative to the graph in Fig.~\ref{fig:UAI}(d),
{\small
\begin{align*}
	p(l_1^{(1)}, r_1, r_2, l_1 \ \Vert \ r_4=1, l_4) 
	=  \sum_{l_3} \frac{ p(l_1^{(1)}, l_3, l_4, r_1, r_2, r_4 = 1) }{ p(r_4 = 1 \mid l_3) \times p(l_4 \mid l_3, r_4 = 1)}.
\end{align*}%
}
The propensity score of $R_1$ in the above kernel is different than the one in original factorization of the m-DAG in Fig.~\ref{fig:UAI}(a).  We refer to this as the pseudo-propensity score of $R_1$ and denote it by $\widetilde{p}(r_1 \mid r_2 \ \Vert \ r_4, l_4).$  We now intervene on $R_1$ by dropping this term from the above kernel, yielding  
{\small
\begin{align*}
	p(l_1, r_2 \ \Vert \  r_1=1, r_4=1, l_4) 
	=  \frac{p(l_1, r_1=1, r_2 \ \Vert \  r_4=1, l_4)}{\widetilde{p}(r_1 = 1 \mid  r_2 \ \Vert \  r_4=1, l_4)}.
\end{align*}%
}
The desired propensity score of $R_2$ (which remains invariant despite all previous operations as its direct causes are still present in Fig.~\ref{fig:UAI}(d)) is then identified in the above kernel as $p(l_1, r_2 \ \Vert \  r_1=1, r_4=1, l_4) / \sum_{r_2} p(l_1, r_2 \ \Vert \  r_1=1, r_4=1, l_4).$ Since all the propensity scores are identified, then the target law in (\ref{eq:ex-uai-gformula}) is identified as well. 

In the above example, the partial orders for $R_1$ and $R_4$ are trivial -- the corresponding propensity scores are immediately identified from the observed distribution. We can summarize the identification procedure for obtaining the propensity score of $R_2$ via the following partial order executed in a graph where $R_4$ then $L_4$ are intervened on and $\{L_2^{(1)}, L_2, L_3\}$ are treated as hidden variables, which we will denote as ${\cal G}_m(V \setminus \{L^{(1)}_2, L_2, L_3\}, \{r_4, l_4\})$: the partial order of interventions for $R_2$ can be summarizes via $\big\{ I_{r_4} < I_{l_4} <  I_{r_1}  < I_{r_2}  \big\} \text{ in }  {\cal G}_m(V \setminus \{L^{(1)}_2, L_2, L_3\}, \{r_4, l_4\}).$ That is, intervention on $R_2$ must occur after interventions on $R_4, L_4$ and $R_1$ in a graph where $L^{(1)}_2, L_2,$ and $L_3$ are marginalized out, and as mentioned earlier interventions on $R_1$ and $R_4$ are incompatible.

\subsection{A Unifying Identification Procedure}
\label{sec:unified}

\cite{bhattacharya19mid} proposed a procedure for target law identification that combines all the ideas discussed above. It proceeds as follows. For each missingness indicator $R_k \in R$,  it proceeds to identify its propensity score  $p({r}_k | \pa_{{\cal G}_{{m}}}({r}_k))%|_{R=1}
$ {evaluated at $R=1$}. It does so by checking if $R_k$ is conditionally independent (given its parents) of the corresponding missingness indicators of its counterfactual parents. If this is the case, the propensity score is identified by a simple conditional independence argument (d-separation). Otherwise, the procedure checks if this condition holds in any intervention distribution where a subset of missingness indicators are  intervened on, in either the original model or marginals of the model where the direct causes of $R_k$ are still preserved.  If the procedure succeeds in identifying the propensity score for each missingness indicator in this manner, then the target law is declared as being identified. 

Necessary conditions for target law identification have been discussed in the literature. A well-known result states that if an underlying variable causes its own missingness status, known as a self-censoring mechanism or nonignorable mechanism, then the target law is provably not identified. This means we can construct two missing data models that differ in target law distributions but both map to the same observed data distribution. The graphical structure simply corresponds  to existence of an edge of the form $L^{(1)}_k \rightarrow R_k$ in the m-DAG. \cite{nabi2023testability, guo2023sufficient} have also shown that the so-called ``criss-cross" structure %is an impediment to 
{prevents}
target law identification. These structures involve a pair of variables $L^{(1)}_i, L^{(1)}_j$ that are directly connected as $L^{(1)}_i \rightarrow L^{(1)}_j$ or $L^{(1)}_i \leftarrow L^{(1)}_j$ and these edges exist simultaneously: $L^{(1)}_i \rightarrow R_j \leftarrow R_i \leftarrow L^{(1)}_j$. On the other hand, \cite{nabi20completeness} have provided sufficient conditions under which a target law is identified. They show under the absence of self-censoring edges and so-called ``colluder" structures, the target law is identified.  A colluder is a special type of collider where there exists $R_i, R_j \in R$ such that $L^{(1)}_i \rightarrow R_j \leftarrow R_i.$ Finding necessary and sufficient conditions (a sound and complete algorithm) for target law identification remains an open problem. 

In many applied problems, some variables are not just missing but completely unobserved. We can use missing data DAG models with hidden variables to encode the presence of unmeasured confounders and generalize the aforementioned identification strategies to such settings. 
{
We discuss %a unifying procedure for identification that takes advantage of all of the observations described above in Appendix \ref{sec:unified}, and
extensions to identification in the presence of both missing data and hidden variables in Appendix~S2.  %\ref{sec:missing_data_w_hidden_vars}.
}

\section{Discussion}
\label{sec:diss}
%####################################################

We have shown how unique features of missing data models represented by DAGs allow identification in seemingly counterintuitive situations, by taking advantage of Markov restrictions linking missingness indicators and underlying counterfactual variables.  For instance, {the target law is identified in the bivariate permutation model (Fig.~\ref{fig:permutation}(a)), but not in the analogous hidden variable causal model (Fig.~\ref{fig:iterate}(a))}. Similarly, the target law is identified in the bivariate block-parallel model, shown in Fig.~\ref{fig:missing-vs-causal}(a), by (\ref{eqn:block-parallel-id}), but not in the analogous  causal model shown in Fig.~\ref{fig:missing-vs-causal}(b).

We described how graphical missing data models are a special case of hidden variable graphical causal models, where hidden variables are replaced by counterfactual variables, which are sometimes observed.  It is this partial observability which allowed identification to be derived.  Just as in missing data models, observed variables in causal models are derived from counterfactual variables by means of consistency, and thus may be viewed as causes (parents) of observed variables in a particular type of causal graphs.  This leads to a natural question: would expressing causal models via graphs which make the relationship between observed and counterfactual variables explicit by placing counterfactuals as vertices on the graph yield new types of identification results? We now show the answer is `no' without further very strong assumptions.

%We now study the causal analogues of the permutation and block-parallel models, and discuss an approach to extending missing data identification strategies to causal models. 

{%\color{red}
%	Rather than placing hidden variables on the graph, we may consider placing counterfactual versions of observed variables on the graph instead, with the understanding that a pair counterfactuals $L^{(1)},L^{(0)}$ corresponding to $L$ (for some treatment $A$) may potentially be easier to take into account in identification compared to hidden variables.  This is because the counterfactual pair $L^{(1)},L^{(0)}$ is partiallly observed, due to the consistency property, unlike a purely hidden variable $U$.  However, this is also insufficient for obtaining identification.
 %the causal inference variation of the permutation model

Consider an elaboration of the graph in Fig.~\ref{fig:iterate}(a) shown in Fig.~\ref{fig:m-DAG-SWIG}(a), where each $U_i$ is replaced by two counterfactuals $L_i^{(1)},L_i^{(0)}$ for $i=1,2$. %are added for each
%$L_i$, 
These counterfactuals are associated, represented by the addition of a common parent $\epsilon_i$.
These counterfactuals are shown as parents of $L_i$, to indicate that $L_i$ is determined from these counterfactuals by the value of $R_i$ via consistency:
$L_i = L_i^{(1)} R_i + L_i^{(0)} (1 - R_i)$.  Furthermore, to preserve the structure of confounding in Fig.~\ref{fig:iterate}(a), $L_2^{(1)},L_2^{(0)}$ are children of $L_1^{(1)},L_1^{(0)}$.  In fact, the graphs in Fig.~\ref{fig:m-DAG-SWIG}(a) and Fig.~\ref{fig:iterate}(a) represent the same causal model.

% , which is drawn in such a way that all counterfactual versions of observed variables all occur prior to every observed variable.  In this graph, variables
% $L_1$ and $L_2$ each have \emph{two} counterfactuals on the graph, with each counterfactual $L_i^{(r_i)}$ inheriting
% all incoming and outgoing edges from 
% %the corresponding 
% each missing data counterfactual $L_i^{(1)}$ Fig.~\ref{fig:iterate}(a).  In addition,
% $L_i^{(1)}, L_i^{(0)}$ are associated, as indicated by the red bidirected edge between them.
% Fig.~\ref{fig:m-DAG-SWIG}(a) yields a causal model that simply asserts all independence statements in the graph.

% Note the causal graph in figure 8a is identical to the graph in figure 1a except with out the counterfactuals added to the graph In  fact  8a and 1a represent the same causal model . Counterfactuals need not be represented on the causal graph (as in Fig 1a) since all counterfactuals have a single child and only variables that are a common cause of 2 or more variables need to be on a causal  graph
}

\begin{figure}
\centering
\scalebox{0.75}{
\begin{tikzpicture}[>=stealth, node distance=1.5cm]
    \tikzstyle{format} = [thick, circle, minimum size=1.0mm, inner sep=0pt]
    \tikzstyle{square} = [draw, thick, minimum size=4.5mm, inner sep=3pt]

    \begin{scope}[xshift=0cm, yshift=0cm]
        \path[->, thick]
        node[format] (x11) {$L^{(1)}_1$}
        node[format, above of=x11, yshift=0cm, red, xshift=0.0cm] (e1){$\epsilon_1$}
        node[format, above of=e1] (x10) {$L^{(0)}_1$}

        % node[format, right of=x11, xshift=0cm, red] (u1) {$U_1$}
        % node[format, right of=u1, xshift=0cm, red] (u2) {$U_2$}
        
        node[format, right of=x11, xshift=1cm] (x21) {$L^{(1)}_2$}
        node[format, above of=x21, yshift=0cm, red, xshift=0cm] (e2){$\epsilon_2$}
        node[format, above of=e2] (x20) {$L^{(0)}_2$}
    
        node[format, below of=x11, yshift=-0.cm] (r1) {$R_1$}
        node[format, below of=x21, yshift=-0.cm] (r2) {$R_2$}
        node[format, below of=r1, yshift=-0.cm] (x1) {$L_1$}
        node[format, below of=r2, yshift=-0.cm] (x2) {$L_2$}

        (e1) edge[->, red] (x11)
        (e1) edge[->, red] (x10)
        (e2) edge[->, red] (x21)
        (e2) edge[->, red] (x20)

        (x10) edge[blue] (x20)
        (x10) edge[blue] (x21)
        (x11) edge[blue] (x20)
        (x11) edge[blue] (x21)
        
        (r1) edge[blue] (r2)
        (x1) edge[blue] (r2)
        
        (x11) edge[gray, bend right=25] (x1)
        (x21) edge[gray, bend left=25] (x2)
        (x10) edge[gray, bend right=35] (x1)
        (x20) edge[gray, bend left=35] (x2)
        (r1) edge[gray] (x1)
        (r2) edge[gray] (x2)

        (x21) edge[blue] (r1)
        (x20) edge[blue] (r1)

        node[format, below of=x1, xshift=1.2cm, yshift=0.5cm] (a) {(a)}; % {\small Permutation}} ;

         \node[draw, dashed, fit=(x11)(e1)(x10), inner sep=4pt, "$U_1$"]{};
         \node[draw, dashed, fit=(x21)(e2)(x20), inner sep=4pt, "$U_2$"] {};
    \end{scope}

    \begin{scope}[xshift=6cm, yshift=0cm]
        \path[->, thick]
        node[format] (x11) {$L^{(1)}_1$}
        node[format, above of=x11, yshift=0cm, red, xshift=0.0cm] (e1){$\epsilon_1$}
        node[format, above of=e1] (x10) {$L^{(0)}_1$}

        % node[format, right of=x11, xshift=0cm, red] (u1) {$U_1$}
        % node[format, right of=u1, xshift=0cm, red] (u2) {$U_2$}
        
        node[format, right of=x11, xshift=1cm] (x21) {$L^{(1)}_2$}
        node[format, above of=x21, yshift=0cm, red, xshift=0cm] (e2){$\epsilon_2$}
        node[format, above of=e2] (x20) {$L^{(0)}_2$}
    
        node[format, below of=x11, yshift=-0.cm] (r1) {$R_1$}
        node[square, below of=x21, yshift=-0.cm] (r2) {$R_2 = 1$}
        node[format, below of=r1, yshift=-0.cm] (x1) {$L_1$}
        node[format, below of=r2, yshift=-0.cm] (x2) {$L_2 = L^{(1)}_2$}

        node[format, below of=r2, yshift=0.2cm] (x2t) {}

        (e1) edge[->, red] (x11)
        (e1) edge[->, red] (x10)
        (e2) edge[->, red] (x21)
        (e2) edge[->, red] (x20)

        (x10) edge[blue] (x20)
        (x10) edge[blue] (x21)
        (x11) edge[blue] (x20)
        (x11) edge[blue] (x21)
        
        % (r1) edge[blue] (r2)
        % (x1) edge[blue] (r2)
        
        (x11) edge[gray, bend right=25] (x1)
        (x21) edge[black, very thick, bend left=60] (x2t)
        (x10) edge[gray, bend right=35] (x1)
        % (x20) edge[gray, bend left=35] (x2)
        (r1) edge[gray] (x1)
        (r2) edge[gray] (x2t)

        (x21) edge[blue] (r1)
        (x20) edge[blue] (r1)

        node[format, below of=x1, xshift=1.2cm, yshift=0.5cm] (b) {(b)}; % {\small Permutation}} ;

         \node[draw, dashed, fit=(x11)(e1)(x10), inner sep=4pt, "$U_1$"] {};
         \node[draw, dashed, fit=(x21)(e2)(x20), inner sep=4pt, "$U_2$"] {};
    \end{scope}

    \begin{scope}[xshift=12cm, yshift=0cm]
        \path[->, thick]
        node[format] (x11) {$L^{(1)}_1$}
        node[format, above of=x11, yshift=0cm, red, xshift=0.0cm] (e1){$\epsilon_1$}
        node[format, above of=e1] (x10) {$L^{(0)}_1$}

        % node[format, right of=x11, xshift=0cm, red] (u1) {$U_1$}
        % node[format, right of=u1, xshift=0cm, red] (u2) {$U_2$}
        
        node[format, right of=x11, xshift=1cm] (x21) {$L^{(1)}_2$}
        node[format, above of=x21, yshift=0cm, red, xshift=0cm] (e2){$\epsilon_2$}
        node[format, above of=e2] (x20) {$L^{(0)}_2$}
    
        node[square, below of=x11, yshift=-0.cm] (r1) {$R_1 = 1$}
        node[square, below of=x21, yshift=-0.cm] (r2) {$R_2 = 1$}
        node[format, below of=r1, yshift=-0.cm] (x1) {$L_1 = L^{(1)}_1$}
        node[format, below of=r2, yshift=-0.cm] (x2) {$L_2 = L^{(1)}_2$}

        node[format, below of=r1, yshift=0.2cm] (x1t) {}
        node[format, below of=r2, yshift=0.2cm] (x2t) {}

        (e1) edge[->, red] (x11)
        (e1) edge[->, red] (x10)
        (e2) edge[->, red] (x21)
        (e2) edge[->, red] (x20)

        (x10) edge[blue] (x20)
        (x10) edge[blue] (x21)
        (x11) edge[blue] (x20)
        (x11) edge[blue] (x21)
        
        % (r1) edge[blue] (r2)
        % (x1) edge[blue] (r2)
        
        (x11) edge[black, very thick, bend right=60] (x1t)
        (x21) edge[black, very thick, bend left=60] (x2t)
        % (x10) edge[gray, bend right=35] (x1t)
        % (x20) edge[gray, bend left=35] (x2)
        (r1) edge[gray] (x1t)
        (r2) edge[gray] (x2t)

        (x11) edge[->, dashed, auto=left, bend left=50, "$g$", very thick, black] (x10)
        (x21) edge[->, dashed, auto=right, bend right=50, "$g$", very thick, black] (x20)
        
        node[format, below of=x1, xshift=1.2cm, yshift=0.5cm] (c) {(c)}; % {\small Permutation}} ;

         \node[draw, dashed, fit=(x11)(e1)(x10), inner sep=4pt, "$U_1$"] {};
         \node[draw, dashed, fit=(x21)(e2)(x20), inner sep=4pt, "$U_2$"] {};
    \end{scope}
    
\end{tikzpicture}
}

\vspace{-.5cm}
%\caption{\red{Contenders }}
		\caption{ (a) A causal inference version of the bivariate permutation model with two counterfactual versions of $L_1$ and $L_2$ on the graph;
				(b) A world where an intervention $r_2=1$ is performed, yielding a situation where $L_2$ and $L_2^{(1)}$ coincide, i.e., $L_2=L_2^{(1)}$;
				(c) A causal model where we impose a \emph{rank preservation} relationship where a known bijective function $g(.)$ exists, such that
				$L_2^{(0)} = g(L_2^{(1)})$.  In this model, the joint distribution $p(L_1^{(r_1)}, L_2^{(r_2)})$ is identified by sequentially intervening on
				$r_2$ then $r_1$.
		}
		\label{fig:m-DAG-SWIG}
\end{figure}
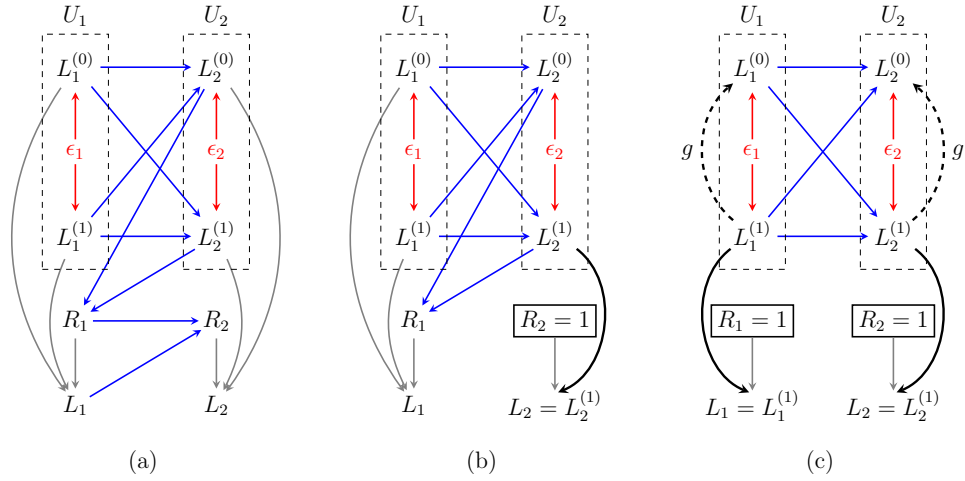

Note that in this causal model, the intervention $r_2=1$ is identified from the observed data, since $R_2$ only has observed
parents, namely $L_1$ and $R_1$.  The graph where this intervention is performed is shown via the conditional causal DAG in Fig.~\ref{fig:m-DAG-SWIG}(b), with the black edge representing the deterministic identity relationship between $L_2^{(1)}$ and $L_2$.

Just as in missing data, intervention $r_2=1$ renders the counterfactual $L_2^{(1)}$ and the observed proxy variable $L_2$ to be the same, meaning that $L_2^{(1)}$ becomes observable.
%and thus renders the former as an observed variable.
Since both $L_2^{(1)}$ and $L_2^{(0)}$ serve as parents of $R_1$, %rendering
the fact that $L_2^{(1)}$ becomes observed after the intervention $r_2=1$
does not suffice to render the subsequent intervention on $R_1$ identified.

In particular, since %$U_2$ generates
two counterfactual versions of $L_2$, namely $L_2^{(1)}$ and $L_2^{(0)}$, exist, observing the former does not suffice to eliminate all confounding. %via $U_2$.  
This is illustrated by a non-causal path in Fig.~\ref{fig:m-DAG-SWIG}(b) from $R_1$ to $L_1$ via $L^{(0)}_2$ and either $L^{(1)}_1$ or $L^{(0)}_1$.
% with $L_2^{(0)}$ acting as an unobserved confounder preventing identification. \blue{That is, there is an unblocked confounding path linking $R_1$ and $L^{(1)}_1$ through the unobserved $L_2^{(0)}$.}
% }

However, if $L^{(1)}_2$ and $L^{(0)}_2$ are continuous random variables and we make the strong additional assumption of \emph{rank preservation}, where $L_2^{(0)}$ is equal to $g(L_2^{(1)})$ for some  function $g(.)$ bijective on the statespace of $L_2^{(1)}$ (and $L_2^{(0)}$), then identification is recovered since
%, since
$p(R_1 \mid L_2^{(1)}, L_1^{(0)}) = p(R_1 \mid L_2^{(1)}, g(L_2^{(1)})) = p(R_1 \mid L_2^{(1)})$ becomes a functional of the observed data law once an intervention $r_2=1$ is performed.
This is illustrated in Fig.~\ref{fig:m-DAG-SWIG}(c), where the bijective relationship between $L_2^{(1)}$ and $L_2^{(0)}$ is represented by a dashed edge indexed by $g$.  Note that once $L_2^{(1)}$ is observed, so is $L_2^{(0)}$, which implies that conditioning on one counterfactual implicitly conditions on the other, due to their deterministic relationship.

% This is shown in Fig.~\ref{fig:m-DAG-SWIG} (c), with the bijective relationship between $L_2^{(1)}$ and $L_2^{(0)}$ shown as a deterministic relationship via a thicker bidirected edge.
We discuss additional causal models in which rank preservation yields identification in the Appendix~S4.

%{\bf Identification in Causal Inference using Missing Data Methods Augmented with Rank Preservation}

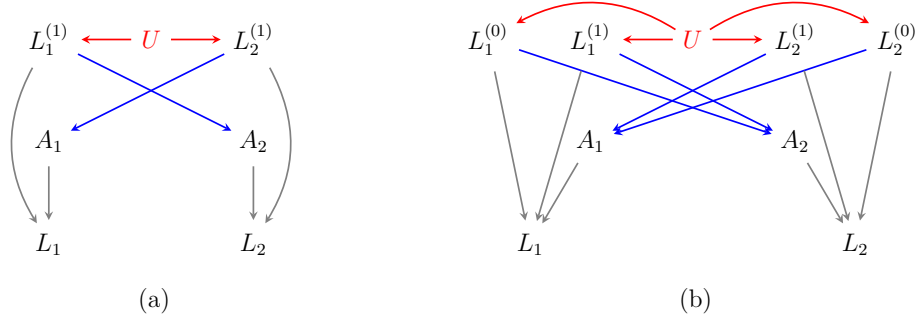
\begin{figure}[t] 
	\begin{center}
		\scalebox{0.8}{
			\begin{tikzpicture}[>=stealth, node distance=1.7cm]
				\tikzstyle{format} = [thick, circle, minimum size=1.0mm, inner sep=2pt]
				\tikzstyle{square} = [draw, thick, minimum size=4.5mm, inner sep=1pt]
				
				\begin{scope}[xshift=0cm]
					\path[->, thick]
					node[format] (l11) {$L^{(1)}_1$}
					node[format, red, right of=l11] (u) {$U$}
					node[format, right of=u] (l21) {$L^{(1)}_2$}
					
					node[format, below of=l11] (r1) {$A_1$}
					node[format, below of=l21] (r2) {$A_2$}

                    node[format, below of=r1] (l1) {$L_1$}
					node[format, below of=r2] (l2) {$L_2$}
					
					(u) edge[red] (l21)
					(u) edge[red] (l11)
					(l11) edge[blue] (r2)
					(l21) edge[blue] (r1)

                    (r1) edge[gray] (l1)
                    (l11) edge[gray, bend right] (l1)
                    (r2) edge[gray] (l2)
                    (l21) edge[gray, bend left] (l2)
					
					node[below of=l1, xshift=1.75cm, yshift=0.75cm] (t) {(a)}
					;
				\end{scope}
				
				\begin{scope}[xshift=9cm]
					\path[->, thick]
					node[format] (l11) {$L^{(1)}_1$}
					node[format, red, right of=l11] (u) {$U$}
					node[format, right of=u] (l21) {$L^{(1)}_2$}
					node[format, left of=l11] (l10) {$L^{(0)}_1$}
					node[format, right of=l21] (l20) {$L^{(0)}_2$}
					
					node[format, below of=l11] (r1) {$A_1$}
					node[format, below of=l21] (r2) {$A_2$}

                    node[format, below of=r1, xshift=-1cm] (l1) {$L_1$}
					node[format, below of=r2, xshift=1cm] (l2) {$L_2$}
					
					(u) edge[red] (l21)
					(u) edge[red] (l11)
					(u) edge[red, bend left] (l20)
					(u) edge[red, bend right] (l10)
					(l11) edge[blue] (r2)
					(l21) edge[blue] (r1)
					(l10) edge[blue] (r2)
					(l20) edge[blue] (r1)

                    (r1) edge[gray] (l1)
                    (l11) edge[gray, bend right=0] (l1)
                    (l10) edge[gray, bend left=0] (l1)
                    (r2) edge[gray] (l2)
                    (l21) edge[gray, bend left=0] (l2)
                    (l20) edge[gray, bend left=0] (l2)
					
					node[below of=l1, xshift=2.75cm, yshift=0.75cm] (t) {(b)}
					;
				\end{scope}
				
			\end{tikzpicture}
		}
	\end{center}
	\vspace{-1cm}
	\caption{
		(a) The bivariate block parallel missing data model.
		(b) The causal model analogue of the model in (a) where identification of causal effects is not possible without further assumptions, but is possible with rank preservation.
	}		
	\label{fig:missing-vs-causal}
\end{figure}

%Some types of identification strategies in missing data that we described do not easily translate to causal inference settings, due to the existence of multiple versions of counterfactuals.

As a second example for illustrating why identification strategies in missing data models do not easily translate to causal models, consider the causal analogue of the bivariate block-parallel model shown in Fig.~\ref{fig:missing-vs-causal}(b), with two binary treatments $A_1,A_2$ instead of missingness indicators, and two observed outcomes $L_1, L_2$, which implies two counterfactual versions of each observed outcome: $L^{(1)}_1, L^{(0)}_1$, and $L^{(1)}_2, L^{(0)}_2$. % (we exclude the observed proxies $L_1, L_2$ from the figure for brevity). 
The model is defined by the following restrictions: 
{\small
\begin{align}
	%	\notag
	A_2 \ci L^{(1)}_2, L^{(0)}_2, A_1 \mid L^{(1)}_1, L^{(0)}_1 \hspace{0.5cm} \text{and} \hspace{0.5cm}
	A_1 \ci L^{(1)}_1, L^{(0)}_1, A_2 \mid L^{(1)}_2, L^{(0)}_2.
	\label{eqn:block-parallel-causal}
\end{align}%
}
These restrictions do not yield nonparametric identification since the propensity scores in the model depend on both versions of the counterfactuals: %{and are thus considered ``cross-world'' quantities}:
{\small
\begin{align*}
	p({a}_2 \mid {l}^{(1)}_2, {l}^{(0)}_2, {a}_1, {l}^{(1)}_1, {l}^{(0)}_1) &= p({a}_2 \mid {l}^{(1)}_1, {l}^{(0)}_1),\\
	p({a}_1 \mid {l}^{(1)}_2, {l}^{(0)}_2, {a}_1, {l}^{(1)}_1, {l}^{(0)}_1) &= p({a}_1 \mid {l}^{(1)}_2, {l}^{(0)}_2).
\end{align*}
}
However, we can recover an argument for identification of $p({l}_1^{(i)}, {l}_2^{(j)})$ for any $(i,j) \in \{ 0, 1 \}^2$ via %the assumption of
{rank preservation}, which states that for $k \in \{ 0, 1 \}$ there exist bijections $g_k(.)$ such that $L_1^{(1-i)} = g_1(L_1^{(i)})$ and $L_2^{(1-j)} = g_2(L_2^{(j)})$.
%$k \in \{ 0, 1 \}$, there exists bijections $L_k^{(1-i)} = g_k(L_k^{(i)})$ and $L_k^{(1-j)} = g_k(L_k^{(j)})$.
%$g_k$ such that $L_k^{(1-j)} = g_k(L_k^{(j)})$. 
Identification of $p({l}^{(i)}_1, {l}^{(j)}_2)$ then proceeds as follows:
%\vspace{-0.25cm}
{\scriptsize
	\begin{align*}
		&p({l}^{(i)}_1, {l}^{(j)}_2) 
		=
		\frac{
			p({l}^{(i)}_1, {l}^{(j)}_2, A_1 = i, A_2 = j)
		}{
			p(A_1 = i, A_2 = j \mid {l}^{(i)}_1, {l}^{(j)}_2)
		}\\
		&=
		\frac{
			p( {l}^{(i)}_1,  {l}^{(j)}_2, A_1 = i, A_2 = j)
		}{
			\sum_{ {l}^{(1-i)}_1,  {l}^{(1-j)}_2} p(A_1 = i, A_2 = j \mid  {l}^{(i)}_1,  {l}^{(j)}_2,  {l}^{(1-i)}_1,  {l}^{(1-j)}_2)\times  p( {l}^{(1-i)}_1,  {l}^{(1-j)}_2 \mid  {l}^{(i)}_1,  {l}^{(j)}_2)
		}\\
		&= 
		\frac{
			p( {l}^{(i)}_1,  {l}^{(j)}_2, A_1 = i, A_2 = j)
		}{
			\sum_{ {l}^{(1-i)}_1,  {l}^{(1-j)}_2} p(A_1 = i \mid  {l}^{(j)}_2,  {l}^{(1-j)}_2, A_2 = j) \times p(A_2 = j \mid  {l}^{(i)}_1,  {l}^{(1-i)}_1, A_1 = i) \times p( {l}^{(1-i)}_1,  {l}^{(1-j)}_2 \mid  {l}^{(i)}_1,  {l}^{(j)}_2)
		} \\
		&= \! 
		\frac{
			p( {l}^{(i)}_1,  {l}^{(j)}_2, A_1 = i, A_2 = j)
		}{
			\sum_{ {l}^{(1-i)}_1,  {l}^{(1-j)}_2} p(A_1 = i |  {l}^{(j)}_2,  {l}^{(1-j)}_2, A_2 = j) \! \times \! p(A_2 = j |  {l}^{(i)}_1,  {l}^{(1-i)}_1, A_1 = i) \! \times \! \mathbb{I}( {l}^{(1-i)}_1 \! = \! g_1( {l}^{(i)}_1), {l}^{(1-j)}_2  \! = \! g_2( {l}^{(j)}_2))
		} \!\!\!\! \\
		&=
		\frac{
			p( {l}^{(i)}_1,  {l}^{(j)}_2, A_1 = i, A_2 = j)
		}{
			p(A_1 = i \mid {l}^{(j)}_2, A_2 = j) \times p(A_2 = j \mid {l}^{(i)}_1, A_1 = i)
		}\\
		&=
		\frac{
			p({l}_1, {l}_2, A_1 = i, A_2 = j)
		}{
			p(A_1 = i \mid {l}_2, A_2 = j) \times p(A_2 = j \mid {l}_1, A_1 = i)
		}.
	\end{align*}
}%

Here, the first and second equalities follow by rules of probability, the third by (\ref{eqn:block-parallel-causal}), the fourth and fifth by the rank preservation assumption, and the
last by consistency.
This derivation is structurally very similar to the derivation for the bivariate block-parallel model, except for the last two steps which explicitly rely on rank preservation. In Appendix~S4, %\ref{app:proofs} 
we show that this structural similarity is quite general, and a similar identification strategy can be defined for a $K$ variable causal analogue of the block-parallel model endowed with rank preservation.

While much of the discussion in this paper has focused on how causal identification techniques can be applied or extended to missing data settings, the above example demonstrates missing data techniques can only be applied to causal settings given much stronger untestable assumptions, such as rank preservation, than those %typically encoded
plausibly assumed in causal models. 

%; in contrast, these assumptions are often implicit in the definitions of classical missing data models.
%This raises interesting questions for future work on the suitability of these assumptions in missing data -- e.g., an examination of the ``no interference'' assumption, which was implicit in all models discussed in this paper and violations of which may require different identification strategies.

% missing data identification strategies we have discussed  can only be
% applied to causal settings given much 
% stronger untestable assumptions, such as rank
% preservation, than those typically plausile oor assumed in causal models. Therefore is gained by including counterfactuals on the causal DAG.

%%%%%%%%%%%%%%%%%%%%%%%%%%%%%%%%%%%%%%%%%%%%%%%%%%%%%%%%%%%%%%%%%%%%%%%%%%%%%%%%%%%%%%%%%%%%%%%%%%%%%%%%%%%%%%%%%%%%%%%%%%%%
\section*{Supplementary Materials}

The supplementary materials contain discussions on identification of missing data DAG models using the odds ratio parameterization extension of identification techniques to m-DAG models with unmesaured confounders, results on identification of the full law,  proofs, and some additional results on identification in causal models using rank preservation.

\par
%%%%%%%%%%%%%%%%%%%%%%%%%%%%%%%%%%%%%%%%%%%%%%%%%%%%%%%%%%%%%%%%%%%%%%%%%%%%%%%%%%%%%%%%%%%%%%%%%%%%%%%%%%%%%%%%%%%%%%%%%%%%
%\section*{Acknowledgements}

%Write the acknowledgements here.
\par

%%%%%%%%%%%%%%%%%%%%%%%%%%%%%%%%%%%%%%%%%%%%%%%%%%%%%%%%%%%%%%%%%%%%%%%%%%%%%%%%%%%%%%%%%%

\bibhang=1.7pc
\bibsep=2pt
\fontsize{9}{14pt plus.8pt minus .6pt}\selectfont
\renewcommand\bibname{\large \bf References}
%\begin{thebibliography}{11}
\expandafter\ifx\csname
natexlab\endcsname\relax\def\natexlab#1{#1}\fi
\expandafter\ifx\csname url\endcsname\relax
  \def\url#1{\texttt{#1}}\fi
\expandafter\ifx\csname urlprefix\endcsname\relax\def\urlprefix{URL}\fi

% use bibfile 
\bibliographystyle{chicago}      % Chicago style, author-year citations
\bibliography{references}   % name your BibTeX data base

%%%%%%%%%%%%%%%%%%%%%%%%%%%%%%%%%%%%%%%%%%%%%%%%%%%%%%%%%%%%%%%%%%%%%%%%%%%%%%%%%%%%%%%%%%%%%%%%%%%%%%%%%%%%%%%%%%%%%%%%%%%%
\pagebreak
\vskip .65cm
\noindent
Razieh Nabi, Department of Biostatistics and Bioinformatics, Emory University
\vskip 2pt
\noindent
E-mail: (razieh.nabi@emory.edu)
\vskip 2pt

\noindent
Rohit Bhattacharya, Department of Computer Science, Williams College
\vskip 2pt
\noindent
E-mail: (rb17@williams.edu)

\noindent
Ilya Shpitser, Department of Computer Science, Johns Hopkins University
\vskip 2pt
\noindent
E-mail: (ilyas@cs.jhu.edu)

\noindent
James Robins, Department of Epidemiology, Harvard T. H. Chan School of Public Health 
\vskip 2pt
\noindent
E-mail: (robins@hsph.harvard.edu)

% \vskip .3cm
%\centerline{(Received ???? 20??; accepted ???? 20??)}\par
\end{document}

%%%%%%%%%%%%%%%%%%%%%%%%%%%%%%%%%%%%%%%%%%%%%%%%%%%%%%%%%%%%%%%%%%%%%%%%%%%%%%%%%%%%%%%%%%%%%%%%%%%%%%%%%%%%%%%%%%%%%%%%%%%%
%%%%%%%%%%%%%%%%%%%%%%%%%%%%%%%%%%%%%%%%%%%%%%%%%%%%%%%%%%%%%%%%%%%%%%%%%%%%%%%%%%%%%%%%%%%%%%%%%%%%%%%%%%%%%%%%%%%%%%%%%%%%

\renewcommand{\baselinestretch}{2}

\markright{ \hbox{\footnotesize\rm Statistica Sinica: Supplement
%{\footnotesize\bf 24} (201?), 000-000
}\hfill\\[-13pt]
\hbox{\footnotesize\rm
%\href{http://dx.doi.org/10.5705/ss.20??.???}{doi:http://dx.doi.org/10.5705/ss.20??.???}
}\hfill }

\markboth{\hfill{\footnotesize\rm RAZIEH NABI, ROHIT BHATTACHARYA, ILYA SHPITSER, AND JAMES M.~ROBINS} \hfill}
{\hfill {\footnotesize\rm CAUSAL VIEWS OF MISSING DATA MODELS } \hfill}

\renewcommand{\thefootnote}{}
$\ $\par \fontsize{12}{14pt plus.8pt minus .6pt}\selectfont

%%%%%%%%%%%%%%%%%%%%%%%%%%%%%%%%%%%%%%%%%%%%%%%%%%%%%%%%%%%%%%%%%%%%%%%%%%%%%%%%%%%%%%%%%%%%%%%%%%%%%%%%%%%%%%%%%%%%%%%%%%%%

 \centerline{\large\bf CAUSAL AND COUNTERFACTUAL VIEWS OF}
\vspace{2pt}
 \centerline{\large\bf MISSING DATA MODELS}
%\vspace{2pt}
% \centerline{\large\bf IF THIRD LINE IS NEEDED}
\vspace{.25cm}
 \centerline{Razieh Nabi, Rohit Bhattacharya, Ilya Shpitser, James M.~Robins}
\vspace{.4cm}
 \centerline{\it Emory University, Williams College, }
 \centerline{\it Johns Hopkins University, Harvard University}
\vspace{.55cm}
 \centerline{\bf Supplementary Material}
\vspace{.55cm}
\fontsize{9}{11.5pt plus.8pt minus .6pt}\selectfont
\noindent
%CONTENT OF A BRIEF NOTE.........\\
%CONTENT OF THE BRIEF NOTE.\\
%CONTENT OF THE BRIEF NOTE.\\
%\par

\setcounter{section}{0}
\setcounter{equation}{0}
\def\theequation{S\arabic{section}.\arabic{equation}}
\def\thesection{S\arabic{section}}

\fontsize{12}{14pt plus.8pt minus .6pt}\selectfont

The supplementary materials are organized as follows. 
Section~\ref{app:odds} describes identification using an odds ratio parameterization, different from m-DAG factorization. 
Section~\ref{sec:missing_data_w_hidden_vars} extends the concepts of m-DAGs to m-DAGs with hidden variables. 
Section~\ref{sec:missing_data_full_law} expands the target law identification arguments in the manuscript to full law identification. 
Section~\ref{app:proofs} contains all our proofs.

%######################################
\section{Identification via an Odds Ratio Parameterization} 
%######################################
\label{app:odds}

We mentioned in the main draft that \cite{nabi20completeness} used an odds ratio parameterization to derive a sound and complete algorithm for full law identification in m-DAGs.
In the following, we go over an example to show how the target law can be identified in some cases via an odds ratio parameterization of conditional distributions \citep{chen07semiparametric,shpitser23lc}.

This parameterization yields a sound and complete algorithm for the identification of the full law $p(r,l^{(1)})$ in graphical missing data models, but does not yield a complete algorithm for identification of the target law $p(l^{(1)})$.  A simple example of a model where the target law is identified but the full law is not is shown in Fig.~\ref{fig:odds-ratio}(a).  In fact, deriving a sound and complete algorithm for the identification of the target law %in missing data models represented by graphs
by \emph{any} method -- whether by the parameterization described below, or methods based on the g-formula described in the main draft -- is currently an open problem.

Consider the m-DAG in Fig.~\ref{fig:odds-ratio}(b). The non-deterministic portion of the full law factorizes as $
	p(l_1^{(1)}) \times p(l_2^{(1)}\mid l_1^{(1)}) \times p(l_3^{(1)} \mid l_1^{(1)}, l_2^{(1)}) \times
	p(r_1 \mid r_2, l_3^{(1)}) \times p(r_2 \mid r_3, l^{(1)}_1)  \times p(r_3 \mid l_1^{(1)}). 
$ 
The following conditional independence statements follow from this factorization: $R_1 \ci \{L^{(1)}_1, L^{(1)}_2, R_3\} \mid R_2, L_3^{(1)}$, and $R_2 \ci \{L^{(1)}_2, L^{(1)}_3\} \mid R_3, L^{(1)}_1$ and $R_3 \ci \{L^{(1)}_2, L^{(1)}_3\} \mid  L^{(1)}_1. 
$

According to Proposition~\ref{prop:miss-id}, we have: 
\begin{align}
	p(l_1, l_2, l_3 \ \Vert \ r=1) = \left. \frac{ p(l_1, l_2, l_3, r_1,
		r_2, r_3) }{ p(r_1 | l_3^{(1)}, r_2) \times  p(r_2
		| l^{(1)}_1, r_3) \times p(r_3 | l_1^{(1)})  } \right|_{{r} = {1}}. 
	\label{eq:ex-odds-ratio-gformula}
\end{align}%
We can use the independencies encoded in the m-DAG and consistency in missing data models to identify the propensity score of $R_1$ as follows:  
\begin{align}
	p(r_1 \mid \pa_{\mathcal{G}}(r_1))\vert_{r=1} &= p(r_1 | l^{(1)}_3, r_2)\vert_{r=1} = p(r_1 = 1 | l_3, r_2 = 1, r_3 = 1). \label{ex:odds-ratio-seq-r1}
\end{align}%
We cannot immediately obtain the propensity score of $R_2$, i.e., $p(r_2 \mid l^{(1)}_1, r_3)\vert_{r=1}$, since $R_2 \not\ci R_1 \mid L^{(1)}_1, R_3.$ This can still be identified using a total order where $R_1$ is intervened on before $R_2$. 

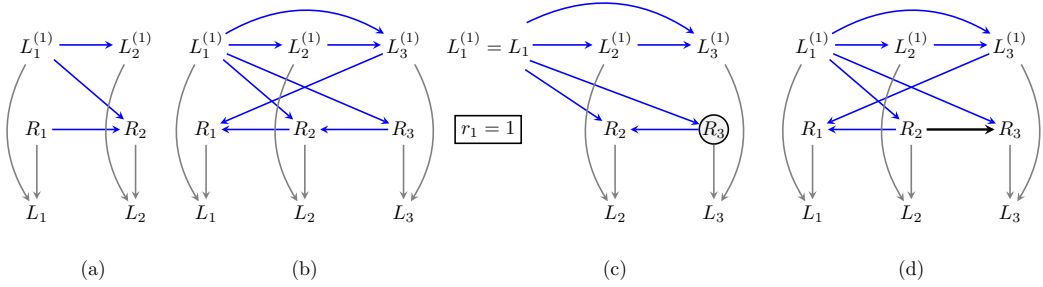
\begin{figure}[!t]
	\begin{center}
		\scalebox{0.75}{
			\begin{tikzpicture}[>=stealth, node distance=1.5cm]
				\tikzstyle{format} = [thick, circle, minimum size=1.0mm, inner sep=0pt]
				\tikzstyle{square} = [draw, thick, minimum size=1mm, inner sep=3pt]
				\tikzstyle{format2} = [draw, circle, fill=gray!10, minimum size=1.0mm, inner sep=0pt]

				\begin{scope}
					\path[->,  thick]
					node[format] (x11) {$L^{(1)}_1$}
					node[format, right of=x11, xshift=0.25cm] (x21) {$L^{(1)}_2$}
					node[format, below of=x11] (r1) {$R_1$}
					node[format, below of=x21] (r2) {$R_2$}
					node[format, below of=r1, yshift=0.cm] (x1) {$L_1$}
					node[format, below of=r2, yshift=0.cm] (x2) {$L_2$}
					
					(x11) edge[blue] (x21)
					(r1) edge[blue] (r2)
					(x11) edge[blue] (r2)
					
					(r1) edge[gray] (x1)
					(x11) edge[gray, bend right] (x1)
					(r2) edge[gray] (x2)
					(x21) edge[gray, bend right] (x2)

					node [below of=x2, xshift=-0.75cm, yshift=0.5cm] {(a)} ;
				
				\end{scope}

				\begin{scope}[xshift=3.cm]
					\path[->,  thick]
					node[format] (x11) {$L^{(1)}_1$}
					node[format, right of=x11, xshift=0.25cm] (x21) {$L^{(1)}_2$}
					node[format, right of=x21, xshift=0.25cm] (x31) {$L^{(1)}_3$}
					node[format, below of=x11] (r1) {$R_1$}
					node[format, below of=x21] (r2) {$R_2$}
					node[format, below of=x31] (r3) {$R_3$}
					node[format, below of=r1, yshift=0.cm] (x1) {$L_1$}
					node[format, below of=r2, yshift=0.cm] (x2) {$L_2$}
					node[format, below of=r3, yshift=0.cm] (x3) {$L_3$}
					
					(x11) edge[blue, bend left] (x31)
					(x11) edge[blue] (x21)
					(x21) edge[blue] (x31)
					(r3) edge[blue] (r2)
					(r2) edge[blue] (r1)
					(x11) edge[blue] (r2)
					(x11) edge[blue] (r3)
					(x31) edge[blue] (r1)
					
					(r1) edge[gray] (x1)
					(x11) edge[gray, bend right] (x1)
					(r2) edge[gray] (x2)
					(x21) edge[gray, bend right] (x2)
					(r3) edge[gray] (x3)
					(x31) edge[gray, bend left] (x3)
					node [below of=x2, xshift=0.0cm, yshift=0.5cm] {(b)} ;
				\end{scope}
				\begin{scope}[xshift=8.cm]
					\path[->,  thick]
					node[format] (x11) {$L^{(1)}_1 = L_1$}
					node[format, right of=x11, xshift=0.75cm] (x21) {$L^{(1)}_2$}
					node[format, right of=x21, xshift=0.25cm] (x31) {$L^{(1)}_3$}
					node[square, below of=x11] (r1) {$r_1=1$}
					node[format, below of=x21] (r2) {$R_2$}
					node[format2, below of=x31] (r3) {$R_3$}
					node[format, below of=r1, yshift=0.cm] (x1) {}
					node[format, below of=r2, yshift=0cm] (x2) {$L_2$}
					node[format, below of=r3, yshift=0cm] (x3) {$L_3$}
					
					(x11) edge[blue, bend left] (x31)
					(x11) edge[blue] (x21)
					(x21) edge[blue] (x31)
					(r3) edge[blue] (r2)
					(x11) edge[blue] (r2)
					(x11) edge[blue] (r3)
					(r2) edge[gray] (x2)
					(x21) edge[gray, bend right] (x2)
					(r3) edge[gray] (x3)
					(x31) edge[gray, bend left] (x3)
					
					node [below of=x2, xshift=0.0cm, yshift=0.5cm] {(c)} ;
				\end{scope}
				\begin{scope}[xshift=13.75cm]
					\path[->,  thick]
					node[format] (x11) {$L^{(1)}_1$}
					node[format, right of=x11, xshift=0.25cm] (x21) {$L^{(1)}_2$}
					node[format, right of=x21, xshift=0.25cm] (x31) {$L^{(1)}_3$}
					node[format, below of=x11] (r1) {$R_1$}
					node[format, below of=x21] (r2) {$R_2$}
					node[format, below of=x31] (r3) {$R_3$}
					node[format, below of=r1, yshift=0.cm] (x1) {$L_1$}
					node[format, below of=r2, yshift=0.cm] (x2) {$L_2$}
					node[format, below of=r3, yshift=0.cm] (x3) {$L_3$}
					
					(x11) edge[blue, bend left] (x31)
					(x11) edge[blue] (x21)
					(x21) edge[blue] (x31)
					(r2) edge[black, very thick] (r3)
					(r2) edge[blue] (r1)
					(x11) edge[blue] (r2)
					(x11) edge[blue] (r3)
					(x31) edge[blue] (r1)
					
					(r1) edge[gray] (x1)
					(x11) edge[gray, bend right] (x1)
					(r2) edge[gray] (x2)
					(x21) edge[gray, bend right] (x2)
					(r3) edge[gray] (x3)
					(x31) edge[gray, bend left] (x3)
					
					node [below of=x2, xshift=0.0cm, yshift=0.5cm] {(d)} ;
				\end{scope}
			\end{tikzpicture}
		}
	\end{center}
	\vspace{-1cm}
	\caption{(a) A simple model where the target law $p(l_1^{(1)}, l_2^{(1)})$ is identified, but the full law $p(l_1^{(1)}, l_2^{(1)}, r_1, r_2)$ is not;
	(b) Example of an m-DAG used to illustrate target law identification with odds ratio parameterization of the missingness selection model;  (c) Graph derived from (b) representing an intervention on $R_1$ and the induced selection bias on $R_3$; (d) An m-DAG that is Markov equivalent to the m-DAG in (b).  }
	\label{fig:odds-ratio}
\end{figure}

Intervening on $R_1$ results in the following kernel that is Markov relative to the graph in Fig.~\ref{fig:odds-ratio}(c), with the induced selection bias on $R_3$. 
\begin{align*}
	p(l_2^{(1)}, l_3^{(1)}, l_1, l_2, l_3, r_2, r_3 \ \Vert \ r_1 = 1) =\left. \frac{p(l_1, l_2^{(1)}, l_3^{(1)}, l_2, l_3, r_1, r_2, r_3)}{p(r_1 \mid r_2, l^{(1)}_3)}\right\vert_{r_1 = 1} 
\end{align*}
The propensity score of $R_2$ evaluated at $R_3=1$ is equivalent to $p(r_2 = 1 \mid r_3 = 1, l^{(1)}_1 \ \Vert \ r_1 = 1)$. This is identified from the marginal kernel $p(l_1, l_3^{(1)}, r_2, r_3=1 \ \Vert \ r_1 = 1)$ which is equal to ${p(l_1,  l_3,  r_1=1, r_2, r_3=1)}/p(r_1$ $= 1 \mid r_2, l_3, r_3=1)$.
%\begin{align*}
%	p(l_1, l^{(1)}_3, r_2, r_3=1 \ \Vert \ r_1 = 1) = \frac{p(l_1,  l_3,  r_1=1, r_2, r_3=1)}{p(r_1 = 1 \mid r_2, l_3, r_3=1)}. 
%\end{align*}

We now proceed to identify the propensity score of $R_3$, $p(r_3 \mid l^{(1)}_1)\vert_{r=1}$, which is not immediately obvious since $R_3 \not\ci R_1 \mid L^{(1)}_1$.  Intervening on $R_1$ and setting it  to $1$ leads to a distribution where $R_3$ is necessarily selected on since the propensity score of  $R_1$  is identified by restricting data to cases where $R_3 = 1$. Thus, we cannot  identify the propensity score of $R_3$  in this post-intervention kernel distribution. A similar issue holds if we try to intervene on $R_2$ since identification of the propensity score of  $R_2$ is obtained from a kernel distribution where we first intervene on $R_1$, which as mentioned introduces selection bias on $R_3$. It seems that we have exhausted all of our options based on the discussion of partial orders of identification. However, there is an alternative strategy that leads to identification of not just the target law, but the full law as well. 

\cite{nabi20completeness} made the observation that the conditional density $p(r_3 \mid r_2, l^{(1)}_1)$ is  identified, since $R_3 \ci R_1 \mid R_2, L^{(1)}_1$. From the preceding discussion, it is also clear that $p(r_2 \mid r_3=1, l^{(1)}_1)$ is identified.  Given that these conditional densities $p(r_2 \mid r_3=1, l^{(1)}_1)$ and $p(r_3 \mid r_2, l^{(1)}_1)$ are identified, they considered an odds ratio parameterization of the joint density $p(r_2, r_3 \mid \pa_{\cal G}(r_2, r_3)) = p(r_2, r_3 \mid l^{(1)}_1)$ as follows \citep{chen07semiparametric},
\begin{align*}
	p(r_2, r_3 \mid l^{(1)}_1) = 
	\frac{1}{Z} \times p(r_2 | r_3 = 1, l^{(1)}_1) \times p(r_3 | r_2 = 1, l^{(1)}_1) \times \text{OR}(r_2, r_3 | l^{(1)}_1),
\end{align*}%
where $Z$ is the normalizing term, and
\begin{align*}
	&\text{OR}(r_2, r_3 \mid l^{(1)}_1) = \frac{p(r_3 \mid r_2, l^{(1)}_1)}{p(r_3 = 1 \mid r_2, l^{(1)}_1)} \times \frac{p(r_3 = 1 \mid r_2 = 1, l^{(1)}_1)}{p(r_3 \mid r_2 = 1, l^{(1)}_1)}. 
\end{align*}%
All the terms in  above parameterization are identified. This immediately implies the identifiability of the individual propensity scores for $R_2$ and $R_3$. This result, in addition to the fact that $p(r_1 \mid r_2, l^{(1)}_1)$ is identified, leads to  identification of both the target law and the full law, as the missingness process $p(r \mid l^{(1)})$ is also identified for all possible values of the missingness indicators. It is interesting to point out that the m-DAG in Fig.~\ref{fig:odds-ratio}(b) is Markov equivalent to the one in Fig.~\ref{fig:odds-ratio}(d), which means, the m-DAG model in both examples implies the same set of independence restrictions on the full data law. It is perhaps easier to see how identification in Fig.~\ref{fig:odds-ratio}(d) proceeds using techniques  discussed in the main draft -- the target law is identified via parallel interventions on $R_1$ and $R_3$ followed by a sequential intervention on $R_2$. That is, identification can be obtained via the partial order $\{ \{I_{r_1}, I_{r_3}\} < I_{r_3} \}.$

%It is worthwhile pointing out that if we add the $L^{(1)}_2 \rightarrow R_3$ edge to the m-DAG in  Fig.~\ref{fig:odds-ratio} (b), the full law and thus the target law remains identified via the odds ratio parameterization.
%However in this case identification is obtained in a single step by factorizing the full law.  For details, see \citep{nabi20completeness}.
 %However, the connection to how target law identification can be carried out using the described methods in terms of finding partial orders of intervention remains an open problem. 

%\clearpage

%####################################################
\section{m-DAG Models with Unmeasured Confounders}
\label{sec:missing_data_w_hidden_vars}
%####################################################

Previous sections illustrated how identification may be accomplished in missing data models represented by a DAG where all variables are either fully or partially observed. However, just as in standard causal inference problems, most realistic missing data models include variables that are completely unobserved. We represent such models with an m-DAG ${\cal G}_m(L, R, L^{(1)}, U)$, where the vertex set $U$ represents unobserved variables.  By analogy with restrictions in Section~\ref{sec:missing_data_dag_models}, we require that $(L^{(1)} \cup U) \cap \left\{ \de_{\mathcal{G}_m}(R) \cup \de_{\mathcal{G}_m}(L) \right\} = \emptyset$, i.e., there are no directed paths from any of the missingness indicators or proxy variables pointing towards variables in $U$ or $L^{(1)}$. To clearly distinguish hidden variables from others in the model, we will render edges adjacent to such vertices in red.

In some m-DAGs with hidden variables, straightforward generalizations of identification strategies developed for m-DAGs without hidden variables can be developed. Consider the hidden variable m-DAG in Fig.~\ref{fig:hidden-vars} where $U_1$, $U_2$, and $U_3$ are completely unobserved. Although the joint over all variables in this model still factorizes with respect to this m-DAG, no factors containing unobserved variables in $U$ can be used in identification or estimation strategies for the target $p({l^{(1)}})$ or the selection mechanism $p(r \mid l^{(1)})$. Thus, in this setting it is useful to consider a factorization of the marginal model defined over variables that are either fully or partially observed. Recall that under any valid topological ordering on the variables, the ordered local Markov property simplifies each factor $p(v_i \mid \past_{{\cal G}_m}(v_i))$ in the chain rule factorization to simply $p(v_i \mid \pa_{{\cal G}_m}(v_i)),$ as each variable is independent of its past (except parents) given its parents. We now describe an analogue of the ordered local Markov property and factorization that relies only on partially or fully observed variables in the m-DAG, and demonstrate how this leads to an identification strategy.

Let $V = L^{(1)} \cup R \cup L$ denote the set of all partially and fully observed variables in ${\cal G}_m$. We define the \emph{district}
%\footnote{Districts are often referred to as confounded-components (c-components) in the causal inference literature, as they often negatively impact identifiability of causal parameters \citep{tian02on}.} 
of $V_i \in V$ as the set of all variables $V_j \in V$ such that there exists a path connecting $V_i$ and $V_j$ that consists of only red edges, where any unmeasured variable $U_k \in U$ along the path is not a collider and any variable $V_k \in V$ along the path is a collider. We will use $\dis_{{\cal G}_m}(V_i)$ to denote the district of $V_i$ in ${\cal G}_m$; by convention $\dis_{{\cal G}_m}(V_i)$ includes $V_i$ itself. Given any valid topological order on all the variables in ${\cal G}_m$ (including unobserved variables) define $({\cal G}_m)_{\overline{V_i}}$ to be the subgraph of ${\cal G}_m$ consisting of only the variables that appear before $V_i$ in the topological order (including $V_i$ itself) and the arrows present between these variables -- {not to be confused with alternative usage of the notation ${\cal G}_{\overline{X}}$ employed in the causal graph literature \citep{pearl00causality} to represent a graph where incoming edges into $X$ have been deleted. }
Then, the \textit{Markov pillow} of $V_i$, denoted as $\tb_{{\cal G}_m}(V_i),$ is defined as the district of $V_i$ and the observed parents of the district of $V_i$ (excluding $V_i$ itself) in the subgraph $({\cal G}_m)_{\overline{V_i}}$. That is,
$\tb_{{\cal G}_m}(V_i) \coloneqq \big\{\dis_{({\cal G}_m)_{\overline{V_i}}}(V_i) \cup \pa_{({\cal G}_m)_{\overline{V_i}}}\big(\dis_{({\cal G}_m)_{\overline{V_i}}}(V_i)\big)\big\} \cap \big\{ V \setminus V_i \big\}.$ We suppress the dependence of the definition of the Markov pillow on the topological order for notational simplicity. Given these definitions, we have the following independence relations among the observed variables in a hidden variable DAG that resemble the ordered local Markov property in fully observed DAGs \citep{tian02on, bhattacharya2022semiparametric}:
\begin{align}
	V_i \ci \past_{{\cal G}_m}(V_i) \cap V \setminus  \tb_{{\cal G}_m}(V_i) \mid \tb_{{\cal G}_m}(V_i).
	\label{eq:markov_pillow_ci}
\end{align}
That is, each variable is independent of its observed past  given its Markov pillow. Using this observation we can simplify the chain rule factorization according to any valid topological order on the observed variables as,
\begin{align}
	p(v) = \prod_{v_k \in V} p(v_k \mid \past_{{\cal G}_m}(v_k)) = \prod_{v_k \in V} p(v_k \mid \tb_{{\cal G}_m}(v_k)).
\end{align}

\begin{figure}[!t]
	\begin{center}
		\scalebox{0.7}{
			\begin{tikzpicture}[>=stealth, node distance=1.8cm]
				\tikzstyle{format} = [thick, circle, minimum size=1.0mm, inner sep=0pt]
				\tikzstyle{square} = [draw, thick, minimum size=1mm, inner sep=3pt]
				\begin{scope}
					\path[->,  thick]
					node[format] (x11) {$L^{(1)}_1$}
					node[format, right of=x11, xshift=0.25cm] (x21) {$L^{(1)}_2$}
					node[format, right of=x21, xshift=0.25cm] (x6) {$L_6$}
					node[format, right of=x6, xshift=0.25cm] (x5) {$L_5$}
					node[format, right of=x5, xshift=0.25cm] (x31) {$L^{(1)}_3$}
					node[format, right of=x31, xshift=0.25cm] (x41) {$L^{(1)}_4$}
					
					node[format, below of=x11] (r1) {$R_1$}
					node[format, below of=x21] (r2) {$R_2$}
					node[format, below of=x31] (r3) {$R_3$}
					node[format, below of=x41] (r4) {$R_4$}
					
					node[format, right of=r2, xshift=0.25cm] (u2) {$U_2$}
					node[format, above right of=x6, xshift=-0.2cm] (u3) {$U_3$}
					node[format, left of=r3, xshift=-0.25cm] (u1) {$U_1$}
					
					node[format, below of=r1, yshift=0.5cm] (x1) {$L_1$}
					node[format, below of=r2, yshift=0.5cm] (x2) {$L_2$}
					node[format, below of=r3, yshift=0.5cm] (x3) {$L_3$}
					node[format, below of=r4, yshift=0.5cm] (x4) {$L_4$}
					
					(x11) edge[blue, bend left] (x5)
					(x31) edge[blue, bend right=20] (x6)
					(x5) edge[blue, bend right = 20] (x21)
					(x6) edge[blue, bend left] (x41)
					
					(x31) edge[blue] (r1)
					(x6) edge[blue] (r1)
					(x41) edge[blue] (r2)
					(x6) edge[blue] (r2)
					(x11) edge[blue] (r3)
					(x5) edge[blue] (r3)
					(x21) edge[blue] (r4)
					(x5) edge[blue] (r4)
					
					(u1) edge[red] (r3)
					(u1) edge[red, bend right=20] (r4)
					(u1) edge[red] (x5)
					
					(u3) edge[red, bend right=5] (x11)
					(u3) edge[red, bend right=7] (x21)
					(u3) edge[red, bend left=7] (x31)
					(u3) edge[red, bend left=5] (x41)
					
					(u2) edge[red] (x6)
					(u2) edge[red, bend left=20] (r1)
					(u2) edge[red] (r2)
					
					(r1) edge[gray] (x1)
					(x11) edge[gray, bend right] (x1)
					(r2) edge[gray] (x2)
					(x21) edge[gray, bend right] (x2)
					(r3) edge[gray] (x3)
					(x31) edge[gray, bend left] (x3)
					(r4) edge[gray] (x4)
					(x41) edge[gray, bend left] (x4)
					
					node[format, below of=u2, xshift=1cm, yshift=-0.35cm] (a) {(a)} ;
				\end{scope} 
				
				\begin{scope}[xshift=13cm, yshift=-1.75cm]
					\path[->, thick]
					node[format, yshift=1cm] (x11) {$X$}
					node[format, right of=x11, xshift=0.5cm] (x21) {$A$}
					node[format, above of=x21, xshift=0cm, yshift=-0.5cm] (u1) {$U_1$}
					node[format, right of=x21, xshift=0.5cm] (x31) {$Y$}
					node[format, below of=x31, yshift=0.0cm] (x3) {$Y^*$}
					node[format,  left of=x3] (r3) {$R_Y$}
					node[format,  above left of=r3, xshift=-0.1cm, yshift=-0.35cm] (u2) {$U_2$}

					(x11) edge[blue] (x21) 
					(x21) edge[blue] (x31) 
					(u1) edge[red, bend left=0] (x11)
					(u1) edge[red, bend left=0] (x31) 
					
					(u2) edge[red] (r3)
					(u2) edge[red] (x11)
					(x31) edge[gray, bend left=0] (x3)
					(r3) edge[gray] (x3)
					
					node[format, below of=r3, xshift=-0.25cm, yshift=0.4cm] (b) {(b)} ;
				\end{scope}
			\end{tikzpicture}
		}
	\end{center}
	\vspace{-1cm}
	\caption{(a) A graph corresponding to a missing data model with hidden variables where identification of the law $p(l_1^{(1)}, l_2^{(1)}, l_3^{(1)}, l_4^{(1)}, l_5, l_6)$ is possible; (b) An example where the causal effect is identified but the target law is not identified. }
	\label{fig:hidden-vars}
\end{figure}
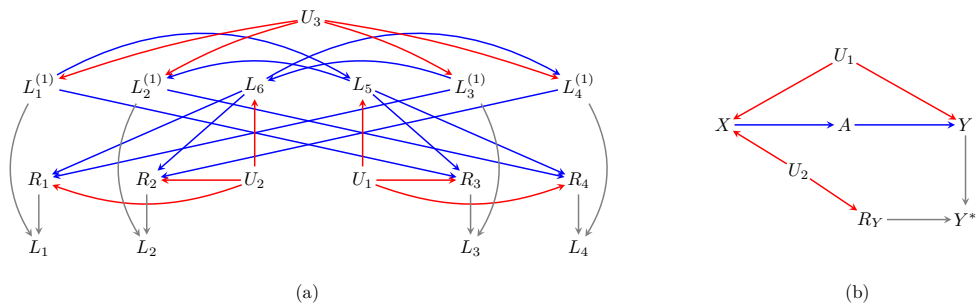

We now apply the above factorization to study the identification of the target laws $p(l^{(1)})$ in hidden variable m-DAGs. The properties of missing data graphs, as we described them, namely that for every $L_k$, $\pa_{\mathcal{G}_m}(L_k) = \{L_k^{(1)}, R_k \}$, and $\{\de_{\mathcal{G}_m}(R) \cup \de_{\mathcal{G}_m}(L) \} \cap {L}^{({1})} = \emptyset$, implies that a version of the g-formula holds for identifying $p({l}^{({1})})$ under topological orderings where variables in ${R} \cup {L}$ come after variables in ${L}^{({1})}$. That is, under a topological ordering defined on the partially and fully observed variables we have,
\begin{align}
	p({l}^{({1})}) = \left. \frac{ p({l},{r}) }{
		\prod_{r_k \in R} \ p(r_k \mid \past_{{\cal G}_m}(r_k)) }
	\right|_{{r}={1}}, 
	\label{eqn:g-formula-past-om}
\end{align}
We have defined each $\past_{{\cal G}_m}(r_k)$ here as containing only observed or missing (but not hidden) variables. However, though the above g-formula does not contain any hidden variables, it still may not necessarily yield identification, unless additional structure of the model can be exploited.

As an example, consider the model in Fig.~\ref{fig:hidden-vars}. Fix a topological
ordering $U_1, U_2, U_3, L_1^{(1)}, L_3^{(1)}, L_5, L_6,$
$L_2^{(1)}, L_4^{(1)}, R_1, R_2, R_3, R_4,$ $L_1, L_2, L_3, L_4$. Considering the subsequence of this ordering on just the observed variables the g-formula for $p(l_1^{(1)}, l_2^{(1)}, l_3^{(1)}, l_4^{(1)}, l_5, l_6)$ is 
\begin{align*}
	& \left. \frac{ p(l_1, l_2, l_3, l_4, l_5, l_6, r_1, r_2, r_3, r_4) }{
		p(r_1 | \past_{\mathcal{G}_m}\!(r_1)) \! \times  \! p(r_2 | \past_{\mathcal{G}_m}\!(r_2)) \! \times \! p(r_3 | \past_{\mathcal{G}_m}\!(r_3)) \! \times \! p(r_4 | \past_{\mathcal{G}_m}(r_4)) } \!\!\
	\right|_{r=1}. 
\end{align*}

Using the independence relations described in \eqref{eq:markov_pillow_ci}, we have that each $p(r_k \mid \past_{{\cal G}_m}(r_k))$ simplifies as $p(r_k \mid \tb_{{\cal G}_m}(r_k))$. The propensity scores for each missingness indicator then simplifies under the proposed topological order on ${\cal G}_m$ as follows:
{\small
\begin{align*}
	p(r_1 \mid \past_{\mathcal{G}_m}%^\text{om}
	(r_1)) \vert_{r=1} &=
	p(r_1 \mid l_3^{(1)}, l_6) \vert_{r=1} = p(r_1 = 1 \mid	l_6, l_3, r_3 = 1), 
	\\
	p(r_2 \mid \past_{\mathcal{G}_m}%^\text{om}
	(r_2)) \vert_{r=1} &=
	p(r_2 \mid l_3^{(1)}, l_4^{(1)}, l_6, r_1) \vert_{r=1} = p(r_2 = 1 \mid l_3, l_4, l_6, r_1=1, r_3 = 1, r_4=1), \\
	p(r_3 \mid \past_{\mathcal{G}_m}%^\text{om}
	(r_3)) \vert_{r=1} &=
	p(r_3 \mid l_5, l_1^{(1)}, r_1)\vert_{r=1} = p(r_3 \mid	l_5, l_1, r_1=1), 
	\\
	p(r_4 \mid \past_{\mathcal{G}_m}%^\text{om}
	(r_4))\vert_{r=1} &=
	p(r_4 \mid l_1^{(1)}, l_2^{(1)}, l_5, r_3)\vert_{r=1} = p(r_4 \mid l_1, l_2, l_5, r_3=1, r_1=1, r_2=1).
\end{align*}%
}
Since these terms are all functions of the observed data law, $p(l_1^{(1)}, l_2^{(1)}, l_3^{(1)}, l_4^{(1)}, l_5,$ $l_6)$ is identified.

%####################################################
\section{Identification of the Full Law}
\label{sec:missing_data_full_law}
%####################################################
All of our examples so far have focused on identification of the target law, or equivalently the missingness mechanism evaluated at $1$, i.e., $p(R = 1 \mid \pa_{{\cal G}_{{m}}}({r}))$. If identification of the full law is of interest (for instance for model selection purposes as in \cite{gain18missing} and \cite{tu2019causal}), the missingness mechanism  $p( r \mid \pa_{{\cal G}_{{m}}}({r})), \text{ for all } r \in \{0, 1\}^K$ must be identified. It is possible that in certain missing data DAG models, the target law is identified whereas the full law is not. For instance, in the model shown in Fig.~\ref{fig:partial-order}(a) in the main draft, $p({r}_2 \mid R_1 = 0, {l}^{(1)}_1)$ is not identified, and in the model in Fig.~\ref{fig:UAI}(a) in the main body of the draft, $p({r}_1 \mid R_2 = 0, {l}^{(1)}_2)$ is not identified, though the target law is identified in both cases. Both examples have a special colluder structure $L^{(1)}_j \rightarrow R_i \leftarrow R_j$ in common. \cite{bhattacharya19mid} show that the presence of colluders in a graph always implies the full law of the corresponding missing data model is not identified. 

\cite{nabi20completeness} studied  identification of the full law  in missing data DAG models, and provided the first completeness result in a subclass of missing data DAGs where the proxy variables $L$ are childless. They show the missingness mechanism $p({r} \mid {l}^{(1)})$ that is Markov relative to a missing data DAG ${\cal G}_{{m}}$, where $L$s are childless, is identified \emph{if and only if} ${\cal G}_{{m}}$ does not contain self-censoring edges and colluders. The identification is given via an odds ratio parameterization \citep{chen07semiparametric} of the missingness mechanism. An example of identification with odds ratio parameterization is provided in Appendix~\ref{app:odds}.  \cite{nabi20completeness} drew an important connection between missing data models of a DAG ${\cal G}_{{m}}$ that are devoid of self-censoring and colluders, and the itemwise conditionally independent nonresponse (ICIN) model described in \citep{shpitser2016consistent, sadinle16itemwise} {(the ICIN model is referred to as the ``no self-censoring" model in \cite{shpitser2016consistent}).} As a substantive model, the ICIN model implies that no partially observed variable directly determines its own missingness, and is defined by the restrictions that for every pair $L_k^{(1)}, R_k,$ it is the case that $L_k^{(1)} \ci R_k \mid R_{-k}, L_{-k}^{(1)}$. 

The no-self-censoring and no-colluder assumptions imply that $L^{(1)}_i$ is not in the \emph{Markov blanket} of $R_i,$ where the Markov blanket is defined as $\mb_{{\cal G}_{{m}}}(V_i) = \pa_{{\cal G}_{{m}}}(V_i) \cup \ch_{{\cal G}_{{m}}}(V_i) \cup \pa_{{\cal G}_{{m}}}(\ch_{\cal G}(V_i)).$ Given the local Markov property, $V_i \ci V \setminus \mb_{{\cal G}_{{m}}}(V_i) \mid \mb_{{\cal G}_{{m}}}(V_i).$ If the full law is identified, then the target law is guaranteed to be identified. For instance, since there is no self-censoring edges or colluder structures in Figs.~\ref{fig:parallel}(a) and \ref{fig:sequential-parallel}(a), we can immediately conclude that the full law and hence the target law are identified. However, the reverse is not necessarily true -- that is if the full law is not identified (due to presence of colluders or self-censoring edges), the target law might still be identified as discussed in examples related to Figs.~\ref{fig:partial-order}(a) and \ref{fig:UAI}(a). \cite{nabi20completeness} generalized this theory to scenarios where some variables are not just missing, but completely unobserved. They proposed necessary and sufficient graphical conditions that must hold in a missing data DAG model with unmeasured confounders to  permit identification of the full law. They defined a \emph{colluding path} between $L^{(1)}_k$ and $R_k$ as a path where every collider is a variable in $L^{(1)} \cup R$ and every non-collider is a variable in $U.$ They showed that in the absence of such paths, the odds ratio parameterization can be used to identify the full law, while their presence results in non-identification.

Often, instead of identifying the entire full law or target law, we might simply be interested in a simple outcome mean or a causal effect. There are plenty of examples where such parameters are indeed identified, but the underlying joint distribution is not. For instance, consider the graph in Fig.~\ref{fig:hidden-vars}(b), which is  discussed in \cite{mohan2021graphical}. The outcome is missing due to a common unmeasured confounder with pre-treatment variables $X$. The causal effect of $A$ on $Y$ here  is indeed identified, even though the target law is not identified. 
{Briefly, the target law is not identified due to the presence of a colluding path between $Y^{(1)}$ and $R_Y$, which prevents identification of $p(y^{(1)} \mid a, x)$ \citep{mohan2021graphical, nabi20completeness}. However, the model encodes the following independence restrictions which enable identification of the  causal effect: $Y^{(a, R_Y=1)} \ci R_Y$ and $Y^{(a, R_Y=1)} \ci A \mid X, R_Y$,  where $Y^{(a, R_Y=1)}$ denotes the potential outcome when $A$ is set to some value $a$ and had we, in fact, been able to observe it. Such counterfactual independence restrictions are often read using d-separation rules applied to single-world intervention graphs (SWIGs).
%\red{
A detailed discussion on how missing data graphical models which contain counterfactuals relate to SWIGs is left to future work.
}
	The counterfactual distribution $p(y^{(a, R_Y=1)})$ is identified as:
\begin{align*}
	p(y^{(a, R_Y=1)}) &= p(y^{(a, R_Y=1)} \mid R_Y = 1) \\
	&= \sum_x p(y^{(a, R_Y=1)} \mid x, R_Y=1) \times p(x \mid R_Y=1) \\
	&=  \sum_x p(y^{(a, R_Y=1)} \mid x, A=a,  R_Y=1) \times p(x \mid R_Y=1) \\
	&= \sum_x p(y \mid x, A=a, R_Y=1)\times p(x \mid R_Y=1).
\end{align*}
The first equality follows from $Y^{(a, R_Y=1)} \ci R_Y$, the second from rules of probability, the third from $Y^{(a, R_Y=1)} \ci A \mid X, R_Y$, and the final equality follows from consistency.

%######################################
\section{Additional Results and Proofs}
%######################################
\label{app:proofs}

\subsection{Proposition \ref{prop:miss-id}}

\begin{proof}
	Since $\de_{\mathcal{G}}(R_i) \cap {L}^{({1})}$, the vertex set ${L}^{({1})}$ is ancestral in $\mathcal{G}_m$
	This implies $p({l}^{({1})})$ is equal to
	{\small
	\begin{align*}
		\sum_{{r} \cup {l}} \ p({l},{r},{l}^{({1})}) = \sum_{{r} \cup {l}} \big( \prod_{v_k \in {r} \cup 
			{l}} p(v_k \mid \pa_{\mathcal{G}_m}(v_k)) \big) \times \big( \prod_{v_k \in 
			{l}^{({1})}} p(v_k \mid \pa_{\mathcal{G}_m}(v_k)) \big) 
		= \prod_{v_k \in {l}^{({1})}} p(v_k \mid \pa_{\mathcal{G}_m}(v_k)).
	\end{align*}%
	}
	Further, using Bayes rule, we conclude the second equality in (\ref{eq:mDAG-bayes}) by noting that $p({l},{r},{l}^{({1})})\vert_{{r}={1}} = p({l},{r})\vert_{{r}={1}}$ (by consistency) and 
	{\small
	\begin{align*}
		p(r, l \mid l^{(1)})\vert_{r=1} = \prod_{l_k \in l} p(l_k \mid r_k = 1, l^{(1)}_k) \times \prod_{r_k \in r} p(r_k \mid \pa_{{\cal G}_m}(r_k))\vert_{{r}={1}} = \prod_{r_k \in r} p(r_k \mid \pa_{{\cal G}_m}(r_k))\vert_{{r}={1}}.
	\end{align*}%
	}
\end{proof}

%+++++++++++++++++++++++++++

\subsection{Lemma \ref{lem:invariance}} 

\begin{proof}
	Given a set $R^* \subseteq R$ and the corresponding set of counterfactuals $L^{*(1)}$, the distribution
	\begin{align*}
		p(l^{(1)}\setminus l^{*(1)}, r \setminus r^*, l \ \Vert\  r^*=1) \coloneqq \frac{p(l^{(1)}\setminus l^{*(1)}, r, l)}{\prod_{r_k \in r^*} p(r_k \mid \pa_{{\cal G}_m}(r_k))}\Bigg\vert_{R^*=1}
	\end{align*}%
	factorizes with respect to a \emph{conditional DAG} (CDAG) $\widetilde{{\cal G}}_m(L^{(1)} \cup \{R\setminus R^*\} \cup L, R^*)$, which is a DAG containing random vertices $L^{(1)} \cup \{R\setminus R^*\} \cup L$ and fixed vertices $R^*$ with the property that all fixed vertices can only have outgoing directed edges. $\widetilde{{\cal G}}_m$ is constructed from original m-DAG ${\cal G}_m$ by removing all edges with arrowheads into $R^*,$ marking $R^*$ as fixed vertices, and treating each $L_k^{(1)} \in L^{*(1)}$ as equivalent to its corresponding proxy (by consistency). The CDAG factorization of any $p(v \ \Vert\ w)$ with respect to a CDAG ${\cal G}(V, W)$ is a straightforward generalization of the DAG factorization:
	\begin{align*}
		p(v \ \Vert\ w) = \prod_{v_i \in v} p(v_i \mid \pa_{\cal G}(v_i) \setminus w \ \Vert\ \pa_{\cal G}(v_i) \cap w),
	\end{align*}%
	where conditioning in $p(v \ \Vert\  w)$ is defined as in \eqref{eqn:kernel}.
	
	If $p(v \ \Vert\ w)$ factorizes with respect to ${\cal G}(V, W)$, it obeys the local Markov property which states that for each variable $V_i$, the distribution $p(v_i \mid \past_{\cal G}(v_i) \setminus w \ \Vert\  w \cap \pa_{\cal G}(v_i))$ is only a function of $V_i$ and its direct causes $\pa_{\cal G}(V_i).$ This immediately implies the conclusion, since the kernel $p(l^{(1)}\setminus l^{*(1)}, r \setminus r^*, l \ \Vert\  r^*=1)$ factorizes according to the CDAG $\widetilde{{\cal G}}_m(L^{(1)} \cup \{R\setminus R^*\} \cup L, R^*)$ where all direct causes of each $R_k \not\in R^*$ are preserved. See \cite{richardson2023nested} for more details on conditional DAG factorization.
	
	%	The distribution
	%	\[
	%		q((R \setminus R^*) \cup L^{(1)} \cup L \mid R^*) =
	%		\frac{
	%			p(R \cup L^{(1)} \cup  L)
	%		}{
	%			\prod_{R_j \in R^*} p(R_j \mid \pa_{\cal G}(R_j))
	%		}
	%	\]
	%	factorizes with respect to a \emph{conditional DAG (CDAG)} $\tilde{\cal G}((R \setminus R^*) \cup L^{(1)} \cup L,  R^*)$ which is a DAG containing random vertices $(R \setminus R^*) \cup L^{(1)} \cup L$
	%	and fixed vertices $R^*$ with the property that all fixed vertices can only have outgoing directed edges.  This DAG is obtained from the original mDAG by dropping all edges into $R^*$, and marking $R^*$
	%	as fixed vertices.  The CDAG factorization of any $q(V \mid W)$ with respect to $\tilde{\cal G}(V,W)$ is a straightforward generalization of the DAG factorization:
	%	\[
	%	q(V \mid W) = \prod_{V_i \in V} q(V_i \mid \pa_{\cal G}(V_i)),
	%	\]
	%	where conditioning in $q(V \mid W)$ is defined in the usual way:
	%	\[
	%	q(V_i \mid \pa_{\cal G}(V_i)) \coloneqq
	%	\frac{
	%	\sum_{V \setminus (\{ V_i \} \cup \pa_{\tilde{\cal G}}(V_i))} q(V \mid W)
	%	}{
	%	\sum_{V \setminus \pa_{\tilde{\cal G}}(V_i)} q(V \mid W)	
	%	}.
	%	\]
	%	If $q(V \mid W)$ factorizes with respect to $\tilde{\cal G}(V,W)$, it obeys the local Markov property which states that
	%	for each variable $V_i$, the distribution $q(V \setminus (\an_{\tilde{G}}(V_i) \setminus \{ V_i \}) \mid W)$ is only a function of $\{ V_i \} \cup \pa_{\tilde{\cal G}}(V_i)$.
	%	This immediately implies the conclusion.  See \cite{richardson2023nested} for more details on conditional DAG factorization.
\end{proof}

%+++++++++++++++++++++++++++

\subsection{Identification under rank preservation in a $K$ variable block-parallel model}

In Section~\ref{sec:diss} we saw how missing data identification strategies can be applied in conjunction with additional assumptions, such as rank preservation, to attain identification in a two variable causal analogue of the block-parallel missing data model. The following theorem shows how this applies to any causal model endowed with rank preservation that is analogous to a $K$ variable block-parallel model (as well as any sub models of it). For simplicity, we will assume all treatment variables are binary though the result trivially extends to non-binary treatments. 

\begin{theorem} %[Identification under rank preservation and block-parallel assumptions]
	Given a causal model that encodes the following independence restrictions: for each $k \in \{1, \dots, K\}$
	\begin{align*}
		A_k \ci L_k^{(0)}, L_k^{(1)}, A_{-k} \mid L_{-k}^{(0)}, L_{-k}^{(1)},
	\end{align*}%
	and the following rank preservation assumptions: for each $k \in \{1, \dots, K\}$ and $j = \{0, 1\}$ there exists a bijection $g_k$ such that $L_k^{(1-j)} = g_k(L_k^{(j)})$. The counterfactual distribution $p(l_1^{(a_1)}, \dots l_2^{(a_K)})$, where each $a_k \in \{0, 1\}$, is identified and given by the following functional:
	\begin{align*}
		\frac{p(l_1, \dots, l_K, A_1=a_1, \dots, A_K = a_K)}{\prod_{A_k \in A} p(A_k = a_k \mid l_{-k}, A_{-k}=a_{-k})}.
	\end{align*}%
\end{theorem}

\begin{proof}
	The counterfactual distribution $p(l_1^{(a_1)}, \dots, l_k^{(a_K)})$ is identified via the following identities following a very similar strategy to the one used in the main text. In the following we will use $l^{(a)}$ and $l^{(1-a)}$ as short hand for $l_1^{(a_1)}, \dots, l_K^{(a_K)}$ and $l_1^{(1-a_1)}, \dots, l_K^{(1-a_K)}$ respectively:
		\begin{align*}
			p(l^{(a)}) &=
			\frac{
				p({l}^{(a_1)}_1, \dots, {l}^{(a_K)}_K, A_1 = a_1, \dots, A_K = a_K)
			}{
				p(A_1 = a_1, \dots, A_K = a_K \mid {l}^{(a_1)}_1, \dots,  {l}^{(a_K)}_K)
			} \\
		&= 	\frac{
				p({l}^{(a)}, A_1 = a_1, \dots, A_K = a_K)
			}{
				p(A_1 = a_1, \dots, A_K = a_K \mid {l}^{(a)})
			} \\
			&=
			\frac{
				p(l^{(a)}, A_1 = a_1, \dots, A_K = a_K)
			}{
				\displaystyle{\sum_{l^{(1-a)} }} p(A_1 = a_1, \dots, A_K = a_K \mid  l^{(a)}, l^{(1-a)}) \times  p( {l}^{(1-a)} \mid  l^{(a)})
			}\\
			&=
			\frac{
				p( l^{(a)}, A_1 = a_1, \dots, A_K = a_K)
			}{
				\displaystyle{\sum_{ l^{(1-a)}}} \ \displaystyle{ \prod_{A_k \in A}} p(A_k=a_k \mid l^{(a)}, l^{(1-a)}, A_{-k}= a_{-k}) \times p( {l}^{(1-a)} \mid  {l}^{(a)})
			}\\
			&=\!\!
%			\frac{
%				p( l^{(a)}, A_1 = a_1, \dots, A_K = a_K)
%			}{
%				\displaystyle{\sum_{ l^{(1-a)}}} \ \displaystyle{\prod_{A_k \in A}} p(A_k \! = \! a_k | l^{(a)}, l^{(1-a)}, A_{-k} \! = \! a_{-k})\! \times \! p( {l}^{(1-a)} |  {l}^{(a)}) \! \times \! \mathbb{I}( l_1^{(1-a_1)} \! = \! g_1(l_1^{(a_1)}), \dots, l_K^{(1-a_K)} \!= \! g_K(l_K^{(a_K)}))
%			}\!\!\!\! \\
				\Bigg\{
				\displaystyle{\sum_{ l^{(1-a)}}} \ \displaystyle{\prod_{A_k \in A}} p(A_k \! = \! a_k | l^{(a)}, l^{(1-a)}, A_{-k} \! = \! a_{-k})\! \times \! p( {l}^{(1-a)} |  {l}^{(a)}) \! \\ 
				&\hspace{1cm} \times \! \mathbb{I}( l_1^{(1-a_1)} \! = \! g_1(l_1^{(a_1)}), \dots, l_K^{(1-a_K)} \!= \! g_K(l_K^{(a_K)}))
			\Bigg\}^{-1} \\ 
			&\hspace{2cm} \times p( l^{(a)}, A_1 = a_1, \dots, A_K = a_K) \\
			&=
			\frac{
				p( l^{(a)}, A_1 = a_1, \dots, A_K = a_K)
			}{
				\displaystyle{\prod_{A_k \in A} } p(A_k=a_k \mid l_{-k}^{(a)}, A_{-k}= a_{-k})
			}\\
			&=
			\frac{p(l_1, \dots, l_K, A_1=a_1, \dots, A_K = a_K)}{ \displaystyle{ \prod_{A_k \in A}} p(A_k = a_k \mid l_{-k}, A_{-k}=a_{-k})}.
		\end{align*}
	The first equality follows from Bayes rule, the second from our notational convention, the third from rules of probability, the fourth from applying the chain rule of factorization and using the independence restrictions implied by the model, the fifth and sixth from rank preservation as well as restrictions encoded by the model, and finally, the last equality follows from consistency.
\end{proof}

%%%%%%%%%%%%%%%%%%%%%%%%%%%%%%%%%%%%%%%%%%%%%%%%%%%%%%%%%%%%%%%%%%%%%%%%%%%%%%%%%%%%%%%%%%

%%\iffalse
%\bibhang=1.7pc
%\bibsep=2pt
%\fontsize{9}{14pt plus.8pt minus .6pt}\selectfont
%\renewcommand\bibname{\large \bf References}
%%\begin{thebibliography}{11}
%\expandafter\ifx\csname
%natexlab\endcsname\relax\def\natexlab#1{#1}\fi
%\expandafter\ifx\csname url\endcsname\relax
%\def\url#1{\texttt{#1}}\fi
%\expandafter\ifx\csname urlprefix\endcsname\relax\def\urlprefix{URL}\fi
%\fi

% use bibfile 
{\small
\bibliographystyle{chicago}      % Chicago style, author-year citations
\bibliography{references}   % name your BibTeX data base
}

\makeatletter\@input{link.tex}\makeatother